\documentclass[12pt,a4paper]{article}
\usepackage[margin=2cm]{geometry}
\usepackage{authblk}

\usepackage[breaklinks]{hyperref}
\usepackage{booktabs}
\usepackage{txfonts}

\usepackage{braket}

\usepackage{amsthm}
\usepackage{amsmath}
\usepackage{amssymb}
\usepackage{amssymb,mathrsfs}
\usepackage{xcolor}
\usepackage{makecell}
\usepackage{tensor}
\usepackage{graphicx}
\usepackage{caption}
\usepackage{subcaption}
\usepackage{pifont}
\usepackage{bbold}

\newcommand{\cmark}{\ding{51}}
\newcommand{\xmark}{\ding{55}}

\newcommand{\nord}[1]{{:}#1{:}}

\newcommand{\supp}{\mathrm{supp}}

\newcommand{\be}{\begin{equation}}
	\newcommand{\ee}{\end{equation}}
\newcommand{\bea}{\begin{eqnarray}}
	\newcommand{\eea}{\end{eqnarray}}

\newtheorem{definition}{Definition}[section]

\newtheorem{proposition}[definition]{Proposition}
\newtheorem{theorem}[definition]{Theorem}
\newtheorem{theoremtemplate}[definition]{Theorem Template}

\newtheorem{conjecture}[definition]{Conjecture}
\newenvironment{proof*}{\smallskip\par\noindent\emph{Proof. }
	\ignorespaces}{\hfill$\Box$\smallskip\par\ignorespaces}
\newenvironment{proofsketch*}{\smallskip\par\noindent
	\emph{Sketch of proof. }\ignorespaces}
{\hfill$\oslash$\smallskip\par\ignorespaces}

\begin{document}

\title{Energy conditions in general relativity and quantum field theory}

\author{Eleni-Alexandra Kontou\thanks{\tt elenikontou@cosmos.phy.tufts.edu, ekontou@holycross.edu} }
\affil{Department of Mathematics, University of York, Heslington, York YO10 5DD, United Kingdom}
\affil{Department of Physics, College of the Holy Cross, Worcester, Massachusetts 01610, USA}
\author{Ko Sanders \thanks{\tt jacobus.sanders@dcu.ie}}
\affil{Dublin City University, School of Mathematical Sciences and Centre for Astrophysics and Relativity, Glasnevin, Dublin 9, Ireland}
\date{\small\today}
\maketitle

\begin{abstract}
	This review summarizes the current status of the energy conditions in general relativity and quantum field theory. We provide a historical review and a summary of technical results and applications, complemented with a few new derivations and discussions. We pay special attention to the role of the equations of motion and to the relation between classical and quantum theories.
	
	Pointwise energy conditions were first introduced as physically reasonable restrictions on matter in the context of general relativity. They aim to express e.g.~the positivity of mass or the attractiveness of gravity. Perhaps more importantly, they have been used as assumptions in mathematical relativity to prove singularity theorems and the non-existence of wormholes and similar exotic phenomena. However, the delicate balance between conceptual simplicity, general validity and strong results has faced serious challenges, because all pointwise energy conditions are systematically violated by quantum fields and also by some rather simple classical fields. In response to these challenges, weaker statements were introduced, such as quantum energy inequalities and averaged energy conditions. These have a larger range of validity and may still suffice to prove at least some of the earlier results. One of these conditions, the achronal averaged null energy condition, has recently received increased attention. It is expected to be a universal property of the dynamics of all gravitating physical matter, even in the context of semiclassical or quantum gravity.
\end{abstract}

\pagebreak

\tableofcontents

\pagebreak

\section{Introduction}

In general relativity, according to the famous quote by Wheeler, ``Space tells matter how to move. Matter tells space how to curve" (\cite{MTW} Sec.~1.1). This mutual influence is encoded in the combination of Einstein's equation and the dynamical equations of the matter that is present. If one focuses solely on Einstein's equation, however, without imposing any restrictions on the matter, any Lorentzian metric field on any manifold can be regarded as a solution. This may lead to surprising phenomena, such as wormholes, superluminal travel and closed timelike curves or other causality violations. The fact that these phenomena have never been observed then requires an explanation. The most common explanation is that such spacetimes typically require the presence of matter fields with exotic properties, such as negative energy densities.

In the context of these exotic properties an important role is played by \emph{energy conditions}. These are certain pointwise restrictions on the stress-energy-momentum tensor of the matter (or stress tensor for short), whose purpose is three-fold. Firstly, because Einstein's equation involves no other properties of matter than its stress tensor, energy conditions allow us to analyse the behaviour of gravitating systems without the need to specify the detailed behaviour of the matter. This method of bypassing a complicated, detailed analysis was the key step that allowed Penrose and Hawking to prove their singularity theorems
\cite{Penrose_prl:1965,Hawking:1966sx}. Secondly, energy conditions aim to capture common features of many different kinds of matter, thereby encoding an idea of ``normal matter'' that should be valid quite generally. Thirdly, the energy conditions aim for a conceptually simple characterization. The positivity of the energy density, e.g., may be related to the stability of the system, at least in the naive sense that systems in classical mechanics are stable when the energy is bounded from below.

The balance between these three purposes has been a persistent cause of tension. One would like to have an energy condition that is weak enough to be valid for as many types of matter as possible, ideally including all observed kinds of matter and possibly also other ``physically reasonable'' matter. At the same time, the condition should be strong enough to have physically interesting consequences in the context of general relativity, e.g.~singularity theorems, the Area Theorem and Black Hole Topology Theorem, the Chronology Protection Conjecture or others. Balancing these two purposes is difficult enough by itself, but it gets even more complicated by the desire to maintain conceptual simplicity without getting lost in technical, mathematical conditions.

As a result of this tension, a variety of energy conditions exist, each with their own notable strengths and weaknesses regarding their range of validity, important consequences and interpretation. The lack of a single preferred energy condition is an important point of criticism \cite{Barcelo:2002bv,Curiel:2017}, especially when the desire to prove strong results leads physicists to invoke energy conditions that appear ad hoc and tailored towards specific proofs. Historically, the first energy conditions were arguably introduced in such an ad hoc way \cite{Penrose_prl:1965,Hawking:1966sx}. Even stronger criticisms are raised against the supposed general validity of energy conditions \cite{Visser:1999de,Barcelo:2002bv}. The four pointwise energy conditions that are most widely used all admit counter-examples in the form of rather innocent looking classical scalar fields with non-minimal scalar curvature coupling. Although some counter-examples in the literature can be dismissed as unphysical, e.g.~because the field configurations involved do not satisfy Einstein's equation, the problems do persist also for quite reasonable looking situations. Moreover, all pointwise energy conditions are necessarily violated by essentially any quantum field theory \cite{Epstein:1965zza}. Finally, a criticism that has been raised at the conceptual level is the lack of a convincing derivation of energy conditions from deeper principles \cite{Curiel:2017}. E.g., energy conditions seem by no means a necessary condition for a classical field theory to have a well-posed initial value formulation, even though they may enter as a useful method in the proof of such a property \cite{Ringstrom}. Attempted derivations from the Raychaudhuri equation are not convincing and we refer to Sec.~2.3 of \cite{Curiel:2017} for a detailed criticism. A more plausible idea is to derive the energy conditions from global stability properties, but a recent attempt to do so requires the addition of an ``improvement term'' which seems unphysical and not covariant \cite{PhysRevLett.118.151601}.\footnote{For a massive minimally coupled scalar field in two-dimensional Minkowski space the procedure of \cite{PhysRevLett.118.151601} yields a non-zero improvement term, although none is needed (cf.~Sec.~\ref{sub:examples} below). For non-minimally coupled scalar fields the improvement also removes the effect of the non-minimal coupling from the energy density. In the massive case there seems to be no covariant ``improved stress tensor'' which corresponds to these improved energy densities. In the massless case there is, but for non-minimal coupling it no longer represents the physical stress, energy and momentum of the field and the corresponding tensor in curved spacetimes is no longer conserved.}

\medskip

Although quantum fields must violate all pointwise energy conditions, they do often satisfy \emph{quantum energy inequalities} \footnote{The term ``quantum inequalities" is also used in the literature, but it can apply to other quantities than energy densities. In this review we use the more accurate term ``quantum energy inequalities"}. These are lower bounds on a weighted average of components of the expectation value of the renormalized stress tensor \cite{Ford:1978qya,Fewster:2012yh}. The average is usually taken along a timelike curve, which suffices to obtain a finite, but possibly negative, lower bound. The validity of a QEI shows that violations of some of the classical energy conditions must be restricted in duration and magnitude. For this reason they are widely regarded as an appropriate rebuttal to those critics who invoke the quantum violations to call the usefulness of the classical energy conditions into question.

We will consider QEIs in the light of the same three purposes as for the pointwise energy conditions. The bounds that QEIs impose on the (expected, renormalized) stress tensor are weaker than the pointwise energy conditions. Although the derivation of e.g.~a singularity theorem from a QEI has yet to be achieved, there is good evidence that they will suffice to prove interesting geometric consequences, see e.g.~\cite{Fewster:2010gm}. However, QEIs are not formulated as a single general condition, but rather as lower bounds that depend on the theory under consideration. This may be one reason why the term ``quantum energy conditions'', instead of QEIs, is not in common use in the literature. Another reason for this discrepancy in terminology appears to be that the early literature on QEIs focused on their role as results that are to be proved, unlike the pointwise energy conditions, which first appeared as assumptions used to prove singularity theorems. In static spacetimes the validity of a QEI is closely related to the stability and thermodynamics of the system in question \cite{FewsterVerch2003} and we note that the presence of a lower bound for the energy density seems a more natural stability condition than insisting on the lower bound zero, which occurs in the pointwise energy conditions. This provides a conceptually clear motivation to consider QEIs. However, we note that the lower bound need not converge to zero in the long time limit, with the Casimir effect providing a particularly interesting counter-example \cite{Casimir:1948dh,PhysRevD.51.4277,Fewster:2006uf}. The relation between QEIs and time-energy uncertainty relations, which is often alluded to in the literature, is potentially misleading and we stress that the derivation of QEIs neither assumes nor proves any time-energy uncertainty relation.

QEIs are difficult to formulate rigorously for interacting quantum fields in general curved spacetimes, due to the difficulties in defining a renormalized stress tensor. The study of free quantum fields, however, is not problematic. In general, when a classical free field satisfies a certain energy condition, one expects that the corresponding quantized field satisfies a corresponding QEI. This relationship seems to be borne out in the literature on the topic, using case by case investigations. To complement this work we will provide below an elementary argument in the opposite direction (see Sec.~\ref{subs:claslimits} for details): when a quantum scalar field in a general curved spacetime satisfies a QEI for a suitably large class of states (e.g.~Hadamard states), then the corresponding classical field must satisfy the corresponding classical energy condition. A similar argument for fields in Minkowski space can be found in \cite{Fewster:2007ec,Fewster:2018pey}. Without aiming for the highest level of generality, these results suggest that classical energy conditions can be derived from QEIs. This also provides one possible explanation why the lower bound in classical energy conditions is zero and not negative (for an earlier, entirely classical explanation see \cite{Tipler:1978zz}).

\medskip

\emph{Averaged energy conditions} take an intermediate position between the pointwise energy conditions and QEIs. Like the latter, they average components of the stress tensor along suitable causal curves, but they do insist on a lower bound zero, just like the former. In this way they try to combine the best of both sides, leading to clearly formulated conditions which are weaker than pointwise conditions, but which may still be strong enough to prove interesting consequences. The condition which comes closest to being generally valid is the \emph{achronal averaged null energy condition} (AANEC), see Sec.~\ref{subs:AANEC} for details. It requires that the average along a complete, achronal null geodesic of the stress tensor, projected onto the geodesic, is non-negative. This condition is weaker than all pointwise energy conditions and there are good reasons to believe that it is also satisfied by quantum fields under reasonable circumstances. Indeed, all known counter-examples in Minkowski space, and most counter-examples in curved spacetimes, seem to have properties that cast doubt on their physical feasibility, such as: (i) one or more of the equations of motion, including Einstein's equation, is violated, (ii) there is an effective Newton's constant which has the ``wrong'' sign, or (iii) the violation has a transversal size of the order of the Planck length, where the validity of the theory is questionable.

The interpretation of the AANEC is that the total energy flux through an achronal null geodesic should be non-negative
\cite{Wald:1991xn}. The wide range of validity makes this condition a very interesting subject to study. In two-dimensional Minkowski space there is a general proof of the AANEC for all quantum fields with a mass gap \cite{Verch:1999nt} and the recent literature in higher dimensional Minkowski space has suggested relationships to several other topics of current interest, such as black hole entropy \cite{Wall:2009wi} and a quantum null energy condition \cite{Ceyhan:2018zfg}. If the AANEC is indeed ``universally'' valid (in some appropriate sense), also in curved spacetimes in the context of semiclassical gravity, then one would hope that there is a deeper reason for this, which should be provided by an underlying quantum theory of gravity. In this sense the AANEC can potentially open a window onto the quantum properties of gravity.

In this regard it may be noteworthy that the pointwise and averaged energy conditions are independent of the time-orientation of spacetime. The same is true for QEIs, even though the time-orientation enters there also through the sign in the commutation relations (expressed e.g.~in the direction of the singularities in two-point distributions) and not only through the direction of the causal curve over which we average. This independence of the time-orientation suggests that these conditions and inequalities are related to microscopic physics, unlike, e.g., thermodynamic entropy. In this context the recent derivation of the AANEC from the generalized second law is paradoxical at least \cite{Wall:2009wi}.\footnote{We thank Erik Curiel for pointing this paradox out to us.}

\medskip

The potential to shed light on quantum gravity is a good motivation that will drive future research on the energy conditions and their quantum counterparts. Particular questions of interest include the general validity, or otherwise, of the AANEC in higher dimensional Minkowski space and in curved spacetimes, the physical implications of the AANEC and other weak energy conditions in the context of singularity theorems and black hole thermodynamics, and the derivation of classical energy conditions from their quantum counterparts for (perturbatively) interacting quantum fields.

\medskip

We will begin our review below with a discussion of the pointwise energy conditions in Section \ref{sec:pointwise}. In addition to a discussion of examples and counter-examples, this will also include in Section \ref{subs:onshell} a discussion of the role that Einstein's equation and other equations of motion play in the validity of such conditions --- a topic that is sometimes disregarded. 
In Section \ref{sec:QEI} we explain why all pointwise energy conditions are violated by quantum fields and we introduce and categorize QEIs. After providing some examples we comment on their interpretation in
Section \ref{subs:uncertainty}. The relation between classical energy conditions and QEIs is the topic of Section \ref{sec:classvsquant}. Besides a discussion of quantization and classical limits, this section discusses averaged energy conditions and derivations of in particular the AANEC in quantum field theory. 
Various consequences of the energy conditions or QEIs are discussed in Section \ref{sec:applications}, although several of these applications would warrant an extensive review in their own right. Finally, in Section \ref{sec:outlook} we 
draw some conclusions and give an outlook on the future of the energy conditions in quantum field theory and gravity.

For further reading on the energy conditions we recommend the review of Curiel \cite{Curiel:2017} and the paper of Barcel\'o and Visser \cite{Barcelo:2002bv}, which, however, takes a less optimistic view than we do. Fewster has written a nice review of QEIs \cite{Fewster:2012yh}. The singularity theorems were extensively reviewed by Senovilla \cite{Senovilla:2014} and Witten wrote a nice introduction to the key concepts of these theorems \cite{Witten:2019qhl}. The connection between the energy conditions and wormholes, time-machines and other exotic phenomena is explained in the collection edited by Lobo
\cite{lobo2017wormholes}, in the book by Visser \cite{visser1996lorentzian} and in the short review paper by Hiscock \cite{hiscock2002wormholes}. Books by Thorne \cite{thorne1995black} and by Everett and Roman \cite{everett2012time} cover similar topics in an accessible way for a general audience. As standard references on general relativity we suggest \cite{wald1984general,MTW,HawkingEllis:1973,weinberg1972gravitation,carroll2019spacetime} and some less standard topics can be found in \cite{ONeill,penrose1984spinors,tolman1987relativity}.

\section{Pointwise energy conditions}
\label{sec:pointwise}

Throughout this review we will use Planck units, so in particular $c=G=\hbar=1$.
We will consider classical Lagrangian field theories on an orientable $n$-dimensional manifold $\mathcal{M}$, where one of the fields is the Lorentzian metric $g_{ab}$. We call the pair $M=(\mathcal{M},g_{ab})$ a spacetime, where we use the $+++$ sign convention of \cite{MTW} and the abstract index notation conventions of \cite{wald1984general,penrose1984spinors}. All Cauchy surfaces are assumed smooth and spacelike unless stated otherwise. Fields other than $g_{ab}$ will be denoted generically by $\Phi$ and the theory is determined by an action of the form 
\bea
S[g,\Phi]&=&S_{EH}[g]-S_M[g,\Phi]\,,\nonumber\\
S_{EH}[g]&=&\frac{1}{16\pi}\int_{\mathcal{M}}\ \mathrm{d}vol_g \, R\,,\label{eqn:action}\\
S_M[g,\Phi]&=&\int_{\mathcal{M}}\  \mathrm{d}vol_g\, \mathcal{L}_{\Phi}[g,\Phi] \,.\nonumber
\eea
Here $S_{EH}$ is the Einstein-Hilbert action for general relativity, where $R$ is the Ricci curvature scalar and $\mathrm{d}vol_g$ is the metric volume form. For simplicity we will not include a cosmological constant or modifications to general relativity. The action $S_M$ for the matter fields involves a Lagrangian density $\mathcal{L}_{\Phi}$ and we will often focus on specific examples, such as scalar fields, the vector potential of electromagnetism or Dirac fields.

Varying w.r.t.~the metric $g^{ab}$ we find Einstein's equation of motion,
\be
\label{eqn:EE}
G_{ab}=8\pi T_{ab}\,,
\ee
where $G_{ab} :=R_{ab}-\frac12 g_{ab} R$ is the Einstein tensor, $R_{ab}$ the Ricci tensor, and
\be
\label{eqn:stresstensor}
T_{ab}:=\frac{2}{\sqrt{|\det g|}}\frac{\delta}{\delta g^{ab}}S_M[g,\Phi]
\ee
is the stress-energy-momentum tensor, or \emph{stress tensor} for short. Varying w.r.t.~the fields $\Phi$ yields additional equations of motion that couple the fields $\Phi$ to the metric $g_{ab}$.

Phenomenological descriptions in cosmology often directly prescribe a form for the stress tensor, rather than a Lagrangian density. One simple and useful example is that of a perfect fluid, whose stress tensor is of the form
\be
\label{eqn:Tperfectfluid}
T^{ab}=(\rho+P) v^a v^b +P g^{ab} \,,
\ee
where $\rho$ is the energy density, $P$ the pressure and $v^a$ is the fluid's four-velocity vector field, satisfying $v^av_a=-1$. In this case no separate equation of motion is specified for the fluid, but Einstein's equation and the second Bianchi identity do entail that the stress tensor is divergence free, $\nabla^aT_{ab}=\frac{1}{8\pi}\nabla^aG_{ab}=0$. (A more general Ansatz for the stress tensor is the Type I form in the Segre classification, cf.~\cite{HawkingEllis:1973}.
We refer to \cite{MaedaMartinez2018} for a discussion of energy conditions using the Segre classification in arbitrary dimensions.)

One of the simplest examples of a Lagrangian field theory, which we will use repeatedly to illustrate various energy conditions, is the real linear scalar field $\phi$ with Lagrangian density 
\be
\label{eqn:lagrangian}
\mathcal{L}_{\mathrm{lsc}}[\phi]=\frac{1}{2} [(\nabla \phi)^2+(m^2+\xi R)\phi^2 ] \,,
\ee
where $m\ge0$ is the field mass and $\xi\in\mathbb{R}$ the scalar curvature coupling constant. The field equation is the Klein-Gordon equation
\be \label{eqn:field}
(-\Box_g+m^2+\xi R)\phi=0 \,.
\ee
The stress tensor is
\be
\label{eqn:tmunu}
T_{ab}=(\nabla_a \phi)(\nabla_b \phi)-\frac{1}{2} g_{ab} ((\nabla \phi)^2+m^2 \phi^2)+\xi(g_{ab} \Box_g-\nabla_a \nabla_b+G_{ab}) \phi^2
\ee
and its trace $T=T\indices{^a_a}$ is given by
\be
\label{eqn:trace}
T = \left(1-\frac{n}{2}\right) (\nabla\phi)^2 - \frac{n}{2}m^2\phi^2 + \xi\left((n-1)\Box_g +  \left(1-\frac{n}{2}\right)R \right)\phi^2 \,.
\ee
Observe that for flat spacetimes, $R=0$, the stress tensor does not reduce to that of minimal coupling, $\xi=0$, unlike the
Lagrangian density and the field equation. Scalar fields are convenient and well-studied models, but their physical applicability is limited to the Higgs field, effective descriptions of bound states of more complicated fields or components of vector-valued fields.

Our second example is the Proca field $A_a$, the classical uncharged spin-$1$ field of mass $m > 0$. Its Lagrangian density is
\be
\mathcal{L}_{\mathrm{Pr}}=\frac{1}{4}F^{ab}F_{ab}+\frac{1}{2}m^2 A^a A_a \,,
\ee
where the Proca two-form is defined as
\be
\label{eqn:fieldstrengthproca}
F_{ab}=2\nabla_{[a}A_{b]} = \nabla_aA_b-\nabla_bA_a\,.
\ee
The field equation is
\be
\label{eqn:fieldproca}
-\nabla^a F_{ab}+m^2 A_b=0 \,.
\ee
The Proca equation (\ref{eqn:fieldproca}) is equivalent to the wave equation
\be
\label{eqn:fieldproca2}
-\Box_gA_a+R\indices{_{a}^{b}}A_b+m^2A_a=0
\ee
with $-\Box_g=\nabla^a\nabla_a$, together with the Lorenz condition
\be
\label{eqn:Lorenzcondition}
\nabla_a A^a=0 \,.
\ee

Closely related is the electromagnetic field, described by equivalence classes of potentials $A_a$ obeying Maxwell's equations
\be
\nabla^a F_{ab}=0 \,
\ee
where we note that $\nabla_{[a}F_{bc]}=0$ follows trivially from (\ref{eqn:fieldstrengthproca}). The equivalence relation is given by the gauge transformations $A_a\to A_a+\nabla_a\chi$, where $\chi$ is any real scalar function. Gauge equivalent potentials correspond to the same physical configuration and one may exploit the gauge freedom to choose a representative $A_a$ satisfying the Lorenz gauge condition (\ref{eqn:Lorenzcondition}). In that case the Maxwell equations are identical to (\ref{eqn:fieldproca2}) with $m=0$.

The stress tensor for both the Maxwell and Proca fields is
\be
\label{eqn:tensorproca}
T_{ab}=F_{ac}F\indices{_{b}^{c}}-\frac{1}{4} g_{ab} F^{cd}F_{cd}+m^2 \left(A_a A_b-\frac{1}{2} g_{ab} A^c A_c\right) \,,
\ee
where the Maxwell case corresponds to the $m=0$ case. 

For further examples of Lagrangian field theories and their stress tensors we refer to the literature \cite{MTW,wald1984general}.

\medskip

To understand the geometric interpretation of energy conditions it is important to introduce Raychaudhuri's equation. This equation describes the evolution of a congruence of timelike or null geodesics and it is derived from the trace of the geodesic deviation equation (for a complete derivation see \cite{wald1984general} Ch.~9.2). Its introduction in 1955 \cite{Raychaudhuri:1953yv} was a leap forward to model-independent singularity theorems (see Sec.~\ref{sub:singularity} for more details). We will use the more general form of the Raychaudhuri equation introduced in \cite{Ehlers:1993}. For a timelike congruence with velocity field $t^a$, parametrized by proper time $\tau$, the expansion $\theta=\nabla_a t^a$ satisfies
\be
\label{eqn:raytime}
\frac{\mathrm{d} \theta}{\mathrm{d}\tau} =-R_{ab} t^a t^b-\sigma_{ab} \sigma^{ab} +\omega_{ab} \omega^{ab}-\frac{\theta^2}{n-1}  \,,
\ee
where $n$ is the spacetime dimension, $\sigma_{ab}$ is the shear tensor and $\omega$ the vorticity or rotation (or twist) tensor. The shear and vorticity tensors are related to the deformation and rotation of a volume element along the curves of the geodesic congruence (cf.~\cite{wald1984general} Ch. 9).

For a null congruence with tangent field $k^a$, parametrized by affine parameter $\lambda$, the expansion $\theta=\nabla_a k^a$ satisfies 
\be
\label{eqn:raynull}
\frac{\mathrm{d}\theta}{\mathrm{d}\lambda} =-R_{ab} k^a k^b-\sigma_{ab} \sigma^{ab} +\omega_{ab} \omega^{ab}-\frac{\theta^2}{n-2}  \,.
\ee
Unless otherwise noted, we will consider \textit{irrotational} congruences, meaning that the rotation tensor vanishes, $\omega_{ab}=0$.

\subsection{Overview of pointwise energy conditions}
\label{sub:overview}

The first energy conditions that appeared in the literature are of a \textit{pointwise} nature: they restrict some contraction of the stress tensor at every spacetime point. The four main conditions are the weak energy condition (WEC), the strong energy condition (SEC), the dominant energy condition (DEC) and the null energy condition (NEC). Less well-known, but of historical interest, is the trace energy condition (TEC). In this subsection we will review these pointwise energy conditions. Another overview of the classical energy conditions with much detail is \cite{Curiel:2017}. For theories that obey Einstein's equation we may rewrite these energy conditions by eliminating the stress tensor in favour of the Ricci curvature tensor, leading to a geometric form of these conditions, as opposed to the original physical form.

Perhaps the most intuitive of the energy conditions is the WEC. In its physical form it requires that for every future-pointing timelike vector $t^a$
\be
\label{eqn:WEC}
T_{ab} t^a t^b \geq 0 \,.
\ee
The WEC expresses the assumption that the energy density measured by any observer on a timelike curve has to be non-negative. For a perfect fluid the WEC implies $\rho \geq 0$, cf.~(\ref{eqn:Tperfectfluid}). Additionally, the pressure cannot be so negative that it dominates the energy density, or $\rho+P \geq 0$. A visualization of the WEC and the other main energy conditions in the case of a perfect fluid is given in Figure \ref{fig:visual}.

The geometric form of the WEC is
\be
G_{ab} t^a t^b \geq 0 \,.
\ee
It is not straightforward to describe the geometric meaning of the components of the Einstein tensor $G_{ab}$ (see \cite{Curiel:2017} Footnote 11 for a detailed discussion), but when $n\ge3$ the WEC can nevertheless be given a geometric interpretation using an idea due to Feynman, cf.~\cite{FeynmanGravitation} Sec.~11.2. For a small spacelike geodesic ball of area $A>0$ we consider the ratio of its spatial volume with the spatial volume of a corresponding ball in Minkowski space with the same area. The WEC holds iff the ratio becomes greater than or equal to $1$ when the ball shrinks to a point. For detailed derivations of this result, using two different formulations of the limit, we refer to \cite{GIBBONS2007317,PhysRevLett.116.201101}.

\begin{figure}[h!]
	\centering
	\includegraphics[scale=0.4]{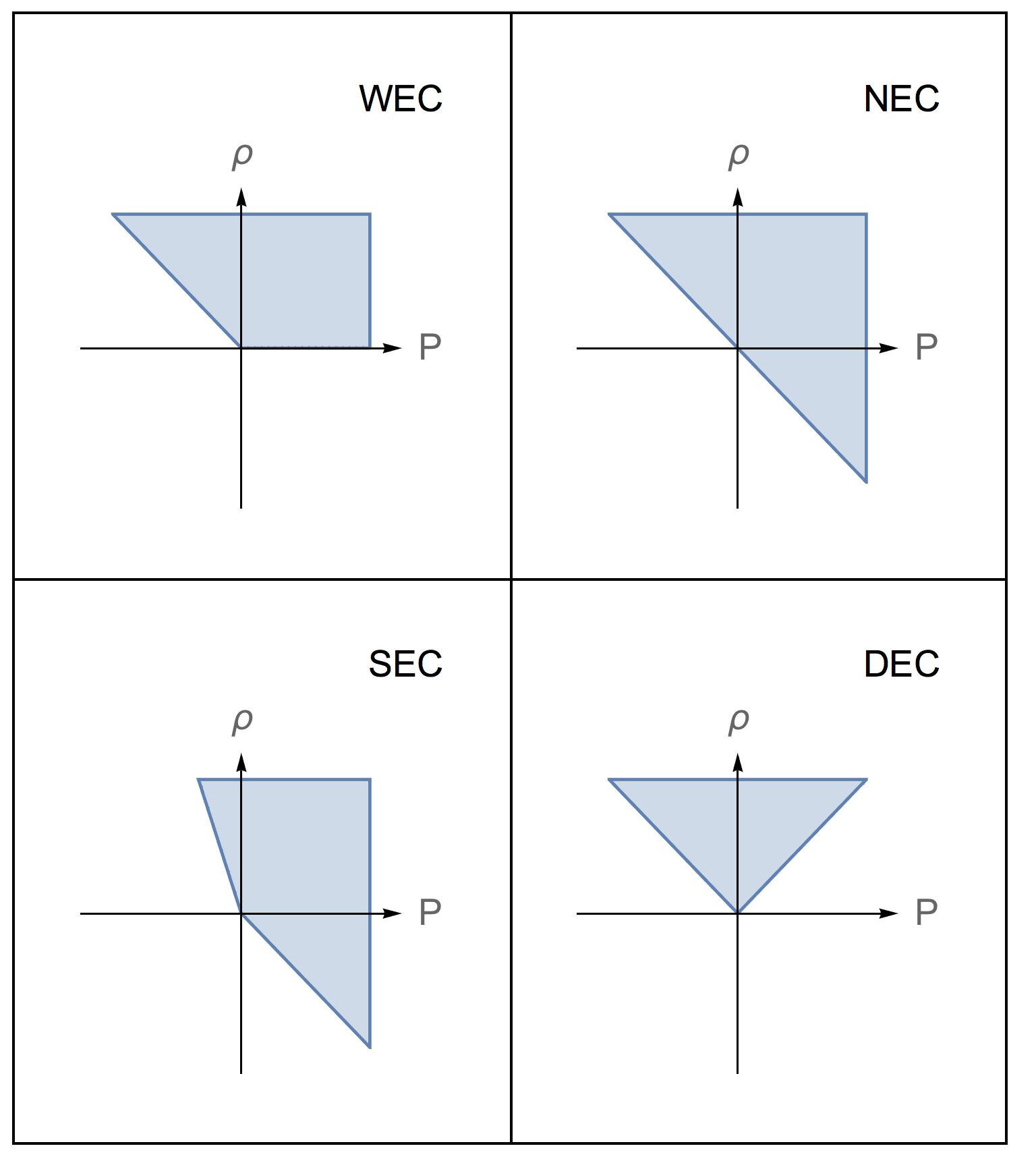}
	\caption{The main energy conditions for perfect fluids represented as regions of allowed energy density and pressure.
	}
	\label{fig:visual}
\end{figure}

The SEC imposes a bound on a more complicated expression,
\be
\label{eqn:EED}
\left(T_{ab} - \frac{T}{n-2}g_{ab}\right) t^a t^b \geq 0\,,
\ee
where we assume $n>2$. For $n=4$, Whittaker \cite{Whittaker:1935} introduced the quantity $T\indices{_0^0}-\frac12T$ as a relativistic analogue of the Newtonian gravitational potential, which appeared in his relativistic formulation of Gauss' law for gravity. Pirani \cite{Pirani:2009} calls this  quantity the \textit{effective density of gravitational mass} while in recent work \cite{Brown:2018hym} it was called the \textit{effective energy density} (EED), using units of energy rather than mass. This suggests that the SEC requires the positivity of an effective energy density, as measured by an observer, but there is no compelling physical argument of why it should be obeyed. The SEC is generally violated more easily than the WEC, as we will see in Subsection~\ref{sub:examples}. However, even though the name suggests it, the SEC does not imply the WEC. The implications between the four main energy conditions are represented in Figure \ref{fig:implications}.

For a perfect fluid we have
\be
T^{ab} - \frac{T}{n-2}g^{ab} = (\rho+P) v^a v^b +\frac{\rho-P}{n-2} g^{ab}
\ee
and the SEC becomes $\rho+P \geq 0$ and $(n-3)\rho +(n-1) P \geq 0$. The geometric form of the SEC is called the \textit{timelike convergence condition},
\be
\label{eqn:timelikeconvergencecondition}
R_{ab} t^a t^b \geq 0 \,.
\ee
It implies that a non-rotating timelike geodesic congruence locally converges. This condition is commonly used and it is one of the main conditions of the Hawking singularity theorem discussed in Section \ref{sub:singularity}.

\begin{figure}[h!]
	\centering
	\includegraphics[scale=1]{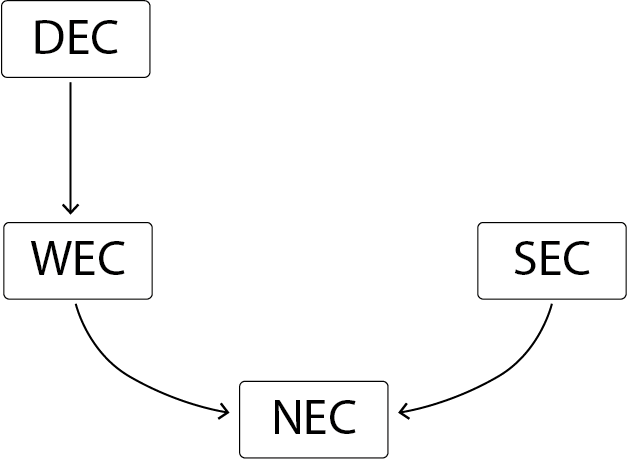}
	\caption{Implications of energy conditions. In the figure an arrow from $a$ to $b$ means $a$ implies $b$.}
	\label{fig:implications}
\end{figure}

The physical form of the DEC can be written as
\be
\label{eqn:DEC}
T_{ab} t^a \xi^b \geq 0 \,,
\ee
for any two co-oriented timelike vectors $t^a$ and $\xi^b$. This is equivalent to the WEC (\ref{eqn:WEC}), with the additional requirement that $-T\indices{^a_b} t^b$ is a future-pointing causal vector field. The latter condition is often written using the quadratic equation $t^c T_{ca}T\indices{^a_b} t^b\le0$. The appeal of (\ref{eqn:DEC}) is that it shows that the set of stress tensors satisfying DEC is closed under sums. (Note that \cite{MartinMorunoVisser:2017} use the quadratic form of the DEC as a motivation for non-linear energy conditions. Although such conditions may have their merits, the claim in Sec.~4 loc.cit.~that sums of stress tensors satisfying the DEC need not satisfy the DEC is incorrect.)

The DEC requires that the flux of energy-momentum measured by an observer is causal and in the direction of the observer's proper time. This is often interpreted as prohibiting superluminal propagation of energy, but Earman \cite{Earman2014} suggests that a well-posed initial value formulation is a better formulation of that idea (see also \cite{Ford:1978qya,Geroch2011}). He notes that the quantity $-T\indices{^a_b} t^b$ does not track the propagation of the underlying matter model and that the DEC can fail for systems which have a perfectly well-defined initial value formulation with a maximum speed equal to that of light (and gravity).

Equivalently, the DEC can be expressed as the requirement that in any orthonormal frame the energy density component of the stress energy tensor dominates all the other components
\be
\label{eqn:dominant}
T_{00} \geq |T_{\mu \nu} | \,.
\ee
For a perfect fluid this becomes $\rho \geq |P|$. The geometric form of the DEC
\be
G_{ab} t^a \xi^b \geq 0 \,,
\ee
does not have a clear interpretation for similar reasons as for the WEC.

The NEC is a variation of the WEC, with the timelike vector replaced by a null vector $k^a$
\be
T_{ab} k^a k^b \geq 0 \,.
\ee
For a perfect fluid the condition becomes $\rho+P \geq 0$. The interpretation of the physical form of the NEC is not straightforward: for observers on null geodesics, the sum of energy density and pressure cannot be negative, but we rarely think of physical observers moving on null geodesics. Its geometric form 
\be
\label{eqn:nullconvergencecondition}
R_{ab} k^a k^b \geq 0 \,,
\ee
is used more often and is sometimes called the \textit{null convergence condition}. It can be considered as a limiting case of the timelike convergence condition (which is related to the SEC) and it implies that a non-rotating null geodesic congruence locally converges, so that gravity is attractive for particles following null geodesics. It has been used in several classical relativity theorems, the most famous being the Penrose singularity theorem and the Hawking area theorem (see Sec.~\ref{sec:applications} for more details).

The four main energy conditions are summarized in Table~\ref{tab:pointwise}. There are several other less known pointwise energy conditions expressing various constraints on the stress energy tensor. These are usually variations of the conditions already mentioned and include the strengthened dominant energy condition (discussed in \cite{Curiel:2017}), the null dominant energy condition (mentioned in \cite{carroll2019spacetime}), the subdominant trace energy condition \cite{Bekenstein:2013ztp} and non-linear conditions, such as the flux energy condition \cite{MartinMorunoVisser:2017}. For later convenience let us only briefly comment on the TEC, which requires that the trace $T$ of the stress tensor has to be non-positive (for our metric conventions). For a perfect fluid this means that $\rho-3P \geq 0$. The TEC was popular until the 1960's when it was discovered by Zel'dovich \cite{zeldovich1961translation} that it was not as general as thought. For further discussion see Subsection~\ref{sub:examples}.

\begin{table}[h!]
	\begin{center}
		\begin{tabular}{|c|c|c|c|} \hline
			\textbf{Condition} &  Physical form & Geometric form & Perfect fluid \\ \hline
			\makecell{WEC} & $T_{ab} t^a t^b \geq 0$ &	$G_{ab} t^a t^b \geq 0$ & \makecell{$\rho \geq 0$ and\\ $\rho+P \geq 0$} \\ \hline
			\makecell{SEC} & $(T_{ab} -\frac{T}{n-2}g_{ab})t^a t^b \geq 0$ &	$R_{ab} t^a t^b \geq 0$ & \makecell{$\rho + P \geq 0$ and\\ $(n-3)\rho+(n-1)P \geq 0$} \\ \hline
			\makecell{DEC} & \makecell{$T_{ab} t^a \xi^b \geq 0$}  &	\makecell{$G_{ab} t^a \xi^b \geq 0$}  & $\rho \geq |P|$ \\ \hline
			\makecell{NEC} & $T_{ab} k^a k^b \geq 0$ &	$R_{ab} k^a k^b \geq 0$ & $\rho + P \geq 0$ \\ \hline
		\end{tabular}
		\caption{The main pointwise energy conditions summarized. Here $t^a$ and $\xi^b$ are co-oriented timelike vectors and
			$k^a$ is a null vector.
		}
		\label{tab:pointwise}
	\end{center}
\end{table}

The NEC is the weakest of the four main pointwise energy conditions. One may ask if it is possible to weaken any of these energy conditions by allowing a lower bound that is negative rather than zero. Tipler \cite{Tipler:1978zz} first showed that this is not possible for the WEC. The following generalized version of Tipler's argument can be applied to all four main energy conditions, taking either $S_{ab}=T_{ab}$ or $S_{ab}=T_{ab}-\frac{T}{n-2}g_{ab}$ at a point.

\begin{proposition}\label{Prop_Tipler3}
	Let $S_{ab}$ be any rank $2$ tensor and $\Gamma$ a set of pairs of vectors that is invariant under positive rescaling. If $S_{ab} \eta^a\xi^b$ is bounded from below as $(\eta^a,\xi^b)$ ranges over $\Gamma$, then the greatest lower bound is zero.
\end{proposition}
\begin{proof}
	For $(\eta^a,\xi^b)\in\Gamma$ let $V^a=r \eta^a$ and $W^b=r \xi^b$ with $r>0$. Then $(V^a,W^b)\in\Gamma$ by scale invariance and $S_{ab} V^a W^b=r^2 S_{ab} \eta^a \xi^b$. If this is bounded below as $r\to\infty$, then $S_{ab}\eta^a \xi^b \geq 0$.
\end{proof}

For the WEC, SEC and DEC one can restrict attention to normalized time-like vectors, which removes the scale invariance used in the proof of Proposition \ref{Prop_Tipler3}. In this case the existence of a lower bound does not imply that the lower bound is non-negative. However, the weaker energy conditions obtained in this way do still imply the NEC, as can be seen from the following generalized version of another result of Tipler \cite{Tipler:1978zz}. (In loc.cit.~Tipler makes the additional assumption that $T_{ab}$ is of Segre Type I, but our proof shows that this assumption is unnecessary.)
\begin{proposition}\label{Prop_Tipler4}
	Let $S_{ab}$ be any rank $2$ tensor in Minkowski space. If $S_{ab} t^at^b$ is bounded from below as $t^a$ ranges over all timelike vectors with $t^at_a=-1$, then $S_{ab}k^ak^b\ge0$ for all null vectors $k^a$.
\end{proposition}
\begin{proof*}
	Let $B$ be a lower bound for $S_{ab}t^at^b$ as $t^a$ ranges over all normalized timelike vectors. Given any null vector $k^a$ we can find another null vector $n^a$ such that $k^an_a=-\frac12$. For every $r>0$ the vector $t^a=\frac{1}{r}k^a+rn^a$ is a normalized timelike vector, so we have
	\be
	B\le \frac{1}{r^2}S_{ab}k^ak^b + (S_{ab}+S_{ba})k^an^b + r^2S_{ab}n^an^b\,.
	\ee
	Multiplying by $r^2$ and taking $r\to0^+$ yields $0\le S_{ab}k^ak^b$ as desired.
\end{proof*}

\subsection{On-shell and off-shell pointwise energy conditions}
\label{subs:onshell}

The various energy conditions intend to describe simple properties of the stress tensor $T_{ab}$ that are shared by many physically reasonable systems in physically reasonable states. Although the term ``physically reasonable'' is flexible, one might think that it ought to include at least that the equations of motion of the system hold, including Einstein's equation. Indeed, without Einstein's equation we cannot turn the physical form of an energy condition into its geometric form, which is typically the first step in applications. It may perhaps come as a surprise that some systems satisfy energy conditions even when the equations of motion are violated. In this subsection we will try to shed some light on this by presenting some general considerations that relate the main energy conditions and the dynamics of the system.

We will call a field configuration $(g_{ab},\Phi)$ \emph{on-shell} when Einstein's equation~(\ref{eqn:EE}) and the equations of motion for the fields $\Phi$ are satisfied. Such configurations are classically possible states of the system. We will only consider on-shell configurations for $n\ge3$, because for $n=2$ we always have $G_{ab}=0$, so if Einstein's equation is satisfied, then all energy conditions hold. By a \emph{test field} configuration we will mean a configuration which satisfies the equations of motion for $\Phi$, but not Einstein's equation. Here we think of the metric $g_{ab}$ as being caused by other fields, $\Phi'$, which are not part of the system under consideration. The fields $\Phi$ are influenced by the metric, but not by the fields $\Phi'$ directly. Furthermore, the fields $\Phi$ do not influence the metric, so they are considered as small perturbations. Finally, an \emph{off-shell} configuration is one where all equations of motion, including Einstein's equation, may be violated.

For simplicity we will consider a scalar field $\phi$ with a Lagrangian of the form
\be
\label{eqn:nlsc}
\mathcal{L}_{\mathrm{sc}}[g,\phi]=\frac12[(\nabla\phi)^2+V[g](\phi)]
\ee
where the functional $V$ may depend on $g$ and its derivatives and on $\phi$, but not on derivatives of $\phi$. An example is $V[g](\phi)=(m^2+\xi R)\phi^2$ for the non-minimally coupled linear scalar field, cf.~(\ref{eqn:lagrangian}). Note that the equation of motion for $\phi$ is a wave equation,
\be
\label{eqn:nlw}
-\Box_g\phi+V'(\phi)=0\,,
\ee
which is in general non-linear.

The wave equation~(\ref{eqn:nlw}) is a symmetric hyperbolic system, and when the potential $V(\phi)$ is well-behaved, e.g.~when it does not depend on derivatives of the metric, then (\ref{eqn:nlw}) coupled to Einstein's equation is a symmetric hyperbolic system of equations. It is well-known that such systems have a well-posed initial value formulation \cite{ChoquetBruhat,Ringstrom,HawkingEllis:1973,wald1984general}, which means that for any initial data (satisfying the problem's constraint equations), there exists a unique maximal globally hyperbolic solution (up to gauge equivalence) which depends continuously on the data. The proof of this fact makes use of energy estimates, which in turn make use of energy conditions. However, the energy conditions used in the proof do not have to be imposed on, or required of, the stress tensor $T_{ab}$ of the field $\phi$. Instead, it suffices to consider them for some auxiliary tensor $S_{ab}$, whose form may correspond e.g.~to the stress tensor of a massless minimally coupled scalar field. In fact, the well-posedness of a system by itself does not seem to imply any of the four main energy conditions of Section \ref{sub:overview}. A simple example violating the WEC is the linear scalar field with an imaginary mass, whose potential function is $V(\phi)=-m^2\phi^2$, $m>0$. (See \cite{Earman2014} and \cite{Curiel:2017} Sec.~5 for further discussion. Note that if $T_{ab}=0$ on a region $\mathcal{O}\subset\Sigma$ of a Cauchy surface $\Sigma$, then $T_{ab}=0$ on $D(\mathcal{O})$, contrary to the claim in \cite{Earman2014} Sec.~5.) 

Given the well-posed initial data formulation, we address in the next two propositions the question whether the validity of an energy condition on-shell also implies its validity off-shell. We will deal with the implication between test field and off-shell validity first, because it is technically easier and we have a slightly more general result.

\begin{proposition}\label{Prop_ECoffshell1}
	Consider a scalar field with the Lagrangian density (\ref{eqn:nlsc}) coupled to general relativity. Let EC denote one of the four main energy conditions (WEC, SEC, DEC or NEC) and assume that EC is satisfied for all smooth test field configurations $(g_{ab},\phi)$, i.e.~ones for which (\ref{eqn:nlw}) holds. Then EC is also satisfied for all smooth off-shell configurations $(g_{ab},\phi)$.
\end{proposition}

The smoothness in Proposition \ref{Prop_ECoffshell1} is not essential and can be weakened to $C^k$ for sufficiently high $k$. 

\begin{proof*}
	Let $(g_{ab},\phi)$ be any field configuration on a smooth manifold $\mathcal{M}$ and let $x\in\mathcal{M}$ be any point. Using local coordinates near $x$ we can locally choose a smooth, spacelike, achronal hypersurface $\Sigma\subset\mathcal{M}$ containing $x$. We can now specify a test field configuration $(g_{ab},\tilde{\phi})$ near $x$ by requiring that $\tilde{\phi}$ has the same initial data on $\Sigma$ as $\phi$. At the point $x$, $\phi$ and $\tilde{\phi}$ have the same stress tensor $T_{ab}$, because it only depends on the initial data at $x$. Because $(g_{ab},\tilde{\phi})$ satisfies the EC by assumption, so does $(g_{ab},\phi)$. Because $x$ and the configuration were arbitrary, the EC holds for all off-shell configurations.
\end{proof*}

For more general fields $\Phi$, e.g.~fields with a gauge symmetry, the proof of Proposition \ref{Prop_ECoffshell1} breaks down, because the initial data on $\Sigma$ need to satisfy constraint equations in order to determine a solution $\tilde{\Phi}$. However, the conclusion will still hold for any configuration $\Phi$ whose data on $\Sigma$ can be modified to satisfy the constraints without changing the stress tensor at a given point $x$. For simple situations, such as the Proca field \cite{Schambach+2018} or the vector potential of electromagnetism \cite{Sanders+2014}, one may explicitly show that this is possible for all configurations, so the conclusion of Proposition \ref{Prop_ECoffshell1} remains valid for all off-shell configurations.

We can use the same strategy to investigate whether the validity of an energy condition for on-shell configurations implies its validity for off-shell (or test field) configurations. In this case we need to take the constraint equations for the metric $g_{ab}$ into account, which complicates the analysis. We are not aware of general results in the literature, but if the potential function $V$ in (\ref{eqn:nlsc}) does not depend on derivatives of $g_{ab}$, we can establish the following result with elementary methods.

\begin{proposition}\label{Prop_ECoffshell2}
	Consider a scalar field with the Lagrangian density (\ref{eqn:nlsc}) coupled to general relativity in dimension $n\ge3$ and assume that $V$ does not depend on derivatives of $g_{ab}$. Let EC denote one of the four main energy conditions (WEC, SEC, DEC or NEC) and assume that EC is satisfied for all smooth on-shell configurations $(g_{ab},\phi)$. Then EC is also satisfied for all smooth off-shell configurations $(g_{ab},\phi)$.
\end{proposition}
\begin{proof*}
	Let $(g_{ab},\phi)$ be any field configuration on a smooth manifold $\mathcal{M}$ and let $x\in\mathcal{M}$ be any point. We may choose local coordinates $x^{\mu}$ with $\mu\in\{0,\ldots,n-1\}$ near $x$ with the following properties: $x=0$, $g_{\mu\nu}(0)=\eta_{\mu\nu}$ the Minkowski metric, $\partial_{\mu}g_{\nu\rho}(0)=0$ and $\partial_{\mu}\phi(0)=0$ if $\mu\ge2$.
	
	The surface $\Sigma=\{x^0=0\}$ is a smooth spacelike hypersurface containing $x=0$ and with local coordinates $x^i$ with $i\in\{1,\ldots,n-1\}$. On $\Sigma$ we will specify initial data for an on-shell configuration $(\tilde{g}_{\mu\nu},\tilde{\phi})$ as follows.
	We first set $\tilde{g}_{ij}=\delta_{ij}$ (the Euclidean metric), $\tilde{\phi}=\phi(0)+x^1\partial_1\phi(0)$ and 
	$\partial_0\tilde{\phi}=\partial_0\phi(0)$. Note that these data are either constant or linear functions on $\Sigma$ which only depend on $x^1$. These data suffice to determine $T_{\mu\nu}$ for $(\tilde{g}_{\mu\nu},\tilde{\phi})$ on $\Sigma$, which also depends only on the coordinate $x^1$. It remains to specify a symmetric tensor $K_{ij}$ on $\Sigma$ satisfying the constraint equations for general relativity. Because the metric on $\Sigma$ is flat these equations read \cite{Bartnik+2004}
	\be
	\partial_jK\indices{_i^j}-\partial_iK\indices{_j^j}=8\pi T_{0i}\,,\\
	\left(K\indices{_i^i}\right)^2-K^{ij}K_{ij}=16\pi T_{00}\,.
	\ee
	To find such a $K_{ij}$ we choose the components as functions of $x^1$ only and we set $K_{ij}=0$, except when $i=1$ or $j=1$ or $i=j=2$. When $i\ge 2$ we choose $K_{1i}=K_{i1}$ to satisfy $\partial_1K_{i1}=8\pi T_{0i}$. We choose $K_{22}$ to satisfy $-\partial_1K_{22}=8\pi T_{00}$ and $K_{22}(0)\not=0$. Finally we choose $K_{11}=8\pi(K_{22})^{-1}(T_{00}+2\sum_{i\ge2}K_{1i}^2)$, which is well-defined near $x=0$. It is easy to see that these choices for $K_{ij}$ satisfy the constraint equations, so the initial data $(\tilde{g}_{ij},K_{ij},\tilde{\phi},\partial_0\tilde{\phi})$ on $\Sigma$ determine an on-shell configuration
	$(\tilde{g}_{ab},\tilde{\phi})$ near $x$.
	
	At the point $x$ the configurations $(g_{ab},\phi)$ and $(\tilde{g}_{ab},\tilde{\phi})$ have the same stress tensor, because $T_{ab}$ only depends on $(g_{ab},\phi,\partial_{\mu}\phi)$ at $x$. Because $(\tilde{g}_{ab},\tilde{\phi})$ satisfies the EC by assumption, so does $(g_{ab},\phi)$. Because $x$ and the configuration were arbitrary, the EC holds for all off-shell configurations.
\end{proof*}

Note that Proposition \ref{Prop_ECoffshell2} applies to the minimally coupled scalar field, but not to general non-minimally coupled scalar fields. We will not investigate here to what extent the proof extends to Lagrangian densities $\mathcal{L}_{\Phi}[g,\phi]$ which depend on derivatives of $g_{ab}$, or to fields $\Phi$ whose initial data must satisfy constraints.

Let us raise instead another natural questions relating energy conditions and dynamics. When a theory has a well-posed initial value formulation, then the validity of an energy condition for a configuration $(g_{ab},\Phi)$ is entirely determined by its initial data on a Cauchy surface $\Sigma$ and one may wonder what initial information is required to verify the energy condition throughout spacetime. In particular, one may wonder whether energy conditions are conserved under the dynamics, i.e.~whether their validity on a Cauchy surface $\Sigma$ implies its validity on the entire spacetime.

For on-shell configurations the situation seems quite complicated and we are not aware of any existing results. For test field configurations, however, the four main energy conditions are in general not conserved, as the following simple example shows. In Minkowski space the function $\phi(x)=\sin(mt)$ is a spatially homogeneous solution to the Klein-Gordon equation~(\ref{eqn:field}) with mass $m>0$ for any coupling constant $\xi$, since $R=0$. The stress tensor (\ref{eqn:tmunu}) takes the form
\be
T_{ab}=m^2(\cos(mt)^2-2\xi\cos(2mt))\nabla_at\nabla_bt +\frac12m^2(1-4\xi)\cos(2mt)g_{ab}\,.
\ee
At $t=\frac{\pi}{4m}$ we have $T_{ab}=\frac12m^2\nabla_at\nabla_bt$, so the DEC and SEC hold for any value of $\xi$. However, for a null vector $k^a$ with $k^a\nabla_at=1$ we have $T_{ab} k^ak^b=m^2(1-2\xi)$ at $t=0$ and $T_{ab} k^ak^b=2\xi m^2$ at $t=\frac{\pi}{2m}$, so the NEC fails when $\xi>\frac12$ or $\xi<0$ respectively.

\subsection{Examples and counter-examples of pointwise energy conditions}
\label{sub:examples}

Here we discuss matter models that obey or violate the main pointwise energy conditions introduced in Section \ref{sub:overview}. A summary of examples and counter-examples is provided in Table \ref{tab:examples} below.

Let us first examine the minimally coupled scalar field with mass $m\ge0$. The stress tensor (\ref{eqn:tmunu}) simplifies to
\be
T_{ab} = \nabla_a\phi\nabla_b\phi -\frac12g_{ab}((\nabla \phi)^2+m^2 \phi^2)\,.
\ee
Due to Proposition \ref{Prop_ECoffshell2} it suffices in this case to study the energy conditions off-shell. For all co-oriented timelike vectors $t^a, \xi^b$ at any point we note that $T_{ab}t^a\xi^b\ge0$, because $t^a\xi_a<0$ and the tensor $t^a\xi^b+\xi^at^b-(t^c\xi_c)g^{ab}$ is positive definite. It follows that the DEC holds and hence so do the WEC and NEC. However, from
\be
T_{ab} -\frac{T}{n-2}g_{ab} = \nabla_a\phi\nabla_b\phi +\frac{1}{n-2}g_{ab}m^2 \phi^2
\ee
it is easily seen that the SEC holds off-shell when $m=0$, but it can be violated when $m>0$.

More generally, consider a scalar field with Lagrangian density of the form (\ref{eqn:nlsc}) with a potential $V(\phi)$ which is independent of $g_{ab}$, so Proposition \ref{Prop_ECoffshell2} applies. The stress tensor takes the form
\be
T_{ab}=\nabla_a\phi\nabla_b\phi - \frac12g_{ab}((\nabla\phi)^2+V(\phi))\,.
\ee
When $V\ge0$ the DEC holds, but the SEC can be violated when $V\not\equiv 0$. When $V\le0$, on the other hand, the SEC holds, but the WEC can be violated when $V\not\equiv0$.

Let us next examine the Proca and Maxwell fields. The stress tensor is given in (\ref{eqn:tensorproca}) and if we express the components of $F_{ab}$ and $A_a$ in an orthonormal basis $e_0^a,\ldots,e_{n-1}^a$ with $t^a=e_0^a$ timelike, then we find for the energy density
\be
T_{ab}t^a t^b=\frac{1}{4}\sum_{\mu=0}^{n-1} \sum_{\nu=0}^{n-1} (F_{\mu\nu})^2+\frac{m^2}{2} \sum_{\mu=0}^{n-1} (A_{\mu})^2 \,,
\ee
where we used the antisymmetry of $F_{ab}$. The energy density of the Proca and Maxwell fields also takes a sum of squares form, just like the minimally coupled scalar field. Thus the WEC holds and hence also the NEC. Using similar arguments as for the minimally coupled scalar field one may show that the DEC also holds, even off-shell.
To examine the SEC we first calculate the trace of $T_{ab}$
\be
T=\frac{4-n}{4} F_{ab}F^{ab}+\frac{2-n}{2} m^2 A_a A^a \,.
\ee
It is interesting to note that the electromagnetic stress tensor ($m=0$) is traceless for $n=4$. Using the components in the same basis as before,
\be
\left(T_{ab}-\frac{T}{n-2}g_{ab}\right)t^a t^b = \frac{1}{2(n-2)}\sum_{i=1}^{n-1} \sum_{j=1}^{n-1} (F_{ij})^2
+\frac{n-3}{n-2}\sum_{i=1}^{n-1} (F_{0i})^2+ m^2 (A_0)^2 \,, 
\ee
so the SEC is satisfied for $n>2$. 

\begin{table}[h!]
	\begin{center}
		\begin{tabular}{|c|c|c|c|c|} \hline
			\textbf{Theory} &  NEC & WEC & DEC & SEC \\ \hline
			massless Klein-Gordon ($\xi=0$, $V(\phi)=0$) & \cmark & \cmark & \cmark & \cmark \\ 
			modified Klein-Gordon ($\xi=0$, $V(\phi)\ge0$) & \cmark & \cmark & \cmark & \xmark \\
			modified Klein-Gordon ($\xi=0$, $V(\phi)\le0$) & \cmark & \xmark & \xmark & \cmark \\
			Klein-Gordon ($\xi \neq 0$, $m\ge0$) & \xmark & \xmark& \xmark & \xmark\\
			Maxwell/Proca ($m \geq 0$)  & \cmark & \cmark & \cmark & \cmark \\
			\hline
		\end{tabular}
	\end{center}
	
	\caption{Validity of pointwise energy conditions for classical fields. The potential function
		$V(\phi)\ge0$ is not identically zero. \cmark\ indicates
		that the condition holds for all off-shell configurations. \xmark\ indicates that the
		conditions can fail even for on-shell configurations.}
	\label{tab:examples}
\end{table}

For non-minimally coupled scalar fields, however, all the main energy conditions can be violated, see e.g.~\cite{Barcelo:2000zf,Barcelo:1999hq} and Section IIA of
\cite{Flanagan:1996gw}. Let us first consider off-shell configurations in an arbitrary metric $g_{ab}$. We can pick an arbitrary null vector $k^a$ and extend it to a geodesic $\gamma$. Contracting with the stress tensor (\ref{eqn:tmunu}) we find
\bea
\label{eqn:tmnnec}
T_{ab} k^a k^b&=&(k^a\nabla_a\phi)^2 +\xi G_{ab}k^ak^b\phi^2-\xi (k^a\nabla_a)^2\phi^2\nonumber\\
&=&(1-2\xi)(k^a\nabla_a\phi)^2+\xi G_{ab}k^ak^b\phi^2-2\xi \phi(k^a\nabla_a)^2\phi\,.
\eea
For any point $x$ on $\gamma$ we can choose a configuration with $\phi(x)=c_1$, $k^a\nabla_a\phi(x)=0$ and
$(k^a\nabla_a)^2\phi(0)=c_2$. Since $\xi\not=0$ we can choose for each $c_1\not=0$ a $c_2$ such that the expression becomes negative. This means that the NEC is violated and hence so are the other main energy conditions.

The argument above can be extended to test field configurations if we use the well-posedness of the characteristic (or null) initial value problem for the Klein-Gordon equation \cite{Lupo2018}. We can choose a partially null Cauchy surface $\Sigma$ that contains a part of the null geodesic $\gamma$ near $x$ and we can prescribe initial values on $\Sigma$ that will violate the NEC at $x$.

The conclusion remains true for on-shell configurations, although the initial value problem is rather subtle. Because the stress tensor (\ref{eqn:tmunu}) contains a multiple of the Einstein tensor $G_{ab}$, we can write the equations of motion (\ref{eqn:field},\ref{eqn:EE}) as
\bea
\label{eqn:nonminimal1}
\frac{1}{8\pi}(1-8\pi\xi\phi^2)G_{ab}&=&
(\nabla_a \phi)(\nabla_b \phi)-\frac{1}{2} g_{ab} ((\nabla \phi)^2+m^2 \phi^2)+\xi(g_{ab} \Box_g-\nabla_a \nabla_b) \phi^2 \,, \nonumber\\
(-\Box_g+m^2+\xi R)\phi&=&0 \,.
\eea
When $\xi>0$, the coefficient in front of $G_{ab}$ vanishes at the points where $8\pi\xi\phi^2=1$. At these points the order of Einstein's equation reduces and we should expect the initial value problem to be ill-posed. Values for $\phi$ with $8\pi\xi\phi^2>1$ are called ``trans-Planckian'' and they are extremely large. \cite{Visser:1999de} argues that trans-Planckian values for $\phi$ should not be considered problematic, unless the components of the stress tensor $T_{ab}$ (in a suitable orthonormal basis) also reach such unphysically large values. Other authors, however, prefer to discard trans-Planckian values of $\phi$ as unphysical, either because of their magnitude, or for the following reason. When $8\pi\xi\phi^2\not=1$ we can write Einstein's equation in terms of an effective stress tensor as
\bea
\label{eqn:nonminimal2}
G_{ab}&=&8\pi\frac{1}{1-8\pi\xi\phi^2}T^{\mathrm{eff}}_{ab}\nonumber\\
T^{\mathrm{eff}}_{ab}&=&(\nabla_a \phi)(\nabla_b \phi)-\frac{1}{2} g_{ab} ((\nabla \phi)^2+m^2 \phi^2)
+\xi(g_{ab} \Box_g-\nabla_a \nabla_b) \phi^2\,.
\eea
Note that our definition of $T^{\mathrm{eff}}_{ab}$ differs from that of \cite{Barcelo:2000zf} in that we separate off the factor $(1-8\pi\xi\phi^2)^{-1}$, which we may think of as an effective Newton's constant. For trans-Planckian values the effective Newton's constant changes sign, which has been cited as another reason to consider trans-Planckian values as unphysical
\cite{Brown:2018hym}.

When $8\pi\xi\phi^2<1$ we can obtain solutions to the initial value problem as follows \cite{BEKENSTEIN1974535,Barcelo:2000zf,Barcelo:1999hq,Huebner_1995}. We can identify the system with a minimally coupled system $(\tilde{g}_{ab},\tilde{\phi})$, where 
\bea
\label{eqn:couplingtrafo}
\tilde{g}_{ab}&:=&|1-8\pi\xi\phi^2|^{\frac{-2}{n-2}}g_{ab}\nonumber\\
\tilde{\phi}&:=&F(\phi)\,,
\eea
and where $\xi_c:=\frac{n-2}{4(n-1)}$ is the conformal coupling constant and the function $F(x)$ satisfies
$\partial_xF(x)=|1-8\pi\xi\phi^2|^{-1}\sqrt{1-8\pi\xi\left(1-\frac{\xi}{\xi_c}\right)x^2}$. The equations of motion turn out to be
\be
\label{eqn:minimal}
\frac{1}{8\pi}\tilde{G}_{ab}=
(\tilde{\nabla}_a \tilde{\phi})(\tilde{\nabla}_b \tilde{\phi})-\frac{1}{2} \tilde{g}_{ab} ((\tilde{\nabla} \tilde{\phi})^2+M(\tilde{\phi})
)-\Box_{\tilde{g}}\tilde{\phi}+\frac12M'(\tilde{\phi})=0 \,,
\ee
where $M(\tilde{\phi})=m^2|1-8\pi\xi\phi^2|^{\frac{-n}{n-2}}\phi^2$ with $\phi=F^{-1}(\tilde{\phi})$. Although the potential $M(\tilde{\phi})$ is non-linear, the minimally coupled system (\ref{eqn:minimal}) has a well-posed initial value problem \cite{Ringstrom}. This allows us to transform initial data for (\ref{eqn:nonminimal2}) with $8\pi\xi(\phi|_{\Sigma})^2<1$ into initial data for (\ref{eqn:minimal}), then find the maximal globally hyperbolic development $(\tilde{g}_{ab},\tilde{\phi})$ and, as long as the values of $\tilde{\phi}$ remain in the range of $F$, we can transform back to a solution $(g_{ab},\phi)$.\footnote{For $\xi\not\in(0,\xi_c)$, the range of $F$ is $\mathbb{R}$, so this imposes no restriction. For $\xi\in(0,\xi_c)$ the range of $F$ is bounded and it is not immediately obvious whether the values of $\tilde{\phi}$ may dynamically develop out of this range.}
The solution found in this way satisfies $8\pi\xi\phi^2<1$ everywhere, but it may not be a maximal globally hyperbolic development, because it may be possible to extend the solution $(g_{ab},\phi)$ to a larger region by allowing $\phi$ to take values with $8\pi\xi\phi^2\ge1$.\footnote{For initial data with $\pm\phi|_{\Sigma}>(8\pi\xi)^{-\frac12}$ we can often proceed in an analogous way to find a well-posed initial value formulation with solutions satisfying $\pm\phi>(8\pi\xi)^{-\frac12}$. Indeed, we can apply the same transformation (\ref{eqn:couplingtrafo}) when $\xi\not\in(0,\xi_c)$, or when $\xi\in(0,\xi_c)$ and $(8\pi\xi)^{-\frac12}<|\phi|<\left(8\pi\xi\left(1-\frac{\xi}{\xi_c}\right)\right)^{-\frac12}$.}

Notwithstanding the subtleties in the initial value problem for non-minimally coupled scalar fields, one may show that the NEC can be violated for all couplings $\xi\not=0$. In the massless case a range of explicit examples of violations, including naked singularities and wormholes, were found in \cite{Barcelo:2000zf,Barcelo:1999hq}. Although all these examples include trans-Planckian values for
$\phi$, one may also construct violations with $8\pi\xi\phi^2<1$ for all $m\ge0$ and $\xi\not=0$, e.g.~by engineering initial data as in the proof of Proposition \ref{Prop_ECoffshell2}. The only added difficulty is that one has to obtain suitable values for second derivatives of $\phi$ in order to violate the NEC. However, using the equations of motion one can eliminate all occurrences of second normal derivatives of $\phi$ and of the scalar curvature $R$ and one can show that initial data solving the constraints and violating the NEC do exist.

\medskip

Finally we will briefly discuss the history of violation of the trace energy condition (TEC), $-T\ge0$. 
This example is historically and physically interesting and is often used as a cautionary tale against ``physically reasonable'' assumptions. However direct analogies with other energy conditions should be avoided as the TEC is its own special case. 

In the late 1950's it was proven that the TEC holds for electromagnetic interactions of a Fermi gas even in high temperatures (see Ch.~6.11 of \cite{zel1996stars}). This created the expectation that the proof could be generalized to other interactions. However, there was an indication that the inequality was not as general as thought. The limit of the inequality $\rho \to 3P$ leads to a limit of the speed of sound $a_{\mathrm{sound}} \to c/\sqrt{3}$, but any relativistic theory would be expected to have as limit the speed of light.  

The proof of violation of the TEC came from Zel'dovich in 1961 \cite{zeldovich1961translation}. He considered baryonic interactions such as those in neutron stars. The interaction there is not governed by the Coulomb potential, but rather by the Yukawa potential
\be
V(r)=\frac{g e^{-\mu r}}{r} \,,
\ee
where $g$ is the charge and $\mu^{-1}$ the interaction radius. The exponential allows for uniform charge density for macroscopic systems and the energy density is 
\be
\rho=n\rho_1=nm+\frac{g^2 n^2}{2}\int \frac{e^{-\mu r}}{r} \mathrm{d}vol= nm+\frac{2\pi g^2 n^2}{\mu^2} \,,
\ee
where we are considering particles at rest with rest mass $m$, $n$ is the number density and $\rho_1$ is the energy of one particle. The pressure is
\be
P=-\frac{\partial \rho_1}{\partial n^{-1}}=\frac{2\pi g^2 n^2}{\mu^2} \,.
\ee
At the limit of large $n$ we have $\rho \to P$ which violates the TEC. Thus, although the TEC seemed like a reasonable condition, it was strongly dependent on the form of the interaction potential. The TEC is largely abandoned, but a weaker form of it, the subdominant trace energy condition has appeared in recent work \cite{Bekenstein:2013ztp}.

\medskip

Thus we see that all pointwise energy conditions can be violated. In particular the NEC, the weakest of the main energy conditions, can be violated by the non-minimally coupled scalar field, which is a relatively simple classical system. This fact raises serious questions about the range of validity of the energy conditions and even about whether they should be used at all \cite{Barcelo:2002bv,Visser:1999de}. Other examples of scalar field theories which can violate the NEC can be found by considering Lagrangians that depend on second order or higher derivatives, but whose equations of motion are still of second order (though not necessarily quasi-linear), see \cite{Rubakov_2014} for a review. Nevertheless, there are indications that violations of an energy condition at some point or region of spacetime are often compensated for, or even overcompensated for, at other points, both for quantum fields \cite{Ford:1999qv,Bekenstein:2013ztp} and for classical ones \cite{Fewster:2007ec}. In Section \ref{subs:ANEC} we will therefore consider the validity of averaged energy conditions for classical fields, taking averages along causal geodesics. First, however, we will consider quantum fields, where violations of energy conditions are generic and averaging over spacetime regions is a standard procedure.

\section{Quantum energy inequalities}
\label{sec:QEI}

In the mid 1960's, due to a simple argument by Epstein, Glaser and Jaffe \cite{Epstein:1965zza}, it became apparent that the positivity of the energy density is in general incompatible with quantum field theory (QFT). In particular they proved that non-positivity is necessary for any (non-vanishing) Wightman field in Minkowski space with vanishing vacuum expectation values. To review their result we need the Reeh-Schlieder Theorem for local observables (see \cite{haag2012local} (II 5.3) or \cite{reeh1961bemerkungen} for the original article in German).

\begin{theorem}
	\label{the:Reeh}
	Let $\mathcal{A}(\mathcal{O})$ the set of all operators of a QFT localized in a fixed region $\mathcal{O}$ in Minkowski space and $\Omega$ the vacuum state. Then the set of vectors $\mathcal{A}(\mathcal{O})\Omega$ is dense in Hilbert space $\mathcal{H}$. 
\end{theorem}

Then the following theorem is immediate \cite{Epstein:1965zza}. 

\begin{theorem}
	\label{the:EpsteinGlaserJaffe}
	Let $A \in \mathcal{A}(\mathcal{O})$ be a self-adjoint local operator of a QFT such that $\langle \psi| A \psi\rangle \geq 0$ for all $\psi$ in the domain of $A$. Suppose $\langle \Omega| A\Omega \rangle=0$. Then $A \equiv 0$. 
\end{theorem}

\begin{proof}
	Since $\langle \Omega| A\Omega \rangle=0$ and $A$ is non-negative, $\| A^{1/2} \Omega \|^2=0$. Then $A^{1/2} \Omega=0$ and so $A\Omega=0$. From Theorem~\ref{the:Reeh} and local commutativity, $A$ must vanish.
\end{proof}

Another way to state the argument is that if $A$ has a zero expectation value in vacuum, it is either identically zero or it necessarily admits negative measurement values (see \cite{Fewster:2012yh} for discussion).

The inevitability of negative energy densities in QFT could potentially have drastic consequences. In 1978 L.H.~Ford \cite{Ford:1978qya} realized that negative energy densities and, more importantly, the existence of negative energy fluxes could lead to the violation of the second law of thermodynamics. Imagine a negative energy flux, initially in a pure state, directed towards a hot body. Then the absorption of the negative energy would decrease the temperature and thus the entropy. A similar situation is possible if the negative energy flux is directed towards a black hole.

In the same seminal paper Ford showed that a macroscopic violation of the second law could be avoided if the magnitude and duration of the negative energy flux is constrained by an inequality of the form
\be
\label{eqn:firstQEI}
|F| \lesssim t_0^{-2} \,,
\ee
where $F$ is the flux and $t_0$ the time that it lasts. Under this assumption, the magnitude of the change of the energy of the absorber cannot exceed $t_0 |F|$. By the time-energy uncertainty relation $\Delta E  \gtrsim t_0^{-1}$ so no macroscopic violation can occur. For a Schwarzschild black hole he argued that a similar conclusion holds, even without invoking the uncertainty principle. Furthermore, Ford also showed that free quantum scalar fields in flat two-dimensional Minkowski space obey an inequality of the form of (\ref{eqn:firstQEI}), even though they violate classical energy conditions.

Equation (\ref{eqn:firstQEI}) is the first example of a \emph{quantum energy inequality} (QEI). Since the pioneering work of Ford in \cite{Ford:1978qya}, QEIs have been derived for a variety of quantum fields in flat and curved spacetimes and in general number of dimensions.

\medskip

For quantum fields, QEIs are formulated in terms of the stress tensor of the field. It is important to note, however, that the definition of the stress tensor in (\ref{eqn:stresstensor}) does not immediately generalize to QFTs. A common way to handle this problem is to formulate a set of assumptions that any reasonable definition of a quantum stress tensor should satisfy and to assume or prove that such a tensor exists.

In Minkowski space it is usually assumed that for any QFT the spacetime translations are generated by self-adjoint operators $P_{\mu}$ on the vacuum Hilbert space. One may then assume that a stress tensor exists, which should be given by a local quantum field $T_{\mu\nu}(x)$ which locally generates the translations in the inertial coordinates $x^{\mu}$, i.e.
\be
\label{eq:qTaxiom}
\int_{x^0=0} T_{0\mu}(x)\ \mathrm{d}^{n-1}x = P_{\mu}
\ee
in some suitable sense. Additional technical assumptions are also often required of the stress tensor, see e.g.~\cite{Epstein:1965zza,Verch:1999nt}.

There is no hope of extending (\ref{eq:qTaxiom}) to a generally covariant setting, because a generic spacetime has no symmetries and hence no operators $P_{\mu}$. However, different axioms have been formulated which aim to generalize the definition (\ref{eqn:stresstensor}) of the stress tensor to the setting of perturbatively interacting quantum fields \cite{Hollands:2004yh}. For polynomial interactions the difficulty is to treat polynomial terms in (derivatives of) the quantum field, which can be defined using regularization and renormalization techniques. In general there is no uniqueness for the stress tensor of a quantum field.

For later reference we recall here the construction of the stress tensor for a free quantum scalar field of any mass $m\ge0$ and scalar curvature coupling $\xi\in\mathbb{R}$ in any globally hyperbolic spacetime (cf.~\cite{Hollands:2004yh}). This construction extends the normal ordering prescription in Minkowski space and the extension to curved spacetimes is local and covariant. In Minkowski space one can verify that the stress tensor satisfies (\ref{eq:qTaxiom}) and we refer to \cite{Brunetti2003} for the closely related computation of the relative Cauchy evolution in general curved spacetimes.

Because the classical stress tensor (\ref{eqn:tmunu}) of a free scalar field is quadratic in the field, we can choose a partial differential operator $T_{ab}^{\mathrm{split}}$ on $M\times M$ such that
\be
T_{ab}=\iota^*\left(T_{ab}^{\mathrm{split}}\,\phi\otimes\phi\right)\,,
\ee
where $\iota^*$ is the pull-back under the diagonal map $\iota:M\to M\times M:x\mapsto (x,x)$ and $(\phi\otimes\phi)(x,x')=\phi(x)\phi(x')$. To be specific, we choose the operator $T_{ab}^{\mathrm{split}}$ near the diagonal $x=x'$ in $M\times M$ to be
\bea
\label{eqn:pointsplitop}
 T_{ab}^{\mathrm{split}}&=&(1-2\xi) (\nabla_a \otimes \nabla_b)-\frac{1}{2}g'_{ab} \left[(1-4\xi)(g')^{cd} (\nabla_c \otimes \nabla_d)-4\xi(\mathbb{1}\otimes_s \Box_g)
+\mathbb{1} \otimes_s (m^2+\xi R)\right]\nonumber\\
&&\qquad +\frac12\xi g'_{ac}(\mathbb{1}\otimes (R\indices{^c_b}-2\nabla^c\nabla_b))
+\frac12\xi g'_{cb}((R\indices{_a^c}-\nabla_a\nabla^c)\otimes\mathbb{1})\,,
\eea
where $g^{ca}g'_{ab}(x,x')$ is the parallel propagator from $T_{x'}M$ to $T_xM$ and $\otimes_s$ the symmetrized tensor product.

The tensor product $\phi\otimes\phi$ is well-defined also for quantum fields and so is $T_{ab}^{\mathrm{split}}(\phi\otimes\phi)$, but the pull-back is in general ill-defined due to the distributional nature of the quantum fields. Nevertheless, using the Hadamard series we can construct a local and covariant distribution $H_2(x,x')$, which is defined near the diagonal in $M\times M$, such that
\be
\label{def:Tfin}
T^{\mathrm{fin}}_{ab} := \iota^*\,T_{ab}^{\mathrm{split}}(\phi\otimes\phi-H_2)
\ee
does make sense. More precisely, there is a large set of states $\omega$, the so-called Hadamard states, for which the distribution
\be
T_{ab}^{\mathrm{split}}\left(\langle\phi\otimes\phi\rangle_{\omega}-H_2\right)\nonumber
\ee
is a continuous tensor field on $M\times M$ and its restriction to the diagonal $M$ is even smooth, but it is not necessarily conserved. Furthermore, we can find a local and covariant function of the metric $Q$ on $M$ such that the operator $T^{\mathrm{fin}}_{ab}-g_{ab}Q$ is conserved. In this way we find a suitable quantum stress tensor $T^{\mathrm{ren}}_{ab}$ with expectation value
\be
\label{def:Tren}
\langle T^{\mathrm{ren}}_{ab}(x) \rangle_{\omega} =\langle T^{\mathrm{fin}}_{ab}(x) \rangle_{\omega} -Q(x) g_{ab} (x) + C_{ab}(x) \,,
\ee
which satisfies all the desired assumptions. Here $C_{ab}(x)$ is the remaining renormalization freedom, which has been classified in \cite{Hollands:2004yh}. It consists of a smooth, symmetric, conserved tensor field $C_{ab}$ that does not spoil the validity of the axioms, but we will not need its detailed form.
If $\omega_0$ is the Minkowski vacuum state of the theory (or another quasi-free state), we can express the normal ordered stress tensor (w.r.t.~$\omega_0$) as $\nord{\,T_{ab}}_{\omega_0}=T^{\mathrm{ren}}_{ab}-\langle T^{\mathrm{ren}}_{ab}\rangle_{\omega_0}$.

\subsection{Overview of quantum energy inequalities}
\label{sub:QEIover}

In this section we present some general aspects of QEIs. We refer the reader to \cite{Fewster:2012yh, Fewster2017QEIs} for more details.

For a general quantum field with a quantum stress tensor, we let $\rho$ denote the energy density, the EED (\ref{eqn:EED}), or a similar quantity that we wish to bound. We will use the terms \emph{quantum weak energy inequality} (QWEI), \emph{quantum strong energy inequality} (QSEI) and \emph{quantum null energy inequality} (QNEI), according to which component(s) of the stress tensor we wish to bound. 
E.g., when we consider QWEIs for free scalar fields we can take $\rho= T^{\mathrm{ren}}_{ab}t^a t^b$, where $t^a$ is smooth timelike vector field. We emphasize that our terminology does not entail any restrictions on the causal nature of the region over which the average will be taken, so we can consider e.g.~a QNEI along a timelike curve.

The most general form of a QEI that bounds $\rho$ takes the form
\be
\label{eqn:QEIgen}
\langle \rho (f) \rangle_\omega \geq -\langle \mathfrak{Q}(f) \rangle_\omega \,,
\ee
where $f$ is a suitable non-negative test function or distribution on spacetime and the operator $\mathfrak{Q}(f)$ is allowed to be unbounded. This inequality should be satisfied for all states $\omega$ in a suitable class and for free fields we will always consider the class of Hadamard states. We may now try to classify the QEI by means of the properties of the operator $\mathfrak{Q}(f)$.

\medskip

A first concern is that the QEI may be trivially satisfied (or violated), e.g.~when $\mathfrak{Q}(f)=-\rho$ (or $\mathfrak{Q}(f)=-\rho-1$). To avoid such trivial bounds Fewster and Osterbrink \cite{Fewster:2007ec} developed a criterion to distinguish them. The main idea is that the lower bound of the QEI needs to be of lower order, in terms of energy, than the averaged energy density.

\begin{definition}
	\label{def:nontrivial}
	A QEI is called \emph{non-trivial} if there do not exist constants $c$, $c'$ such that
	\be
	\langle \rho (f) \rangle_\omega \leq c+c' | \langle \mathfrak{Q}(f) \rangle_\omega| \,,
	\ee
	for all states $\omega$ in the class, unless $f$ is identically zero.
\end{definition}

For non-minimally coupled scalar fields, \cite{Fewster:2017mtt} showed that any bound where the operator $\mathfrak{Q}(f)$ depends only on the fields $\nord{\Phi^2}$ and $\mathbb{1}$, while the energy density involves terms with derivatives of $\Phi$, is non-trivial.

\medskip

When the operator $\mathfrak{Q}(f)$ is a multiple of the identity operator we call the QEI \emph{state-independent}. Because
$\langle \mathbb{1} \rangle_\omega=1$ we see that $\langle \mathfrak{Q}(f) \rangle_\omega=Q(f)$ is indeed independent of the state. This also makes the QEI non-trivial. 

For scalar fields, the existence of a state-independent QEI implies that the corresponding classical theory satisfies a corresponding pointwise energy condition (cf.~Thm.~\ref{thm:classlimit} below). For a simple physical explanation we note that a single-particle state can only admit negative energies for some test function if the classical model violates the energy condition.
If $\Phi_1$ is the normalized state vector of this state we have
\be
\langle \Phi_1, \rho(x) \Phi_1 \rangle <0 \,, \, \textrm{for some} \, \, x \,.
\ee
We may then consider the $N$-particle state vector $\Phi_N$ which is found by tensoring together $N$ copies of $\Phi_1$. Then the average energy is
\be
\langle \Phi_N, \rho(x) \Phi_N \rangle=N\langle \Phi_1, \rho(x) \Phi_1 \rangle \,.
\ee
Because the average energy scales with $N$ we cannot have a state-independent bound. Explicit examples of this are given in \cite{Fewster:2007ec} for the energy density of the non-minimally coupled scalar field and in \cite{Fewster:2018pey} for the EED of the minimally coupled scalar field. Note that these references do obtain non-trivial QEIs, but they are state-dependent.

We should note that this line of argument does not work for fermionic fields. Since they follow Pauli's exclusion principle, we cannot tensor together multiple copies of single-particle states. However, QEIs have been established for Dirac and Rarita-Schwinger fields as we'll discuss in Section \ref{sub:examplesQEI}.

For interacting fields the situation is also different. As in the free theory, a single-particle state can only admit negative energies if the classical theory violates the energy condition. However, the average energy density of an $N$-particle state is in general no longer additive with respect to tensor products of single-particle states,
\be
\langle \Phi_N, \rho(x) \Phi_N \rangle \neq N\langle \Phi_1, \rho(x) \Phi_1 \rangle \,.
\ee
Interestingly, this non-linearity can lead to state-independent QEIs for interacting theories even though it is easy to construct single-particle states with negative energy. An example of that is the QEI for the Ising model \cite{Bostelmann:2013mxa}. 

\medskip

Another natural condition to impose on a QEI is that the operator $\mathfrak{Q}(f)$ in the lower bound depends in a local and covariant way on the spacetime metric and the test function $f$, cf.~\cite{Fewster:2006iy}. We will call such QEIs \emph{absolute} and they naturally arise for free fields when the stress tensor is regularized and renormalized using the Hadamard prescription. Note that the last two terms in (\ref{def:Tren}) yield state-independent, local and covariant terms, so our terminology is independent of the renormalization freedom. However, if we regularize the stress tensor using a reference state $\omega_0$ instead of the Hadamard parametrix, we do not automatically find an absolute QEI. Instead one arrives at a QEI of the form
\be
\langle\nord{\rho}_{\omega_0}(f)\rangle_{\omega} = \langle  \rho (f)\rangle_{\omega} - \langle  \rho (f)\rangle_{\omega_0}\geq -\langle \mathfrak{Q}_{\omega_0}(f)\rangle_\omega \,,
\ee
which is called a \emph{difference} QEI. Note that the lower bound may depend on the reference state $\omega_0$ as well as on the state $\omega$. By rearranging this difference QEI and introducing $\mathfrak{Q}(f)=\mathfrak{Q}_{\omega_0}(f)-(\nord{\rho}_{\omega_0} (f)) \mathbb{1}$ we do recover the form of (\ref{eqn:QEIgen}), but this only yields an absolute QEI when the lower bound is local and covariant. In general this may fail, because it is impossible to make a local and covariant choice of a preferred reference state in general spacetimes. An absolute QEI is only recovered in this way when
\be
\mathfrak{Q}_{\omega_0}-\mathfrak{Q}_{\omega_1}=\left(\rho_{\omega_0} (f)-\rho_{\omega_1} (f) \right) \mathbb{1} \,,
\ee
for all states $\omega_0$ and $\omega_1$ \cite{Fewster:2006iy}. The quantum inequality of the Wick square of a free scalar field, e.g., satisfies this condition.

\medskip

Finally, QEIs can be distinguished by the support of the distribution $f$. We will mostly be interested in \textit{worldline} QEIs, where
$f$ is supported on a smooth timelike curve $\gamma$. (For a discussion of QEIs on null curves see Sec.~\ref{sub:examplesQEI}.)
For worldline QEIs we can typically change the notation and write them as lower bounds on 
\be 
\langle \rho \circ \gamma \rangle_\omega (f^2)=\int_{-\infty}^{\infty} \mathrm{d}\tau f^2 (\tau) \langle \rho\rangle_\omega   (\gamma(\tau)) \,,
\ee 
where $\tau$ is an affine parameter (often the proper time) along $\gamma$ and $f \in C_0^{\infty}(\mathbb{R}, \mathbb{R})$ is a real-valued test function. (Interestingly, it does not seem possible to replace $f^2$ by a general positive test function $h\ge0$, because $f=\sqrt{h}$ may not be regular enough.) We use the term \textit{worldvolume} QEI when the average is taken over a spacetime region, or over a region in a smooth timelike hypersurface $\Sigma$ in $M$ of dimension $\ge2$.

\subsection{Examples of quantum energy inequalities}
\label{sub:examplesQEI}

One of the simplest examples of a QEI is for the minimally coupled scalar field of mass $m\ge0$ in Minkowski space. Along the worldline $\gamma$ of an inertial observer in four dimensions the normal ordered energy density $\nord{\rho} = T^{\mathrm{ren}}_{00}-\langle T^{\mathrm{ren}}_{00}\rangle_{\omega_0}$ (where $\omega_0$ is the Minkowski vacuum) satisfies the QWEI
\be
\label{eqn:QEIminflat}
\langle \nord{\rho} \circ \gamma \rangle_\omega (f^2) \geq -\frac{1}{16\pi^2} \int_{-\infty}^\infty \mathrm{d}t\, |f''(t)|^2
\ee
for all Hadamard states $\omega$. This result was first derived by Ford and Roman~\cite{Ford:1994bj,PhysRevD.55.2082} for Lorentzian functions $f(t)=\frac{t_0}{\pi(t^2+t_0^2)}$, where $t_0>0$ is a characteristic time scale. After it was shown \cite{Ford:1997fa} that one should generally use sampling functions with compact support, Fewster and Eveson~\cite{Fewster:1998pu} extended the scope of the inequality of (\ref{eqn:QEIminflat}) to include all such functions. 

We can understand the physical behaviour of this bound by rescaling the sampling function,
\be
\label{eqn:fscale}
f_{t_0}(t)=\frac{f(t/t_0)}{\sqrt{t_0}} \,.
\ee
Then (\ref{eqn:QEIminflat}) becomes
\be
\label{eqn:QEIminflatscale}
\frac{1}{t_0} \int_{-\infty}^{\infty} \mathrm{d}t\, \langle \nord{\rho} (\gamma(t)) \rangle_{\omega}\  |f(t/t_0)|^2 \,  \geq -\frac{C}{t_0^4} \,,
\ee
where the constant $C$ depends on the function $f$ (e.g.~for a Lorentzian function, $C=3/(32\pi^2)$). For $t_0 \to 0$ the lower bound approaches negative infinity, consistent with the fact that the expectation value of the energy density at a point is unbounded below. For $t_0 \to \infty$ the bound goes to zero, connecting with the averaged weak energy condition (see
Sec.~\ref{subs:ANEC}).

Using the same methods one can prove QEIs which average over general timelike curves, not just geodesics. The lower bound then also involves the acceleration of the curve $\gamma$. As an example, the QWEI (\ref{eqn:QEIminflat}) when averaged over a uniformly accelerating timelike curve becomes \cite{Fewster:2006kt}
\be
\langle \nord{\rho} \circ \gamma \rangle_\omega (f^2) \geq -\frac{1}{16\pi^2} \int_{-\infty}^\infty \mathrm{d}t \left(|f''(t)|^2+2\alpha_0^2|f'(t)|^2+\frac{11 \alpha_0^4}{30} |f(t)|^2 \right)
\ee
where $\alpha_0$ is the acceleration. For $\alpha_0 t \gg 1$ it is the last term that dominates the bound. This example shows that observers on non-inertial trajectories can experience negative energies for long times. Note, however, that the energy required to keep the constant acceleration grows exponentially with proper time, faster than the negative energy experienced by the observer \cite{Fewster2017QEIs}. Ford and Roman \cite{Ford:2013kga} also considered the case of accelerated observers and studied the worldline of a particle undergoing sinusoidal motion in space. They showed that in this case it is possible for the averaged energy density to become arbitrarily negative at a linear rate. As they note, though, this example does not change the constraints set by inertial trajectory QEIs on macroscopic effects with negative energy.  

\medskip

Historically, the first QEIs were difference QWEIs in Minkowski space, followed by difference QWEIs in static spacetimes \cite{Fewster:1998xn}. The first QWEI for the minimally coupled scalar field on general spacetimes was derived by Fewster \cite{Fewster:1999gj}. It is a state-independent difference QWEI, which makes use of an arbitrary Hadamard reference state $\omega_0$ to regularize the energy density. If we write 
$\rho^{\mathrm{split}}(\tau,\tau') = \left(\xi^a\xi^bT_{ab}^{\mathrm{split}}\phi\otimes\phi)\right)(\gamma(\tau),\gamma(\tau'))$ for the (unregularized) point-split energy density, with both points on the curve $\gamma$, then the QWEI reads
\be
\label{eqn:QEIgendiff}
\langle \nord{\rho}_{\omega_0} \circ \gamma \rangle_\omega(f^2) \geq -\int_0^\infty \frac{\mathrm{d}\alpha}{\pi} \int_{-\infty}^\infty
\mathrm{d}\tau\ \mathrm{d}\tau'\ \langle \rho^{\mathrm{split}}(\tau,\tau')\rangle_{\omega_0} \overline{f_\alpha(\tau)} f_\alpha(\tau') \,,
\ee
where $f_\alpha(\tau)=f(\tau) e^{i\alpha \tau}$. The method of derivation of this QWEI has since been used in several other cases, so it is worth summarizing its main points. First we introduce the normal ordered point-split energy density
\be
\langle \nord{\rho}_{\omega_0}^{\mathrm{split}} \rangle_\omega\ (\tau,\tau') =\langle \rho^{\mathrm{split}}(\tau,\tau')
\rangle_\omega -\langle \rho^{\mathrm{split}}(\tau,\tau')\rangle_{\omega_0} \,,
\ee
and note that the left hand side of (\ref{eqn:QEIgendiff}) is that quantity smeared and restricted to the diagonal $\tau=\tau'$. We have
\bea
\langle \nord{\rho}_{\omega_0} \circ \gamma \rangle_\omega(f^2)&=&\int_{-\infty}^{\infty} \mathrm{d}\tau\, \mathrm{d}\tau'\, f(\tau) f(\tau') \delta(\tau-\tau')
\langle \nord{\rho}_{\omega_0}^{\mathrm{split}} \rangle_\omega\ (\tau,\tau') \nonumber\\
&=& \int_0^\infty \frac{\mathrm{d}\alpha}{\pi} \int_{-\infty}^{\infty} \mathrm{d}\tau\, \mathrm{d}\tau'\, f(\tau) f(\tau') e^{i\alpha(\tau-\tau')}  
\langle \nord{\rho}_{\omega_0}^{\mathrm{split}} \rangle_\omega\ (\tau,\tau') \nonumber\\
&=&\int_0^\infty \frac{\mathrm{d}\alpha}{\pi} \int_{-\infty}^{\infty} \mathrm{d}\tau\, \mathrm{d}\tau'\, 
\langle \rho^{\mathrm{split}}(\tau,\tau')\rangle_\omega \overline{f_\alpha(\tau)} f_\alpha(\tau') \nonumber\\
&&-\int_0^\infty \frac{\mathrm{d}\alpha}{\pi} \int_{-\infty}^{\infty} \mathrm{d}\tau\, \mathrm{d}\tau'\, 
\langle \rho^{\mathrm{split}}(\tau,\tau')\rangle_{\omega_0} \overline{f_\alpha(\tau)} f_\alpha(\tau') \,,
\eea
where we used the properties of delta functions and the fact that
$\langle \nord{\rho}_{\omega_0}^{\mathrm{split}} \rangle_\omega$ is symmetric. Now (\ref{eqn:QEIgendiff}) is obtained by noting that $\langle \rho^{\mathrm{split}}(\tau,\tau')\rangle_\omega$ is of positive type. Additionally, the right-hand side of (\ref{eqn:QEIgendiff}) is finite, because the integrand decays rapidly as $\alpha \to \infty$. For the more technical parts of the poof see Theorem 4.1 of \cite{Fewster:1999gj}, and for a slight generalization of the argument see Lemma 4 of \cite{Fewster:2018pey}. 

The first general absolute QEIs were derived by Fewster and Smith~\cite{Fewster:2007rh}. They include worldline and worldvolume bounds on the expectation value of any operator that is a renormalized square which is symmetric and of positive type. This includes the energy density (QWEI) as well as the null energy density (QNEI) of the minimally coupled scalar field. Their argument was an adaptation of the proof of (\ref{eqn:QEIgendiff}), replacing $\omega_0(x,x')$ by $\tilde{H}(x,x')=(1/2) [H(x,x')+H(x',x)+iE(x,x')]$, where $\frac12iE(x,x')$ is the antisymmetric part of the two-point function. The bound depends only on the local geometry and in the worldline case it takes the form
\be
\langle \nord{\rho} \circ \gamma \rangle_\omega(f^2) \geq -\int_0^\infty \frac{\mathrm{d}\alpha}{\pi} \left[(f\otimes f) (\gamma\otimes\gamma)^* (Q \otimes Q\ \tilde{H}_k ) \right]^{\wedge}(-\alpha,\alpha) \,,
\ee
where $(\gamma\otimes\gamma)^*$ is the distributional pull-back, $^{\wedge}$ indicates the Fourier transform of the term in square brackets and $Q$ an operator of positive type. For the massless minimally coupled Klein-Gordon field, Kontou and Olum used this inequality to derive the first QEIs with explicitly curvature dependent bounds, both for QWEI \cite{Kontou:2014tha} and QNEI \cite{Kontou:2015yha}. These results build on previous work on Minkowski space with background potentials \cite{Graham:2002yr, Kontou:2014eka}.

\medskip

Just like for the pointwise energy conditions, the situation is considerably more complicated for non-minimally coupled fields. To date no absolute QEIs have been derived for the non-minimally coupled scalar field. The root cause is that the operator $T_{ab}^{\mathrm{split}}$ (see (\ref{eqn:pointsplitop})) is not of positive type or symmetric. Fewster and Osterbrink~\cite{Fewster:2007ec} were the first to derive a difference QWEI for the non-minimally coupled scalar field which is non-trivial (in the sense of Def.~\ref{def:nontrivial}), but they also proved that it is impossible to have state-independent bounds in this case. Recently Fewster and Kontou derived a state-dependent, but non-trivial difference QSEI for the massive non-minimally coupled scalar field, both for worldlines and worldvolumes \cite{Fewster:2018pey}.

In terms of the Maxwell and Proca fields, where the classical theories obey the WEC (see Sec.~\ref{sec:pointwise}), difference worldline QEIs have been derived on general globally hyperbolic spacetimes. Fewster and Pfenning~\cite{Fewster:2003ey, Pfenning:2001wx} proved state-independent QWEI bounds using techniques similar to \cite{Fewster:1999gj} and explicit examples for static and Rindler spacetimes were worked out. QEIs for pre-metric electrodynamics have also been derived \cite{Fewster:2017mtt}. 

The situation is quite different for the Dirac field. The Dirac equation for a fermion field $\psi$ is
\be
(i\gamma^\mu \partial_\mu -m) \psi=0 \,,
\ee
where $m$ is the mass of the field. The $\gamma$ matrices have the defining property $\{ \gamma^\mu , \gamma^\nu \}=2 \eta^{\mu \nu}$. The Lagrangian density is
\be
\int \mathrm{d}^4 x \sqrt{-g} \left\{ \frac{i}{2} \left( \bar{\psi} \gamma^\mu \nabla_\mu \psi-(\nabla_\mu \bar{\psi})\nabla^\mu \psi\right)+m \bar{\psi} \psi \right\} \,,
\ee
where $\bar{\psi}=\psi^\dagger \gamma^0$.  The Belinfante-Rosenfeld \cite{belinfante1940current, rosenfeld1940tenseur} stress tensor is found by symmetrizing 
\be
T_{\mu \nu}=\frac{i}{2} \left(\bar{\psi} \gamma_\mu (\nabla_\nu \psi)-(\nabla_\mu \bar{\psi}) \gamma_\nu \psi \right) \,.
\ee
Then the energy density is
\be
T_{00}=\frac{i}{2} [\psi^\dagger \dot{\psi}-\dot{\psi}^\dagger \psi] \,.
\ee
This quantity is unbounded both from above and below and thus it is often stated in the literature that the classical Dirac field violates the WEC. The reason is that viewing the Dirac equation as a single-particle equation can lead to arbitrary negative energies in the ``Dirac sea". Dirac's original view was that antiparticles are ``holes" in this negative energy sea. 

However, from the perspective of QFT the Dirac field has a positive Hamiltonian in Minkowski space as a result of normal ordering using fermionic creation and annihilation operators. Vollick \cite{Vollick:1998sk} was the first to show that explicitly constructed negative energy density states of the Dirac field obey a QWEI in two dimensions. Fewster and Verch \cite{Fewster:2001js} proved that Dirac states that obey the Hadamard condition have bounded normal ordered energy density. Later Dawson, Fewster and Mistry established explicit bounds for difference QWEIs in both Minkowski space \cite{Fewster:2003zn} and globally hyperbolic spacetimes of arbitrary curvature \cite{Dawson:2006py}. Smith \cite{Smith:2007pp} derived the first absolute QWEI for the Dirac field. Rarita-Schwinger fields of spin $\frac32$ are also known to satisfy a QWEI in Minkowski space \cite{PhysRevD.69.064008}.

\medskip

As demonstrated by the previous examples, QEIs have been derived for free fields in quite general settings. However, their status in interacting models is less clear. The earliest interacting examples are conformal field theories (CFTs) in two-dimensional Minkowski space, for which a state-independent QEI was established by Fewster and Hollands under rather general assumptions \cite{FewsterHollands2005}. Moreover, they also showed that the lower bound in the QEI is sharp. The arguments that they presented rely heavily on the analysis of Flanagan almost a decade earlier, who proved a sharp, state-independent QEI for a massless scalar field in two-dimensional Minkowski space \cite{Flanagan:PhysRevD.56.4922}. The CFTs covered by the result of
\cite{FewsterHollands2005} include the unitary, positive energy minimal models and rational CFTs. E.g., for a chiral CFT, the result gives
\be
\int_{-\infty}^{\infty} \mathrm{d}t\, G(t) T_{00}(\gamma(t))\ \ge -\frac{c}{12\pi}\int_{-\infty}^{\infty}  \mathrm{d}t\, \left(\frac{\mathrm{d}}{\mathrm{d}t}\sqrt{G(t)}\right)^2\,,
\ee
where $\gamma$ is the wordline of an inertial observer, $G$ is any positive Schwartz function and $c$ is the central charge of the theory. Note that the same optimal lower bound also applies to averages along spacelike geodesics, although in general dimensions there is usually no QEI for energy densities which are only averaged in a spacelike direction. It is worth noting that Flanagan extended his analysis of the massless free field to curved two-dimensionsal spacetimes using conformal transformations \cite{Flanagan:PhysRevD.66.104007}, but we are not aware of a similar generalization of the results of \cite{FewsterHollands2005} to
curved spacetimes.

More recently, Bostelmann, Cadamuro, and Fewster established a state-independent QEI of the massive Ising model in two-dimensional Minkowski space \cite{Bostelmann:2013mxa}. The bound that they found is interesting, because single-particle states can have negative energies in this model, but due to the interactions the total energy density does not scale linearly with the number of particles (see also Sec.~\ref{sub:QEIover}). While the existence of a state-independent QEI in this model is important, the methods used to derive it may not generalize to other interacting models.

It is worth mentioning that Bostelmann and Cadamuro also established QEIs for another class of integrable QFTs, including the sinh-Gordon model, but only for the one-particle states of such models \cite{Bostelmann:2015lxa}. Another relevant result derives the existence of quantum inequalities for operators arising from an operator product expansion~\cite{Bostelmann:2008hq}. This result applies to theories that satisfy a microscopic phase space condition, which is a condition that restricts the number of degrees of freedom of the QFT in bounded regions of phase space (see loc.cit.~for definitions). Whether the operators in the quantum inequalities include components of the stress tensor for interacting theories is presently unclear, mainly due to the difficulties in defining such a tensor at this level of generality. However, one may expect that this result applies to the energy density of perturbative bosonic QFTs.

\begin{table}[h!]
	\begin{center}
		\begin{tabular}{|c|c|c|c|} \hline
			\textbf{Theory} &  QWEI & QNEI & QSEI \\ \hline
			massless Klein-Gordon ($\xi=0$, $m=0$) & \cmark  & \cmark & \cmark \\ 
			massive Klein-Gordon ($\xi=0$, $m>0$) & \cmark & \cmark & \xmark \\
			Klein-Gordon ($\xi \neq 0$, $m\ge0$) & \xmark & \xmark  & \xmark  \\
			Maxwell/Proca ($m \geq 0$)   & \cmark & ? & ? \\
			Dirac & \cmark & ? & ?  \\
			CFT & \cmark & \cmark & * \\
			Ising & \cmark & ? & ? \\
			\hline
		\end{tabular}
	\end{center}	
	
	\caption{Validity of QEIs. \cmark\ indicates that there is a state-independent QEI for all Hadamard states. \xmark\ indicates that it is impossible to have a state-independent QEI, but state-dependent ones have been derived, while ? indicates that no such QEI has been derived yet but it is perhaps possible. * The CFTs are unitary, positive energy CFTs with a stress tensor in two-dimensional Minkowski space. Due to the dimension, the QSEI cannot be formulated.\label{tab:examples2}}
\end{table}

Table \ref{tab:examples2} summarizes the results on QEIs described. To conclude this section we will now comment on QEIs where the averaging is done over null geodesics. Ford and Roman~\cite{Ford:1994bj} established the first result for a null averaged QEI for the massless Klein-Gordon field in two-dimensional Minkowski space, using a Lorentzian sampling function. Their bound has the form
\be
\frac{\lambda_0}{\pi} \int_{-\infty}^\infty \mathrm{d}\lambda \frac{\langle \nord{T_{ab}}_{\omega_0}k^a k^b\rangle_\omega (\gamma(\lambda))}{\lambda^2+\lambda_0^2} \geq -\frac{1}{16 \pi \lambda_0^2} \,,
\ee
for all $\lambda_0>0$, an affinely parametrized null geodesic $\gamma$, $k^a=(d\gamma/d\lambda)$ where $\lambda$ is the affine parameter. It is easy to see that this inequality is invariant under affine reparametrization.

However, in four spacetime dimensions the situation is rather different. Fewster and Roman~\cite{Fewster:2002ne} showed by an explicit counter-example that QEIs along null geodesics do not exist in four-dimensional Minkowski space for the massless minimally coupled scalar field. In their analysis they used quantum states that are superpositions of the vacuum and multimode two-particle states. In these states the excited modes are those whose three-momenta lie in a cone centered around a chosen null vector. In the limit of a sequence of those states the three-momenta become arbitrarily large and the null energy density averaged over a null geodesic can become arbitrarily negative. As they note, their result could be understood as a consequence of the fact that, in any algebraic QFT in Minkowski space of dimension $d > 2$ obeying a set of conditions there are no (non-trivial) observables localized on any finite null line segment (for more details regarding this argument see endnotes 45 and 46 in \cite{Fewster:2002ne}). Perhaps this could provide a general argument against the existence of null worldline QEIs in more general QFTs, but no such rigorous result exists to this day. 

Recently, Freivogel and Krommydas~\cite{Freivogel:2018gxj} have argued that the counter-example of \cite{Fewster:2002ne} is problematic, because it creates a state from the vacuum using excitations with arbitrary momenta. If the momenta are restricted to be below the UV cutoff of the theory, then there is a finite lower bound, namely
\be
\int_{-\infty}^\infty \mathrm{d}\lambda f(\lambda) \langle \nord{T_{ab}}_{\omega_0}k^a k^b\rangle_\omega (\gamma(\lambda)) \geq - \frac{B}{G_N} \int_{-\infty}^\infty \mathrm{d}\lambda \frac{f'(\lambda)^2}{f(\lambda)} \,,
\ee
where $B$ is constant of order $1$ and $G_N\sim 1/N$ is the effective Newton's constant with $N$ being the number of fields in a given QFT. \cite{Freivogel:2018gxj} have conjectured that this bound should always hold for any smooth and non-negative smearing function $f$ as long as the cutoff rule is satisfied. The bound has been verified for an induced gravity on a brane in AdS/CFT by Leichenauer and Levine~\cite{Leichenauer:2018tnq}.

\subsection{Quantum energy inequalities and uncertainty relations}
\label{subs:uncertainty}

The literature on QEIs has emphasized since early on that their derivations do not assume or use any quantum mechanical time-energy uncertainty relation \cite{Ford:1997fa,Fewster:1998xn}. Such an uncertainty relation has only been invoked to argue that QEIs prevent macroscopic violations of the second law of thermodynamics \cite{Ford:1978qya}. Nevertheless, the restrictions that QWEIs impose on the duration and magnitude of negative energy densities have sometimes been described as being of uncertainty principle-type \cite{Ford:1997fa}, or reminiscent of the uncertainty principle in quantum mechanics \cite{Fewster:1998xn}. In this section we want to comment in more detail on the meaning and interpretation of QWEIs.

Unfortunately, the time parameter in quantum mechanics is not an observable, so the formulation of a time-energy uncertainty relation is more problematic than the uncertainty relation for e.g.~position and momentum. Perhaps the closest analog was
derived in \cite{Mandelstam1991} for isolated systems with a (constant) Hamiltonian operator $H$. If $a$ is any other self-adjoint operator and $\sigma_a:=\sqrt{\omega(a^2)-\omega(a)^2}$ its standard deviation in a (pure or mixed) state $\omega$, then
\be
\label{eqn:UP}
\sigma_H\sigma_a\ge \frac12\left|\omega([a,H])\right|=\frac12\left|\partial_t\omega(a)\right|\,.
\ee
A QWEI differs in several fundamental ways from this time-energy uncertainty relation. Firstly, a QWEI estimates expectation values rather than standard deviations. In particular, a QWEI is sensitive to the sign of $T_{ab}$, whereas (\ref{eqn:UP}) is independent of the sign of $H$. Secondly, the right-hand side of (\ref{eqn:UP}) is positive rather than negative, so the uncertainty principle requires deviations away from classical behaviour (namely non-sharp values), whereas QWEIs restrict deviations away from classical behaviour (WEC violations). We see that the direct analogy between QWEIs and the time-energy uncertainty relation does not run very deep.

\medskip

Perhaps a better analogy arises when comparing QEIs with suitable bounds on non-classical phenomena in quantum mechanics. Consider e.g.~a one-dimensional particle of mass $m>0$ in a finite square potential well with potential $-V<0$ in the interval $[-L,L]$ and with vanishing potential on $|x|>L>0$. If the particle is classical and has a total energy $E\in \left(-V,0\right)$, then it must be inside the potential well and remain there forever. For a quantum particle, however, the wave function $\psi(x)$ also extends outside the well, although the amplitude $|\psi(x)|$ falls off exponentially with the distance from the well. If the particle is in a stationary state of energy $E\in(-V,0)$, then the probability of finding the particle a distance $r>0$ outside the well is
\be
\label{eqn:Pviolation}
P(|x|\ge L+r)=\int_{|x|\ge L+r} \mathrm{d}x\, |\psi(x)|^2 = C(V,E,L,m)e^{-\sqrt{8m|E|}r}\,,
\ee
where $C(V,E,L,m)$ is a positive constant (cf.~\cite{Messiah} Vol.~I, Ch.~3, Sec.~6). Like a QEI, this equality is a bound which restricts non-classical behaviour.

Whereas QWEIs involve a trade-off between the magnitude and \emph{duration} of WEC violations, (\ref{eqn:Pviolation}) entails a trade-off between the magnitude $r$ and the \emph{probability} of the violation of classical behaviour. That this analogy also has its limitations follows from investigations of the probability distribution for measurement outcomes of a suitably averaged energy density $T^{\mathrm{ren}}_{00}(f)$ \cite{FewsterFordRoman2010,FewsterFordRoman2012,FewsterHollands2019}. E.g., for conformal theories in $n=2$ there is a large class of states and weighting functions for which the probability distribution turns out to be a shifted gamma distribution \cite{FewsterHollands2019}, supported on an interval $[-b,\infty)$, where the number $-b$ is the infimum of the spectrum of $T^{\mathrm{ren}}_{00}(f)$. This value is also the sharpest possible value in a QWEI \cite{FewsterFordRoman2010} and it may be small in absolute value, depending on the weighting function $f$ amongst other things. However, this does not mean that the probability of measuring a negative result is also small. On the contrary, when $\omega(T^{\mathrm{ren}}_{00}(f))=0$, then large positive measurement outcomes, which may occur with low probability, must be balanced by small negative outcomes with higher probabilities. For a massless free field in $n=2$ Minkowski space and $f$ a Gaussian, the probability of finding a negative value for $T^{\mathrm{ren}}_{00}(f)$ has been computed at ca. $84$\% \cite{FewsterFordRoman2010}.

The compensation of negative values by positive ones along a timelike curve is also restricted in terms of temporal separation. This has been studied in detail by Ford and Roman \cite{Ford:1999qv} who named it \textit{quantum interest}, using a financial analogy: you are allowed to ``borrow'' negative energy, but you have to ``repay'' it in the form of positive energy with some ``interest''. The longer the duration of the ``loan'' the larger the ``interest''. In other words, the positive energy pulse has to overcompensate the negative one by an amount which increases with the time separation of the pulses. Ford and Roman proved this for massless quantized scalar fields in Minkowski space and Pretorius \cite{Pretorius:1999nx} generalized their results to include massive fields. 
A different method of proof was introduced by Fewster and Teo \cite{Fewster:1999kr}, who used a reformulation of QEIs for massless scalar fields as a eigenvalue problem. (See also \cite{Abreu:2008dh,Teo:2002ne,Fewster2017QEIs}.)

The basic idea of \cite{Ford:1999qv} can be described using delta pulses with energy density
\be
\label{eqn:rhopulse}
\rho(t) = \frac{|\Delta E|}{A} [-\delta (t) + (1 + \epsilon) \delta(t-T)] \,,
\ee
where $A$ is the collecting area of the flux with energy $|\Delta E|$ and $T>0$ is the separation time. A QEI in four-dimensional Minkowski space has the general form  (see (\ref{eqn:QEIminflat}) in Sec.~\ref{sec:QEI})
\be
\int_{-\infty}^{\infty} \mathrm{d}t\, \rho(t) f_{t_0}^2(t)  \geq -\frac{C}{t_0^4} \,,
\ee
where the constant $C$ depends on the choice of the compactly supported function $f$ and $t_0>0$ is a scaling parameter as in (\ref{eqn:fscale}). Using (\ref{eqn:rhopulse}) and rearranging we get
\be
|\Delta E| T^3 \leq \frac{CA }{\beta^3\left(f(0)^2-(1+\epsilon)f(\beta)^2\right)}\,,
\ee
where $\beta := t_0/T$. For exactly compensating pulses, $\epsilon \to 0$, the right-hand side vanishes as $\beta \to \infty$. That means that either $|\Delta E| =0$ or $T=0$. In all other cases we need to have quantum interest, $\epsilon>0$. 

\medskip

The quantum mechanical inequality that provides the best analogy to a QWEI is, in our view, the stability condition that the Hamiltonian operator $H$ is bounded from below. This analogy arises in an elementary way from the analysis of \cite{Fewster:1998pu}, who derived a QWEI for free scalar fields in Minkowski space in $n\ge 2$ dimensions. Extending this analysis to $n=1$, where spacetime reduces to just time and the linear scalar field reduces to a harmonic oscillator, we find a normal ordered Hamiltonian $H=:T_{00}:$, which is no longer a density on dimensional grounds. $H$ is constant in time and $H\ge0$, which shows that the harmonic oscillator is stable. In this context it is interesting to note that the proof of Theorem \ref{the:EpsteinGlaserJaffe} does not work in a quantum mechanical setting.

A more elaborate analysis of the relation between stability and QWEIs brought to light close connections \cite{FewsterVerch2003}. For general dynamical systems in a static spacetime, the existence of a class of states satisfying a QWEI entails the existence of so-called passive states. These are thermodynamical equilibrium states in the sense that one cannot extract energy from them by a cyclic process. This connects QWEIs, as a mesoscopic stability criterion, with the macroscopic stability of thermodynamic equilibrium states. In the case of free scalar fields the connection is even tighter, because both conditions are essentially equivalent (up to some technicalities) to the Hadamard condition, which may be viewed as a microscopic stability criterion.

For an isolated quantum system we can add any constant to $H$ without affecting the dynamics, so for stable systems we may always assume $H\ge0$. This is exactly the bound that we found in the elementary comparison with the harmonic oscillator above. However, just like for classical systems, the gravitational setting introduces absolute energies rather than just energy differences and it can therefore lead to a different lower bound in a natural way. Negative energies can arise e.g.~due to non-trivial metrics, compactification of spatial dimensions \cite{Banach_1979,PhysRevD.58.024007} or due to boundary conditions as in the Casimir effect \cite{Casimir:1948dh,PhysRevD.51.4277} (see Sec.~\ref{subs:AANEC}). Such negative energy densities can arise for stationary states too, so the energy density can remain negative for arbitrary amounts of time. However, Bekenstein has recently argued that for the Casimir effect the total energy, including that of the plates, is non-negative \cite{Bekenstein:2013ztp}. This suggests that negative energy densities can only persist for long times when the system is not isolated or the spacetime is not asymptotically flat. It would be interesting to see if this idea can be given a precise mathematical formulation and a proof. One of the challenges for this is the fact that stationary states of quantum systems are global objects, whereas the renormalization of the stress tensor is a local procedure. This complicates the computation of local energy densities in stationary states in general stationary spacetimes. As a first result in this area, one of us (KS) \cite{Sanders:2016bwc} identified general sufficient conditions on the background geometry of a static spacetime to guarantee that the expectation value of the Wick square of a free field is non-negative in all stationary states.

\section{Classical vs.~quantum energy inequalities}
\label{sec:classvsquant}

A quantum field is often expected to satisfy a QEI when the corresponding classical field satisfies a corresponding pointwise energy condition (cf.~the definitions in Sec.~\ref{sub:QEIover}). In the inverse direction, QEIs can be used to derive classical energy conditions by taking a classical limit. In this section we will consider the relationships between classical energy conditions and QEIs in more detail, thereby connecting the discussions of Sections \ref{sec:pointwise} and \ref{sec:QEI}. An important role will be played by averaged energy conditions, which combine aspects of pointwise energy conditions and QEIs.

\subsection{Averaged energy conditions}\label{subs:ANEC}

An averaged energy condition is a condition that requires a contraction of the stress tensor of a theory, averaged over a suitable spacetime region, to be non-negative. The averaging is similar to QEIs, but the strict lower bound $0$ is reminiscent of classical physics and of the pointwise energy conditions. Although one could consider many spacetime regions to average over, our discussion will only consider averages along causal geodesics, which is also the most common choice in the literature (see also the comments in Sec.~2.2 of \cite{Curiel:2017}).

To formulate averaged energy conditions we let $T_{ab}$ denote the stress tensor of the theory, or, in the case of a quantum theory, the expectation value of the renormalized stress tensor in a suitable state, cf.~Section \ref{sec:QEI}. We let $\gamma$ denote a smooth causal geodesic with affine parameter $\tau$ and tangent vector
$\dot{\gamma}^a$.

The \emph{averaged weak energy condition} (AWEC) requires that for all inextendible timelike geodesics
\be 
\int_{-\infty}^{\infty} \mathrm{d}\tau\ T_{ab}(\gamma(\tau))\dot{\gamma}^a(\tau)\dot{\gamma}^b(\tau)\ge0
\ee
when the integral is absolutely convergent. This condition is weaker than the pointwise WEC and a classical example where the WEC fails but the AWEC holds is the non-minimally coupled Klein-Gordon field for certain values of the coupling constant $\xi$
(cf.~\cite{Fewster:2006ti} and the discussion below). The AWEC is also motivated by the ideas that violations of the WEC, caused e.g.~by quantum effects, should be small in magnitude or duration and that on average they should disappear in the long time limit. For suitable theories this can be derived as a limiting case of a QWEI. However, the AWEC is known to fail under some circumstances, such as in the Casimir effect, where the boundary conditions enforce a static negative energy density, cf.~the discussion in Sections \ref{subs:uncertainty} and \ref{subs:AANEC}.

Similarly, the \emph{averaged strong energy condition} (ASEC) in a spacetime of dimension $n>2$ requires that 
\be 
\int_{-\infty}^{\infty} \mathrm{d}\tau\ \left(T_{ab}(\gamma(\tau))-\frac{T}{n-2}g_{ab}\right)\dot{\gamma}^a(\tau)\dot{\gamma}^b(\tau)\ge0
\ee
for all inextendible timelike geodesics when the integral over the effective energy density is absolutely convergent. The ASEC is weaker than the pointwise SEC. It may be recovered as a long time limit of a QSEI \cite{Fewster:2018pey} for free scalar fields only in the massless minimally coupled case (cf.~Sec.~\ref{subs:AANEC}). In view of Table \ref{tab:examples} one would expect that the quantum Maxwell and Proca fields could provide further examples.

For the \emph{averaged null energy condition} (ANEC) we consider null geodesics $\gamma$ with affine parameter $\lambda$. The ANEC then requires that
\be 
\int_{-\infty}^{\infty} \mathrm{d}\lambda\ T_{ab}(\gamma(\lambda))\dot{\gamma}^a(\lambda)\dot{\gamma}^b(\lambda)\ge 0
\ee
for all inextendible null geodesics when the integral is absolutely convergent. An even weaker condition is the \emph{achronal averaged null energy condition} (AANEC), which requires the same inequality only for all achronal inextendible null geodesics. (Achronal means that no two points on the geodesic can be connected by a timelike curve.) The range of validity of the ANEC and the AANEC will be discussed in detail below and in Section \ref{subs:AANEC}.

We have formulated the averaged energy conditions only in the case that the averaging integral is absolutely convergent. It is possible to extend these conditions to integrals that are not absolutely convergent, but such an extension is not unique. One finds several extensions in the literature, which are often motivated by the mathematical methods used in proofs. E.g., 
Verch \cite{Verch:1999nt} uses the condition
\be
\label{eqn:ANECext1}
\liminf_{r_{\pm\to\infty}}\int_{-r_-}^{r_+} \mathrm{d}\tau\ T_{ab}(\gamma(\tau))\dot{\gamma}^a(\tau)\dot{\gamma}^b(\tau)\ge0
\ee
for the AWEC and the ANEC. Borde \cite{Borde:1987b} replaces the ANEC and the ASEC by weaker, but more complicated conditions, which also involve integration over bounded intervals. A different extension for the ANEC was used by Wald and Yurtsever \cite{Wald:1991xn}, namely
\be
\label{eqn:ANECext2}
\liminf_{\lambda_0 \to\infty}\int_{-\infty}^{\infty} \mathrm{d}\lambda\ T_{ab}(\gamma(\lambda))\dot{\gamma}^a(\lambda)\dot{\gamma}^b(\lambda)
c\left(\frac{\lambda}{\lambda_0}\right)^2\ge0
\ee
for all compactly supported real-valued functions $c\in C^1_0(\mathbb{R},\mathbb{R})$ such that $(1+k^2)^{1+\delta}|\hat{c}(k)|$ is bounded for some $\delta>0$. This condition does not imply (\ref{eqn:ANECext1}), but both extend the ANEC for absolutely converging integrals. Similar extensions of the AWEC and ANEC, but with $c\in C^2_0(\mathbb{R},\mathbb{R})$, appeared in the work of Fewster and Osterbrink \cite{Fewster:2006ti}.

The averaged energy conditions are summarized in Table \ref{tab:AECs}, where we focus on absolutely convergent integrals. From now on we will suppress the argument $\tau$ in the integrands as well as the fact that the stress tensor is to be evaluated along $\gamma$.

\begin{table}[h!]
	\begin{center}
		\begin{tabular}{|c|c|c|c|} \hline
			\textbf{Condition} &\textbf{Inequality} &\textbf{Inextendible Geodesics} \\ \hline
			\makecell{AWEC} & $\int \mathrm{d}\tau\ T_{ab} \dot{\gamma}^a\dot{\gamma}^b\ge0$
			&timelike\\ \hline
			\makecell{ASEC} & $\int \mathrm{d}\tau\ (T_{ab} -\frac{T}{n-2}g_{ab})\dot{\gamma}^a\dot{\gamma}^b\geq 0$ 
			&timelike\\ \hline
			\makecell{ANEC} & $\int \mathrm{d}\lambda\ T_{ab} \dot{\gamma}^a\dot{\gamma}^b\ge0$
			&null\\ \hline
			\makecell{AANEC} & $\int \mathrm{d}\lambda\ T_{ab} \dot{\gamma}^a\dot{\gamma}^b\ge0$
			&achronal null\\ \hline
		\end{tabular}
		\caption{The main averaged energy conditions summarized. Here $\tau$ is an affine parameter along the geodesic $\gamma$ of the stated class and the integrands are evaluated on $\gamma$.}
		\label{tab:AECs}
	\end{center}
\end{table}

In the remainder of this section we will consider the validity of averaged energy conditions for classical non-minimally coupled free scalar fields. We have seen in Section \ref{sub:examples} that these fields can violate all the main pointwise energy conditions in suitable on-shell configurations. However, the averaged energy conditions are harder to violate. In the analysis of \cite{Barcelo:2000zf} for massless scalar fields, e.g., the ANEC can only be violated for $\xi\in(0,\frac14)$ and only when the field takes trans-Planckian values somewhere. We saw in Section \ref{sub:examples} that such values may be considered unphysical and that they are potentially problematic for the initial value formulation.

For $\xi\in[0,\frac14]$, \cite{Fewster:2006ti} considers test field configurations\footnote{Note that \cite{Fewster:2006ti} calls a configuration on-shell when it is a test field configuration in our sense, cf.~Section \ref{subs:onshell}.}
and investigates lower bounds for averaged energy densities along causal geodesics and over regions of spacetime. E.g., for a causal geodesic $\gamma$ with affine parameter $\tau$ and for any $f\in C_0^2(\mathbb{R},\mathbb{R})$ they show that (in our sign conventions)
\be
\int_{-\infty}^{\infty}\ \mathrm{d}\tau\, T_{ab}\dot{\gamma}^a\dot{\gamma}^bf^2 \ge
-2\xi\int_{-\infty}^{\infty} \mathrm{d}\tau \bigg( (f')^2\phi^2-\frac12R_{ab}\dot{\gamma}^a\dot{\gamma}^bf^2\phi^2 -\frac14(1-4\xi) R
\dot{\gamma}_a\dot{\gamma}^af^2\phi^2\bigg) \,.
\ee
Although the lower bound may be negative, it no longer depends on the derivatives of $\phi$ and it can be estimated in terms of $R_{ab}$, $\phi$ and $f$. In particular, if the geodesic is complete and $R_{ab}=0$, then \cite{Fewster:2006ti} show that the ANEC and AWEC hold as long as the integral $\int T_{ab}\dot{\gamma}^a\dot{\gamma}^bf^2\mathrm{d}\tau$ exists. Similar conclusions hold if the null or timelike convergence condition $R_{ab}\dot{\gamma}^a\dot{\gamma}^b\ge0$ is satisfied.

Note that the null, resp.~timelike, convergence condition for the metric is equivalent to the NEC, resp.~the SEC, for the background matter that generates the metric. Thus we see that the non-minimally coupled scalar field with $\xi\in[0,\frac14]$ can only violate the ANEC when the background matter violates the NEC and it can only violate the AWEC when the background matter violates the SEC. Arguably, the violation of the ANEC or AWEC by the non-minimally coupled field is then of less interest than the violation of the NEC or the SEC by the background matter. This provides some (incomplete) justification to consider averaged energy conditions especially in curved spacetimes satisfying, e.g., the NEC, or the ANEC (see also Sec.~\ref{subs:AANEC}).

\medskip

To avoid the need to put strong restrictions on the background spacetime, such as $R_{ab}=0$, one can consider on-shell configurations. For suitable parameters $m$ and $\xi$, analogous results to those of \cite{Fewster:2006ti} can be obtained, where the curvature of the manifold is now dynamically coupled to the field $\phi$. Taking the trace of Einstein's equation and substituting the Klein-Gordon equation we find that
\be
\label{eqn:eetrace}
\frac{1}{8\pi}\left(1-8\pi\xi\left(1-\frac{\xi}{\xi_c}\right)\phi^2\right)R=\left(1-\frac{\xi}{\xi_c}\right)(\nabla\phi)^2 \left(\frac{n}{n-2}-\frac{\xi}{\xi_c}\right)m^2\phi^2\,.
\ee
Assuming $8\pi\xi\phi^2<1$ this allows us to express $R$ entirely in terms of $\phi$. We can then use the Klein-Gordon equation and the expression for $R$ to rewrite the effective stress tensor (\ref{eqn:nonminimal2}) as
\bea
T^{\mathrm{eff}}_{ab}&=&(1-2\xi)\nabla_a\phi\nabla_b\phi-2\xi\phi\nabla_a\nabla_b\phi\nonumber\\
&&-\frac12g_{ab}\frac{1}{1-8\pi\xi\left(1-\frac{\xi}{\xi_c}\right)\phi^2}
\left[ \left(1-4\xi-8\pi\xi\left(1-\frac{\xi}{\xi_c}\right)\phi^2\right)(\nabla\phi)^2\right.\nonumber\\
&&\quad \left.+\left(1-4\xi-8\pi\xi\left(1-\frac{(n-3)\xi}{(n-1)\xi_c}\right)\phi^2\right)m^2\phi^2\right]\,.\label{eqn:Teffonshell}
\eea
If $\gamma$ is a null geodesic and $\lambda$ an affine parameter, then we find
\be
T_{ab}\dot{\gamma}^a\dot{\gamma}^b=(1-8\pi\xi\phi^2)^{-1}\left((1-2\xi)(\partial_{\lambda}\phi)^2
-2\xi\phi\partial_{\lambda}^2\phi\right)\,,\nonumber
\ee
where we recall that $T_{ab}=(1-8\pi\xi\phi^2)^{-1}T^{\mathrm{eff}}_{ab}$. When $\xi\in[0,\frac12]$ then the first term is
non-negative. If $\lambda$ ranges over an open interval $I\subset\mathbb{R}$ and $f\in C_0^1(I)$, then we can use an integration by parts and
$0\le (1-8\pi\xi\phi^2)^{-2}16\pi\xi f^2\phi^2(\partial_{\lambda}\phi)^2+(1-8\pi\xi\phi^2)^{-1}(\partial_{\lambda}(f\phi))^2$ to estimate
\bea
\int_I \mathrm{d}\lambda \, f^2 T_{ab}\dot{\gamma}^a\dot{\gamma}^b  &\ge& 
-2\xi \int_I \mathrm{d}\lambda\, f^2 (1-8\pi\xi\phi^2)^{-1}\phi\partial_{\lambda}^2\phi\ \nonumber\\
&\ge&-2\xi \int_I \mathrm{d}\lambda \, (1-8\pi\xi\phi^2)^{-1}\phi^2(\partial_{\lambda}f)^2 \,.\label{eqn:ANECpi}
\eea
By rescaling the function $f$ to $f_{\lambda_0}$ as in (\ref{eqn:fscale}) and taking the limit $\lambda_0 \to\infty$ we can deduce as in \cite{Fewster:2006ti} that the ANEC holds on any complete null geodesic on which $8\pi\xi\phi^2<1$ and $(1-8\pi\xi\phi^2)^{-1}$ remains bounded.

If $\gamma$ is a timelike geodesic, parametrized by proper time $\tau$, then we find for $m=0$ and $\xi\in[0,\xi_c]$
\bea
T_{ab}\dot{\gamma}^a\dot{\gamma}^b&\ge&-2\xi(1-8\pi\xi\phi^2)^{-1}\phi(\dot{\gamma}^a\nabla_a)^2\phi\,,\nonumber
\eea
where we added the first term in (\ref{eqn:Teffonshell}) to the $(\nabla\phi)^2$ term to find a non-negative contribution.
Using a computation as in (\ref{eqn:ANECpi}) we can estimate 
\be
\int_I  \mathrm{d}\tau \, f^2 T_{ab}\dot{\gamma}^a\dot{\gamma}^b
\ge-2\xi \int_I \mathrm{d}\tau \, (1-8\pi\xi\phi^2)^{-1}\phi^2(\partial_\tau f)^2 \nonumber
\ee
when $\tau$ ranges over an open interval $I\subset\mathbb{R}$ and $f\in C_0^1(I)$. Rescaling the function $f$ to $f_{\tau_0}$ as in (\ref{eqn:fscale}) and taking the limit $\tau_0 \to\infty$ now yields the AWEC on any complete timelike geodesic on which $8\pi\xi\phi^2<1$ and $(1-8\pi\xi\phi^2)^{-1}\phi^2$ remains bounded.

Similarly \cite{Brown:2018hym} considered lower bounds on the averaged weighted EED. For $\xi \in [0,2\xi_c]$ and averaging a test field configuration over a timelike geodesic $\gamma$ they showed that
\be
\label{eqn:cline}
\int \mathrm{d}\tau \left(  T_{ab}\dot{\gamma}^a \dot{\gamma}^b-\frac{1}{n-2} T  \right)f^2 \geq - \int \mathrm{d}\tau \bigg\{ \frac{1-2\xi}{n-2} m^2 f^2+\xi \bigg(2( \partial_\tau f )^2
+R_{ab} \dot{\gamma}^a \dot{\gamma}^b f^2 - \frac{2\xi}{n-2} Rf^2 \bigg) \bigg\}\phi^2  \,.
\ee
When $\xi \in [0,\xi_c]$ and $8\pi\xi\phi^2<1$ they then showed that for on-shell configurations
\be
\label{eqn:cline2}
\int \mathrm{d}\tau \left(  T_{ab}\dot{\gamma}^a \dot{\gamma}^b-\frac{1}{n-2} T  \right)f^2 \geq 
- \int \mathrm{d}\tau \bigg\{ \frac{1-2\xi}{n-2} \frac{m^2 f^2}{1-8\pi\xi\phi^2}+2\xi \bigg( \partial_\tau \frac{f}{\sqrt{1-8\pi\xi\phi^2}} \bigg)^2 \bigg\}\phi^2 \,.
\ee
Writing $\phi_{\max}:=\sup_\gamma |\phi|$ for the maximum field amplitude of the field along $\gamma$ and assuming $8\pi\xi\phi_{\max}^2<1$ we see that the first term on the right-hand side of (\ref{eqn:cline2}) is bounded from below by
\be
\label{expr}
- \phi_{\max}^2 \frac{1-2\xi}{n-2} \frac{m^2}{1-8\pi\xi\phi_{\max}^2}  \int \mathrm{d}\tau f^2\,.
\ee
Rescaling the function $f$ to $f_{\tau_0}$ as in (\ref{eqn:fscale}) we note that $\int \mathrm{d}\tau f_{\tau_0}^2$ remains constant and positive, so in the limit $\tau_0\to \infty$ we find that a positive multiple of the prefactor in (\ref{expr}) contributes to the lower bound of the ASEC integral. In the massive case this means that we do not obtain the ASEC, not even in the minimally coupled case. However, for $m=0$ the ASEC is recovered \cite{Brown:2018hym}.

Although the derivation of the ASEC fails in the massive case, one can derive an analogous inequality by averaging over Minkowski space for bounded test field configurations. In terms of an inertial time coordinate $t$, a compactly supported test function $f$ and the rescaled functions $f_r(x):=r^{-n/2}f(x/r)$ with $r>0$, \cite{Brown:2018hym} have shown that
\be
\liminf_{r \to \infty} \int \mathrm{d}vol \,   \left(  T_{tt}-\frac{1}{n-2} T  \right)  \, f_r^2 \geq 0 \,.
\ee
(Note that equation (50) loc.cit.~is missing absolute values under the integral, but this does not affect the subsequent derivation.)

\subsection{Quantization and classical limits}
\label{subs:claslimits}

In this section we will consider the general relationship between QEIs and the corresponding pointwise energy conditions. As a guiding principle we may use the analogy with the relationship between the motion of a classical particle and a quantum particle. The path of a classical particle can be found by minimizing an action over all possible paths. Following Feynman one may ask how the particle knows which path will give the minimum, since it can only follow one path. His resolution was to view the classical particle as an approximate description of a quantum particle, which takes all possible paths simultaneously, but with different probability amplitudes. In our setting this viewpoint provides a motivation to study classical energy conditions also in off-shell configurations.

Let us consider any classical field theory which has an associated quantum field theory. Amongst the states of the quantum theory we wish to distinguish off-shell, test field and on-shell states, in analogy to the classical theory. By an \emph{off-shell} state for the quantum theory we will mean a set of $n$-point distributions,
\begin{eqnarray}
\omega_n(x_1,\ldots,x_n)&:=&\langle \Phi(x_1)\cdots\Phi(x_n)\rangle_{\omega}\,,\nonumber
\end{eqnarray}
which need not satisfy an equation of motion, but which do satisfy suitable commutation relations, singularity restrictions (e.g.~the microlocal spectrum condition \cite{brunetti1996}) and the necessary positivity conditions to give a Hilbert space representation. For \emph{test field} configurations, we require that the $n$-point distributions also satisfy the equation of motion of the quantum field in every variable. For free fields this is not problematic, but for non-linear equations the formulation of such equations of motion is a non-trivial matter, which may involve perturbation theory, regularization and renormalization. Similarly, for \emph{on-shell} states we will also impose the semiclassical Einstein equation
\be
\label{eqn:SEE}
G_{ab}=8\pi G \langle T^{ren}_{ab} \rangle_\omega \,.
\ee
A solution of (\ref{eqn:SEE}) includes both a state $\omega$ and a metric that must simultaneously satisfy it. This equation involves similar, or perhaps even more severe, renormalization difficulties as test field configurations and only a few solutions are known, all for spacetimes with a high degree of symmetry. However, let us suppose for the moment, for the sake of argument, that the theory we consider admits reasonable definitions of test field and on-shell states. What relations should we then expect between the pointwise or averaged energy conditions and the QEIs?

Let us consider a QEI for a quantum field and the corresponding pointwise energy condition for the corresponding classical field. We can try to derive this poinwise energy condition by taking classical limits on both sides of the QEI. In order to avoid problems with the lower bound we will assume that the QEI is state-independent. We would then expect, quite naively, that a state-independent off-shell, resp.~test field, resp.~on-shell QEI implies the corresponding off-shell, resp.~test field, resp.~on-shell pointwise energy condition for the classical field. 

In the converse direction, however, we appeal to our guiding principle and note that every quantum state probes all configurations of the classical system, even off-shell ones. For this reason we cannot expect to derive a test field or on-shelll QEI from a pointwise 
energy condition which only holds for classical test field or on-shell configurations. Nevertheless, if the pointwise energy condition holds for all off-shell classical configurations, then we would expect the QEI to hold as well for all off-shell quantum states and then, a fortiori, also for test-field and on-shell states.

The expectations we formulated in both directions warrant a few comments. Note that classical field theory in itself does not seem to provide any particularly good reason to consider the pointwise energy conditions off-shell, because off-shell configurations are classically impossible. However, the transition to QEIs does provide a motivation to study classical off-shell configurations. On the other hand, off-shell QEIs seem to be of little physical interest and they have never been considered in the literature. Nevertheless, for minimally coupled free scalar fields they could in principle be derived using the same argument as for test field QEIs, cf.~Section \ref{sub:examplesQEI}, because the key ingredients (positivity, commutation relations and singularity condition) are all in place.

It is also interesting to note that for free classical fields the distinction in the dynamics is not of much importance, due to Propositions 
\ref{Prop_ECoffshell1} and \ref{Prop_ECoffshell2}. This means e.g.~that for scalar fields we expect a state-independent QEI to hold for all test field quantum states if and only if the classical field satisfies the corresponding pointwise energy condition for all test-field configurations.

\medskip

The expected relationship between QEIs and pointwise energy conditions as formulated above is not mathematically precise. Any attempt to make it rigorous will have to rely on the details of the quantization method in one direction and of the classical limit in the other. For free fields in off-shell or test field configurations, these details are readily available. Nevertheless, we are not aware of any general results that derive QEIs from the classical theory. The reason seems to be that the guiding principle that we used to motivate our expectations is perhaps most conveniently expressed in a path-integral formalism, which, however, is not so well-suited to a generally covariant setting. For this reason the existing derivations of QEIs do not explicitly invoke the corresponding classical energy conditions, but they proceed in a direct way from first principles and the definition of the theory. Nevertheless, the results obtained in the literature are in line with the expectations formulated above, as can be seen by comparing Table \ref{tab:examples} in Section \ref{sub:examples} and Table \ref{tab:examples2} in Section \ref{sub:examplesQEI}.

In the opposite direction, for fields in Minkowski space derivations of pointwise energy conditions from QEIs can essentially be recovered from the literature \cite{Fewster:2007ec,Fewster:2018pey}, see also the discussion in Section \ref{sub:QEIover}. For free scalar fields in general spacetimes we will now prove the following rather elementary result. (We expect a similar result for any free bosonic field, but our argument and the argument in \cite{Fewster:2007ec,Fewster:2018pey} do not immediately generalize to fermi fields.)
\begin{theorem}\label{thm:classlimit}
	Let $\phi$ be a linear scalar quantum field on a globally hyperbolic spacetime $M$ which satisfies a state-independent QEI of the following form: For every timelike geodesic $\gamma:\mathbb{R}\to M$ and every test function
	$f\in C_0^{\infty}(\mathbb{R},\mathbb{R})$ there is a number $Q(f)\in\mathbb{R}$ such that for all Hadamard two-point distributions $\omega_2$ we have
	\be
	\int_{-\infty}^{\infty} \mathrm{d}\tau \, \omega_2(T^{\mathrm{ren}}_{ab})\dot{\gamma}^a\dot{\gamma}^bf^2\  \ge -Q(f)\ .
	\ee
	Then the corresponding classical field satisfies the pointwise WEC.
\end{theorem}

\begin{proof*}
	Fix a Hadamard two-point distribution $\omega_2$ and let $\omega_1$ by any one-point distribution, i.e.~$\omega_1$ is a smooth and real-valued solution to the (Klein-Gordon) field equation. For every $r>0$ the distribution
	\be
	\omega^{(r)}_2 :=\omega_2+r\omega_1\otimes\omega_1
	\ee
	is also a Hadamard two-point distribution (it is a bi-solution of positive type and has the correct antisymmetric part and singularities).
	This means that we have for all $r>0$
	\bea
	-Q(f)&\le&\int_{-\infty}^{\infty} \mathrm{d}\tau \, \omega^{(r)}_2(T^{\mathrm{ren}}_{ab})\dot{\gamma}^a\dot{\gamma}^bf^2\ \nonumber\\
	&=&\int_{-\infty}^{\infty} \mathrm{d}\tau \, \omega_2(T^{\mathrm{ren}}_{ab})\dot{\gamma}^a\dot{\gamma}^bf^2\ 
	+r\int_{-\infty}^{\infty}  \mathrm{d}\tau \, T_{ab}[\omega_1]\dot{\gamma}^a\dot{\gamma}^bf^2\ 
	\eea
	where $T_{ab}[\omega_1]$ is the classical stress tensor of the corresponding classical field theory in the field configuration $\omega_1$. Dividing both sides by $r>0$ and taking the limit $r\to\infty$ we find
	\bea
	0&\le&\int_{-\infty}^{\infty}\mathrm{d}\tau \, T_{ab}[\omega_1]\dot{\gamma}^a\dot{\gamma}^bf^2\ \,.
	\eea
	We can choose $f\in C_0^{\infty}(\mathbb{R},\mathbb{R})$ with $\int f^2(\tau)\ \mathrm{d}\tau=1$ and define
	$f_{\tau_0}(\tau)$ as in (\ref{eqn:fscale}) for all $\tau_0>0$. This means that $f_{\tau_0}^2$ approximates a $\delta$-distribution. Substituting $f_{\tau_0}$ for $f$ in the estimate above and taking the limit $\tau_0\to 0$ we find
	\be
	0\le T_{ab}[\omega_1](\gamma(0)))\dot{\gamma}^a(0)\dot{\gamma}^b(0)\ .
	\ee
	Because we can choose the geodesic $\gamma$ to start at any given point $x\in M$ with any future-pointing timelike vector $\dot{\gamma}^a$ we find that the classical stress tensor $T_{ab}[\omega_1]$ satisfies the WEC. This proves the WEC for test field configurations $\omega_1$. Because $M$ is globally hyperbolic we see from Proposition
	\ref{Prop_ECoffshell1} that the classical theory also satisfies the WEC off-shell.
\end{proof*}

Note that the proof of Theorem \ref{thm:classlimit} works for difference QEIs as well as absolute QEIs, as long as the lower bound is state-independent. Moreover, analogous arguments would work to derive the SEC and NEC from a QSEI and a QNEI (along a timelike curve), respectively. The formulation and proof of an analogous result for averaged energy conditions is also possible, as long as divergent integrals are handled carefully, as we will discuss in the next section. Using a rather different kind of argument, Parikh and van der Schaar \cite{Parikh:2014mja} have derived the NEC along the worldsheet of a bosonic string to leading order in $\alpha'$ (the string tension).

We have seen in Section \ref{sub:examples} that the classical non-minimally coupled scalar field violates all the main pointwise energy conditions. In line with Theorem \ref{thm:classlimit} the corresponding quantum field does not admit any of the corresponding state-independent QEIs. Nevertheless, it does admit state-dependent QEIs, as we have seen in Section \ref{sub:examplesQEI}. Moreover, we have seen in Section \ref{subs:ANEC} that the classical field can satisfy averaged energy conditions like the AWEC and ANEC if we disallow trans-Planckian values and if we consider on-shell configurations or if we restrict the background metric. Whether there is a pattern in these results that can be generalized to other non-minimally coupled fields is not yet clear.

\subsection{Derivation of averaged energy conditions from quantum field theory}
\label{subs:AANEC}

Averaged energy conditions can sometimes be derived from the corresponding QEIs, even if they are state-dependent. Let us first consider the simplest case of a minimally coupled scalar quantum field in Minkowski space, where we can start from the QEI (\ref{eqn:QEIminflat}) of Section \ref{sub:examplesQEI}. Considering the scaling of the test function as in (\ref{eqn:QEIminflatscale}), rearranging and taking a limit we then obtain the AWEC in the form
\be
\label{eqn:limAWEC}
\liminf_{t_0 \to\infty}\int_{-\infty}^{\infty}\mathrm{d}t \, \langle \nord{\rho} (\gamma(t)) \rangle_{\omega}\  |f(t/t_0)|^2 \,  \geq 0\,.
\ee
For non-minimally coupled scalar fields there is no state-independent QWEI, but for $\xi\in[0,\frac14]$ the same scaling argument also works for the state-dependent QWEI in Minkowski space \cite{Fewster:2007ec}, so the AWEC (\ref{eqn:limAWEC}) still holds. Similarly, the state-dependent QSEI for non-minimally coupled scalar fields in Minkowski space gives rise to the ASEC \cite{Fewster:2018pey}, but only in the massless case.

\medskip

Although QEIs also hold in spacetimes which are curved or which have non-trivial topology or boundaries, cf.~(\ref{eqn:QEIgendiff}), the scaling argument may no longer work and indeed the corresponding averaged energy condition may not hold. For example, the AWEC fails when the Cauchy surface of Minkowski space is compactified to a torus \cite{Banach_1979} and even when only one of its dimensions is compactified, cf.~\cite{PhysRevD.58.024007} Section IIB. A particularly interesting example is the Casimir effect. In 1948 Casimir \cite{Casimir:1948dh} showed that two neutral conducting parallel plates in the vacuum attract each other, cf.~\cite{Plunien:1986ca} and \cite{Bordag:2001qi} for reviews. The effect is explained by a change in the zero-point energy of the electromagnetic field, caused by the presence of the plates. To see how this arises we consider the stress tensor of the electromagnetic field (\ref{eqn:tensorproca}) ($m=0$) along with the vanishing trace and divergence properties. If we orient the coordinate frame so that one plate is at $z=0$ while the other is at $z=L$, then the stress tensor takes the form \cite{Brown:1969na, Fewster2017QEIs}
\be
T^{\mu \nu}=\frac{C(z)}{L^4} ( \eta^{\mu \nu}-4 \hat{k}^\mu \hat{k}^\nu ) \,,
\ee
where $\hat{k}=(0,0,0,1)$ points in the $z$-direction and $C(z)$ is dimensionless and $C'(z)=0$ except possibly at the plates. Now it is easy to see that the energy density between the plates (in a finite distance) is $-C/L^4$ and the constant can be computed exactly as $C=\pi^2/720$. Even though the amount of negative energy is small compared with the positive energy required to create and measure it, it nevertheless violates all the pointwise energy conditions. More notably it also violates the AWEC, since the negative energy is constant in time (cf.~the discussion in Sec.~\ref{subs:uncertainty}). However it has be shown to obey the ANEC \cite{Graham:2005cq,Fewster:2006uf}. 

Even when the AWEC fails, one may show that in stationary spacetimes the lowest average energy along a stationary curve is obtained by the ground state, which is a kind of ``difference AWEC'' \cite{PhysRevD.51.4277,Fewster2017QEIs}. It appears to be unknown under what conditions the ground state on a stationary or static spacetime will have a non-negative energy density, reproducing the AWEC at least for stationary observers.

\medskip

Whereas the AWEC and ASEC can often be derived from a QEI along timelike curves using the limit of long time smearing, a similar procedure does  not work for the ANEC, unfortunately. The problem is that there is no known QEI along null curves in dimensions more than 2 (see Sec.~\ref{sub:examplesQEI}). However, it is possible to prove ANEC using a QNEI averaged over timelike curves. Such a proof was first presented by Fewster, Olum and Pfenning~\cite{Fewster:2006uf} for a flat spacetime with boundaries, for geodesics which stay a finite distance away from the boundary. The starting point is the QNEI for flat spacetimes derived by Fewster and Roman~\cite{Fewster:2002ne} (see also the discussion in Sec. \ref{sub:examplesQEI})
\be
\label{eqn:QNEI}
\int_{-\infty}^\infty \mathrm{d}t \langle \nord{T_{ab}}k^a k^b\rangle_\omega f_{t_0}^2
\geq - \frac{(\dot{\gamma}_a k^a)}{12\pi^2t_0^2} C \,,
\ee
where $f\in C_0^{\infty}(\mathbb{R},\mathbb{R})$, $f_{t_0}$ is defined in (\ref{eqn:fscale}) and $C=\int_{-\infty}^\infty \mathrm{d}t f''(t/t_0)^2$. Assuming uniform continuity of $\langle \nord{T_{ab}}k^a k^b\rangle_\omega$
near the null curve $\zeta$ their construction is simple: because $f$ is compactly supported, the desired integrand has compact support on a segment of the null geodesic. This segment can be approximated by a timelike geodesic segment where the QNEI (\ref{eqn:QNEI}) applies. Using a careful limiting argument, in which $t_0\to\infty$ in $f_{t_0}$ and the timelike geodesics are boosted towards the null geodesic, they derive the ANEC. A variation of this argument also applies in the absence of uniform continuity. Kontou and Olum~\cite{Kontou:2015yha, Kontou:2015bta} used a similar argument to prove AANEC in spacetimes with small curvature. 

\medskip

Of all the energy conditions the AANEC has the best chance of being valid under all physically reasonable circumstances. It is weaker than all pointwise energy conditions and the ANEC and it has no known physically reasonable counter-examples. For classical non-minimally coupled scalar fields, the ANEC can only be violated when the effective Newton's constant changes sign, cf.~Section \ref{subs:ANEC}. Even for quantum fields, where other QEIs have a negative lower bound, there are good reasons to believe that the AANEC will hold with a zero lower bound. It is not violated by the Casimir effect, unlike the AWEC \cite{Graham:2005cq,Fewster:2006uf}, and known violations of the AANEC involve Planck length distances
\cite{Flanagan:1996gw,PhysRevLett.78.2050,PhysRevD.65.084028,Garattini_2005}, where the validity of QFT in curved spacetimes and/or semiclassical quantum gravity is doubtful.

In the remainder of this section we will review the evidence that supports the idea that the AANEC should hold under all physically reasonable circumstances. More precisely, we will consider two formulations of this idea. The first formulation states
\begin{itemize}
	\item[(i)] In Minkowski space the ANEC holds for any reasonable QFT with a stress tensor which (locally) generates the translations of the theory.
\end{itemize}
Note that all null geodesics in Minkowski space are automatically achronal, so the AANEC is the same as the ANEC. Condition (i) can include interacting theories for which the stress tensor is usually introduced axiomatically. The second formulation, which applies to curved spacetimes, makes use of the Lagrangian formulation of the stress tensor in its renormalized form. It states
\begin{itemize}
	\item[(ii)] On-shell configurations can violate the AANEC only over Planck scale distances in directions transversal to the null geodesic. \end{itemize}
The slightly stronger conjecture that the AANEC holds for any self-consistent solution in semiclassical gravity has been advanced by Graham and Olum \cite{Graham:2007va}, see also \cite{penrose1993positive}. Note that for condition (ii) a detailed specification of the physically reasonable circumstances is considerably complicated by the intricacies of the semiclassical Einstein equation, which on-shell configurations are required to satisfy. Violations at Planck scale are allowed, because they probably lie outside the range of validity of the semiclassical Einstein equation. Because any Planck scale phenomena should really be treated in the context of a full theory of quantum gravity, condition (ii) suggests that the AANEC might be a consequence of a fundamental property that should hold in full quantum gravity.

\medskip

Early results on the ANEC mostly pertained to free scalar quantum fields. On the positive side, Klinkhammer showed in 1990 that the ANEC holds for a large subset of the Hadamard states in Minkowski space \cite{Klinkhammer:1991ki}. For conformally coupled scalar fields this result was extended to all asymptotically flat two-dimensional spacetimes by Yurtsever \cite{Yurtsever:1990gx}, who was the first to suggest that the ANEC could hold generally, at least in asymptotically flat spacetimes. Wald and Yurtsever \cite{Wald:1991xn} proved that the AANEC holds for all Hadamard states of a massless scalar field in any globally hyperbolic, two-dimensional spacetime along any complete, achronal null geodesic. Results for massive fields and in four-dimensional Minkowski space were also obtained in \cite{Wald:1991xn}, but for a smaller class of states. For the physically more interesting case of electromagnetism, Folacci  \cite{PhysRevD.46.2726} extended the analysis of \cite{Klinkhammer:1991ki} to the vector potential $A_a$ along null geodesics in Minkowski space.

Alongside these positive results, a number of counter-examples were discovered. Klinkhammer \cite{Klinkhammer:1991ki} also showed that negative results can be obtained by integrating along curves which are not geodesics, or geodesics which are not achronal in a flat spacetime with a compactified dimension. In four dimensions, Wald and Yurtsever \cite{Wald:1991xn} argued that the ANEC cannot hold for all states of a massless field in all spacetimes. These counter-examples helped to shape the current formulation of the AANEC, insisting on complete achronal null geodesics.

For conformally coupled quantum fields, violations of the AANEC can be generated rather easily and elegantly using conformal rescalings of the metric. Visser \cite{VISSER1995443} has used global scale transformations to argue that the AANEC is violated in conformal quantum states with a non-zero scale-anomaly for suitable choices of the renormalization scale. More recently Urban and Olum \cite{Urban:2009yt} used conformal transformations of the Minkowski metric to generate AANEC violations of a slightly different nature for the conformally coupled massless scalar field. These violations do not exploit the scale anomaly, nor do they rely on a choice of the renormalization scale. It should be noted, however, that all these violations are only known for test fields, so the semiclassical Einstein equation is not imposed.

Perhaps the best support to date for the general validity of the AANEC, in the sense of formulation (ii) above, comes from Flanagan and Wald \cite{Flanagan:1996gw}. They studied perturbations of a metric and a massless scalar quantum field around Minkowski space and the vacuum state. After a detailed discussion of the range of validity of the semiclassical Einstein equation, they considered second order perturbations of the system, assuming that at first order the metric perturbation is not dominated by incoming gravitational radiation, and for the quantum field they considered a long wave length limit (compared to the Planck length). Under these circumstances they showed that the integral of the energy density along a complete achronal null geodesic, combined with a transversal averaging over several Planck areas, is non-negative. From this result they concluded that macroscopic violations of the AANEC (above the Planck scale in transversal directions) are ruled out. Violations of the ANEC may occur if no averaging in transversal directions is included, but \cite{Graham:2007va} have pointed out that it is unclear if these violations of the ANEC also violate the AANEC, i.e.~if the null geodesics are achronal for the perturbed metric.

Around the same time, Ford and Roman \cite{Ford:1995gb} studied the ANEC-integral over half complete geodesics in Schwarzschild spacetimes of dimensions two and four, focusing on the context of black hole evaporation (cf.~Section~\ref{sub:blackholes}). Although negative results can occur, they found that the magnitude and duration of such violations is bounded by a kind of quantum inequality. Based on this independent evidence they too argued that macroscopic effects of ANEC violation are ruled out.

\medskip

Let us now turn to formulation (i) of the general validity of AANEC, which pertains to Minkowski space. For general quantum fields in two-dimensional Minkowski space which have a mass gap and satisfy a few other reasonable minimal assumptions, the ANEC was proved by Verch in 1999 \cite{Verch:1999nt}. For conformal field theories in two-dimensional Minkowski space the AWEC and ANEC were established by Fewster and Hollands \cite{FewsterHollands2005}. These results suggest that the validity of the ANEC goes beyond free fields, even though the arguments cannot easily be generalized to higher dimensions.

Recent attempts to prove the ANEC for general QFTs in higher dimensional Minkowski space have drawn on various ideas and methods related to black hole thermodynamics, quantum gravity and the AdS/CFT correspondence. \cite{Kelly:2014mra} proves the ANEC for CFTs in Minkowski space, assuming that the AdS/CFT correspondence provides a consistent dual theory that includes gravity and that has good causal behaviour. The argument is elegant and fairly elementary, but unfortunately the class of theories to which it applies is difficult to determine exactly. \cite{Ishibashi:2019nby} considers the AdS/CFT correspondence with curved boundaries and finds metrics that violate the AANEC using similar methods as \cite{Urban:2009yt}.

Wall \cite{Wall:2009wi} gives a rather different line of argument, which aims to derive the ANEC from the generalized second law of thermodynamics. This is the law that states that there can be no decrease in the generalized entropy, which is the sum of the von Neumann entropy of the matter in a region $O$ plus a contribution proportional to the area of the boundary of $O$. Although the argument provides a physically interesting insight, the von Neumann entropy is an ill-defined quantity in QFT and requires regularization. Following Casini \cite{Casini_2008} one can regularize the entropy by using the relative entropy and choosing a suitable reference state. However, whether this regularization scheme, or any other scheme, can make the argument of \cite{Wall:2009wi} rigorous is not fully settled.

Considerations of entropy in black hole thermodynamics and quantum gravity have actually given rise to a new type of energy condition which is closely related to the ANEC. It is the quantum null energy condition, or QNEC, which requires that
\be
\langle T_{ab}(x) \rangle_{\omega} k^a k^b \geq S''(\omega) \,,
\ee
where $T_{ab}(x)$ generates the translations of the theory, where $S(\omega)$ is the von Neumann entropy of the state
$\omega$ in a region whose boundary contains the point $x$ and $S''(\omega)$ is a second functional derivative with respect to deformations of that region near $x$ in the null direction $k^a$. Although the naive formulation of this condition suffers from the fact that the von Neumann entropy is ill-defined, it has nevertheless been argued that under suitable circumstances the QNEC implies the ANEC \cite{Bousso:2015wca} and, conversely, that the ANEC implies this QNEC \cite{Ceyhan:2018zfg}. Using somewhat similar arguments, Faulkner et al.~have argued that the ANEC can be derived directly using modular Hamiltonians \cite{Faulkner:2016mzt}. The physical arguments that were brought forward in these papers suggest the existence of interesting relationships between the null components of the stress tensor, entanglement and the modular operators that arise naturally as a general mathematical tool in QFT. The clarification of these connections and the conditions under which they hold is of great interest and we refer to Longo \cite{Longo2020} for recent results on this.

\section{Applications}
\label{sec:applications}

Energy conditions were first introduced as assumptions in theorems in classical general relativity and this is still the main reason for their importance. The range of validity for these theorems in the context of semiclassical gravity is not fully clear yet, but significant progress has been made in the past few decades. In this Section we review the main applications of the energy conditions, with an emphasis on recent results that make use of averaged energy conditions and QEIs. 

\subsection{Singularity theorems}
\label{sub:singularity}

The notion of spacetime \textit{singularities} appeared early in the history of general relativity, but there are several ways to try and capture this idea mathematically. One of the first solutions of the Einstein equation, the Schwarzschild solution, has metric components which diverge at coordinate values $r=2M$ and $r=0$. The singularity at $r=2M$ is a coordinate singularity, i.e.~it can easily be removed by a change of coordinates. Coordinate singularities are mathematical artefacts of little physical interest, but the singularity at $r=0$ is of a different nature. To see this we note that
\be
R^{abcd} R_{abcd}=\frac{48 M^2}{r^2} \,,
\ee
diverges as $r \to 0$, where the left-hand side is a coordinate invariant function. Following this example, physical singularities were initially characterized as places where such a curvature scalar diverges.

An important question in general relativity during the first half of the 20th century was whether these singularities were due to the high level of symmetry of the models studied. For example, can an inhomogeneous universe avoid the primordial singularity predicted in Friedmann-Robertson-Walker models? Singularity theorems answer this question, but to do this they need to use a notion of singularity that also works in less symmetric spacetimes and the characterization in terms of curvature scalars is not suitable for that purpose.

Indeed, the previous characterization is both ill-defined and incomplete. In general it is impossible to talk about a ``place" when the spacetime metric is not defined. In some highly symmetric cases the singular part of spacetime can be represented as a boundary, but this is not true in more general models. Additionally, there are spacetime singularities that cannot be described by the divergence of curvature scalars. A simple example \cite{Wald:1983ky} is that of a Minkowski space where we have removed the points with angular coordinate $0<\phi<\phi_0$ and then connected the points with $\phi=0$ to those with $\phi=\phi_0$, see Figure \ref{fig:cone}. All of the points with $\phi=0$ and $r>0$ can be defined in the new metric, but not the point with $\phi=0$ and $r=0$ (conical singularity). However, because the Riemann tensor is zero everywhere else, there is no curvature divergence.

\begin{figure}[h!]
	\center
	\includegraphics[scale=0.5]{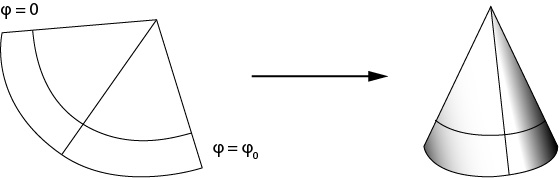}
	\caption{``Wedge" removal of Minkowski space.}
	\label{fig:cone}
\end{figure}

Then, how can we define singularities? A better way is utilising the fact that geodesics cannot be extended in the singular parts of spacetimes, so they have a finite ``length" (proper time in the case of timelike geodesics, affine parameter in the null case). Formally,

\begin{definition}
	A spacetime is singular if it possesses at least one incomplete and inextendible geodesic.  
\end{definition}

This definition is by far the most general one, but it does not give us information about the nature of the singularity and whether curvature scalars diverge in its neighbourhood. Nevertheless, it allowed the proof of the first general results that predicted the existence of singularities as features of generic spacetimes, subject only to energy conditions and topological assumptions. These are the seminal \textit{singularity theorems} of Hawking and Penrose.  

\medskip

Before we discuss particular results we present the general structure of most singularity theorems, following the extended review by Senovilla \cite{Senovilla:2014}. The common pattern helps to organize the conditions for the existence of singularities in general spacetimes. It can be stated as the following template:

\begin{theoremtemplate}
	If the spacetime satisfies 
	\begin{enumerate}
		\item 
		a boundary or initial condition
		\item 
		an energy condition
		\item 
		a causality condition
	\end{enumerate}
	then it contains at least one incomplete causal geodesic. 
\end{theoremtemplate}
The initial or boundary condition is the existence of a trapped bound. That means there is a region of spacetime which has its future (or its past) initially contained within a compact spatial region of decreasing size. The energy condition, loosely interpreted as the attractiveness of gravity, then establishes a focussing effect for geodesics. To avoid closed timelike curves, a causality condition is necessary. 

Raychaudhuri \cite{Raychaudhuri:1953yv} and Komar \cite{komar1956necessity} were the first to establish singularity theorems, but their results require a stress tensor of dust or perfect-fluid form. The first singularity theorem for a general spacetime was proved by Penrose in 1965 \cite{Penrose_prl:1965} and introduced the important concept of a trapped surface. 

\begin{definition}
	A trapped surface is a co-dimension two spacelike submanifold which has two null normals with negative expansions. 
\end{definition}

For a more detailed analysis of trapped surfaces, their types and generalizations see \cite{Senovilla:2013haf}. To understand the concept of a trapped surface it is useful to consider the case of a Schwarzschild black hole. In the simple picture of spherical stellar collapse leading to the formation of a black hole, an event horizon forms before the trapped surface (often called apparent horizon) and asymptotically the two become the same surface. Inside the trapped surface we then have the strange effect of both inward-pointing and outward-pointing light rays moving inwards, see Figure \ref{fig:light}.

\begin{figure}[h!]
	\center
	\includegraphics[scale=0.5]{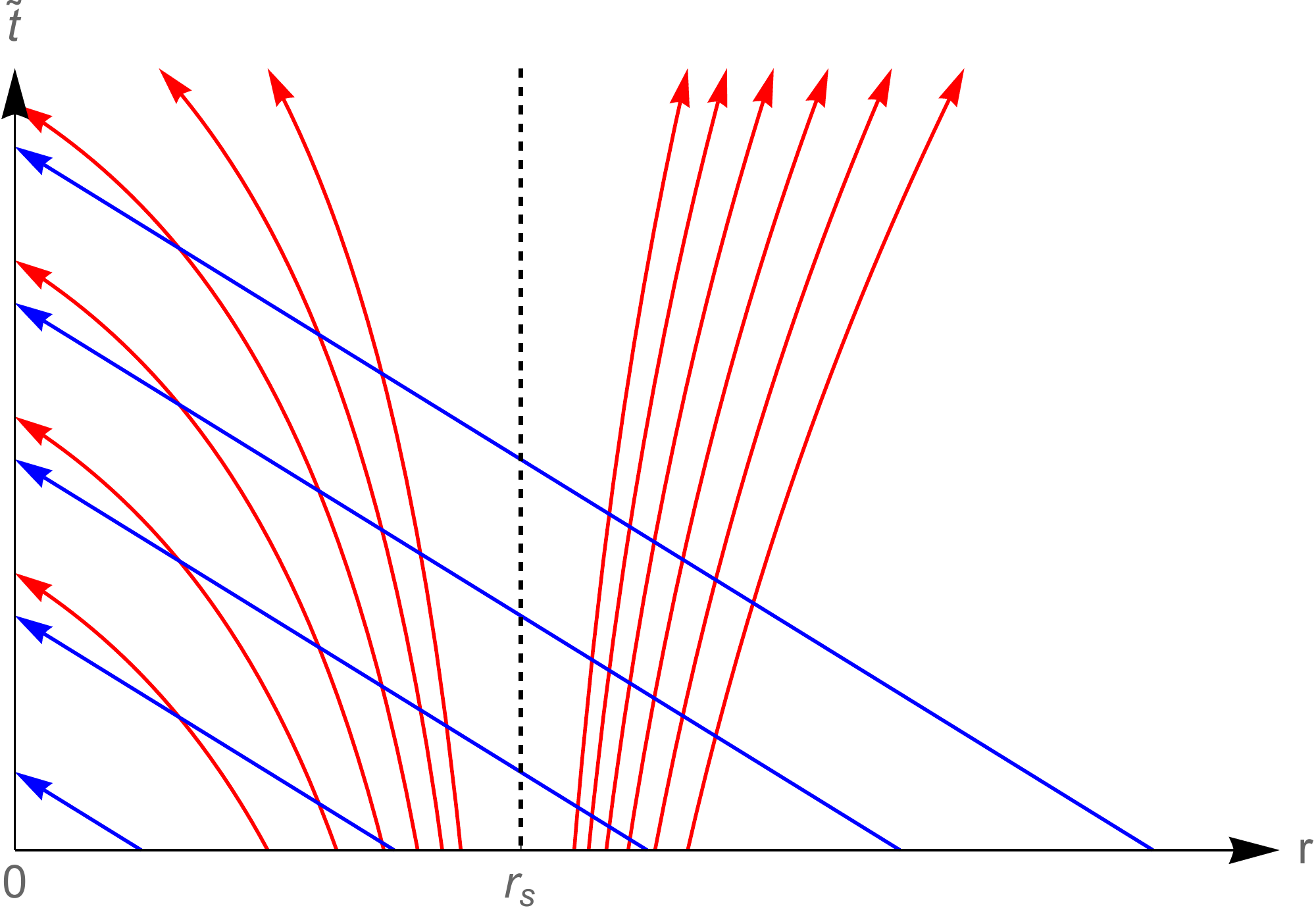}
	\caption[]{Spacetime diagram of light rays in Schwarzschild spacetime (in color online). We use Eddington-Finkelstein coordinates with 
		
		\begin{minipage}{\linewidth}
			\begin{equation*}
			-\left( 1-\frac{2M}{r} \right) \mathrm{d}v^2+2\mathrm{d}v \mathrm{d}r=0 \,,
			\end{equation*}
		\end{minipage}
		for radial null geodesics and we set $\tilde{t}=v-r$. For $v=$const. we have the ingoing and for $v \neq$ const. the outgoing light rays. However, inside the horizon $r<r_s \equiv 2M$ both ingoing (blue) and outgoing (red) radial null geodesics are directed towards $r=0$.}
	\label{fig:light}
\end{figure}

After establishing the concept of a trapped surface Penrose proceeded to prove the following theorem:

\begin{theorem}
	Suppose that:
	\begin{enumerate}
		\item 
		there is a compact achronal smooth trapped surface
		\item 
		the null converge condition (\ref{eqn:nullconvergencecondition}) holds
		\item 
		$M$ is globally hyperbolic with non-compact Cauchy hypersurfaces
	\end{enumerate}
	then $M$ is future null geodesically incomplete.
\end{theorem}

Note that the null convergence condition is the geometric form of the NEC. Just a year later, Hawking  \cite{Hawking:1966sx} proved a similar singularity theorem for timelike geodesics, using the SEC, or rather its geometric form, the timelike convergence condition:

\begin{theorem}
	Suppose that:
	\begin{enumerate}
		\item 
		$\Sigma$ is a Cauchy surface in $M$ and the timelike geodesic congruence emanating orthogonal to $\Sigma$ has a negative initial expansion 
		\item 
		the timelike converge condition (\ref{eqn:timelikeconvergencecondition}) holds
		\item 
		$\Sigma$ is compact and the strong causality condition holds
	\end{enumerate}
	then $M$ is future timelike geodesically incomplete.
\end{theorem}

It was later shown by Hawking \cite{hawking1967occurrence} that the theorem can be proven without the strong causality condition assumption. Since Hawking's theorem was intended for cosmological applications it is worth mentioning that it works similarly for positive initial expansion and past geodesic incompleteness (primordial singularity). In 1969 Hawking and Penrose co-authored another paper \cite{Hawking:1969sw} which included the proof of a theorem for causal geodesic incompleteness, weakening some of their previous causality assumptions.

\medskip

Many of the proofs of singularity theorems use the Raychaudhuri equation (\ref{eqn:raytime},\ref{eqn:raynull}). We can sketch the proof of Hawking's singularity theorem using the Raychaudhuri equation for the (irrotational) congruence emanating normally from $\Sigma$
\be
\frac{\mathrm{d} \theta}{\mathrm{d}\tau} =-R_{ab} t^a t^b-\sigma_{ab} \sigma^{ab}-\frac{\theta^2}{3}  \,,
\ee
where we have also assumed we are in a four-dimensional spacetime. If the timelike convergence condition holds 
\be
\frac{\mathrm{d}\theta}{\mathrm{d}\tau}+\frac{\theta^2}{3} \leq 0 \,,
\ee
so $\theta^{-1} \geq (\theta^{-1}|_{\Sigma}+\tau/3)$ where $\theta|_{\Sigma}$ is the initial expansion on $\Sigma$. For negative initial expansion, $\theta$ diverges after finite time no larger than $3/|\theta|_{\Sigma}|$. 

Now we need the well-known property of timelike geodesics as maximizing the proper time (cf.~\cite{ONeill} for details on the following concepts). A focal point $\gamma(s)$ on a geodesic $\gamma$ in a geodesic congruence that emanates normally from a surface $\Sigma$, is a point beyond which length is no longer extremized. This means that $\lim_{\tau \to s^-}\theta=-\infty$, or, if $\gamma(s)$ lies to the past of $\Sigma$,
$\lim_{\tau\to s^+}\theta=\infty$. Some versions of singularity theorems, such as the Hawking-Penrose ones, have used pairs of conjugate points instead. We may think of these as points $\gamma(a)$ and $\gamma(b)$ on a geodesic $\gamma$ where length is no longer extremized beyond each conjugate point.

Continuing the argument for Hawking's singularity theorem, according to the Raychaudhuri equation, for every timelike geodesic family emanating normally from $\Sigma$, $\theta \to -\infty$ after a finite time if the geodesics extend far enough. In particular, every future complete $\gamma$ must develop a focal point.

At this stage we need the topological assumptions and in particular the existence of a compact Cauchy surface. This implies (see e.g.~\cite{galloway_1986} for a rigorous proof) the existence of a future inextendible timelike geodesic emanating (necessarily normally) from $\Sigma$ and which is length-maximizing between $\Sigma$ and each of its points. Now suppose that the geodesic is future complete. Then it would be a future complete timelike geodesic containing no focal points to $\Sigma$, which leads to a contradiction with the Raychaudhuri equation. The contradiction is avoided by concluding that the geodesic is future incomplete.

The proof of Penrose's singularity theorem follows a similar logic. For rigorous proofs of the theorems, apart from the original manuscripts, see \cite{HawkingEllis:1973} and \cite{Senovilla:2014gza}. 

\medskip

O'Neill \cite{ONeill} presents somewhat different proofs, avoiding the Raychaudhuri equation but instead making use of index form methods. Note that the following discussion is for timelike geodesics and that the analysis for null geodesics is similar, using the action integral instead of the length functional. The \emph{index form} $I[V]$ is the second derivative of the length functional w.r.t.~a
variation of the geodesic $\gamma$. If $\gamma_{\alpha}$ is a one-parameter family of piecewise smooth constant speed curves joining a fixed spacelike hypersurface $\Sigma$ to a fixed point on the geodesic $\gamma$ and if $\gamma_0=\gamma$ and
$V:=\frac{\mathrm{d}}{\mathrm{d}\alpha}\gamma_{\alpha}|_{\alpha=0}$, then 
\be \label{eq:indexform}
I[V]:=\left.\frac{\mathrm{d}^2}{\mathrm{d}\alpha^2}L[\gamma_{\alpha}]\right|_{\alpha=0}=\int_0^{\tau_0} \mathrm{d}\tau \, \left( -\frac{\mathrm{D}V^a}{\mathrm{d}\tau}\frac{\mathrm{D}V_a}{\mathrm{d}\tau} + R_{cdab} t^c V^d t^a V^b \right)\,- K_{ab}V^a V^b|_{\gamma(0)} \,,
\ee
where $K_{ab}$ is the extrinsic curvature tensor of $\Sigma$ (see \cite{ONeill} for a derivation). The idea is to use the second derivative test to investigate if the geodesic $\gamma$ maximizes length locally. The second derivative of the length is the index form, so there exists a focal point $\gamma(s)$ for some $s \in (0,\tau_0]$ if and only if $I[V]\geq 0$ for some $V$.

Now any piecewise smooth function $f$ obeying $f(0)=1$, $f(\tau_0)=0$ determines a continuous piecewise smooth variation field
\begin{equation}
V^a = f v^a.
\end{equation} 
where $v^a$ is a unit spacelike vector tangent to $\Sigma$ at $\gamma(0)$ and parallel transported. Then we have
\begin{equation}
I[V] = \int_0^{\tau_0} \mathrm{d}\tau \, \left( -\dot{f}^2 + f^2 R_{cdab} t^c v^d t^a v^b \right)\,  - K_{ab}v^a v^b|_{\gamma(0)}. 
\end{equation}
Considering an orthonormal basis $e_{\alpha}$ ($\alpha=0,\ldots,n-1$) for the tangent space to $\Sigma$ in which $e_0^a=t^a$ and summing yields
\begin{equation}
\label{eqn:index2}
\sum_{i=1}^{n-1} I[f e_i] =  \int_0^{\tau_0}\mathrm{d}\tau  \, \left( -(n-1)\dot{f}^2 + f^2 R_{ab} t^a t^b  \right)\, - \theta|_{\gamma(0)} \,,
\end{equation}
where the expansion $\theta \equiv K= g^{ab} K_{ab}$ . 
If the right-hand side is non-negative for some such $f$, then the same must be true of at least one of the terms on the left, and it follows that there is a focal point in $(0,\tau_0]$. 

Now we can see an alternative way of proving singularity theorems. If $\theta|_{\gamma(0)}<0$, $\tau_0 \ge (n-1)/|\theta|_{\gamma(0)}|$, and the Ricci tensor obeys the timelike convergence condition (\ref{eqn:timelikeconvergencecondition}), then there is a focal point to $\Sigma$ along $\gamma$. This result is immediate from (\ref{eqn:index2}) using the function $f(\tau)=1-\tau/\tau_0$.

\medskip

In both lines of argument, the use of energy conditions is central in proving the singularity theorems and the applicability of the theorem is determined to a large extent by the range of validity of the energy condition. As we have seen in Sections \ref{sec:pointwise} and \ref{sec:QEI}, pointwise energy conditions, such as those used in the original singularity theorems, are easily violated by some classical and all quantum fields. Ford \cite{Ford:2003qt} in his 2003 review discussed the possibility of circumventing the singularity theorems with quantum fields. One early example is the model considered by Fulling and Parker \cite{parker1973quantized} in 1973. They constructed a quantum state which violates the SEC and the closed universe can bounce at a finite curvature scale. Thus the need to prove singularity theorems with weaker energy conditions was evident almost at the same time the first singularity theorems were proven. (There are also multiple versions of singularity theorems with weakened causality assumptions but this is not the focus of this review).

In 1977, Tipler \cite{Tipler:1978zz} showed that instead of the easily violated SEC the Hawking singularity theorem can be proven assuming the WEC and the ASEC. During the same year \cite{tipler1978general} he analysed a general form of the Ricatti inequality 
\be
\frac{\mathrm{d} \theta}{\mathrm{d}\tau} \leq -R_{ab} \xi^a \xi^b-\frac{\theta^2}{n-r}  \,,
\ee
where $\xi^a$ is a casual vector and $r=2$, resp.~$r=1$, for timelike, resp.~null, congruences and $\tau$ is the proper time, resp.~an affine parameter. He proved that focal and conjugate points are inevitable, assuming the integral of the contracted Ricci Tensor is non-negative,
\be
\liminf_{\tau_1\to -\infty, \tau_2\to \infty} \int_{\tau_1}^{\tau_2} \mathrm{d}\tau\, R_{ab} \xi^a \xi^b  \geq 0 \,,
\ee
similar to averaged energy conditions, cf.~(\ref{eqn:ANECext1}). 

Roman and Borde \cite{Roman:1988vv, Roman:1986tp, Borde:1987a, Borde:1987b} improved and extended this method, proving theorems for averaged energy conditions holding on either half-complete or complete geodesics. A few years before Roman and Borde's work, Chicone and Ehrlich \cite{chicone1980line} used instead index form methods and the ASEC or ANEC (or an analogous Riemannian condition) to prove the existence of conjugate points in Riemannian and Lorentzian manifolds. 

\medskip

The averaged energy conditions used in the works cited above are not derived from a QFT in the ways discussed in Section \ref{sec:classvsquant}. In order to prove singularity theorems which are obeyed (semiclassically) by quantum fields, the energy restriction of the theorem should be a QEI or a condition derived from a QEI. This possibility was first indirectly addressed by Fewster and Galloway \cite{Fewster:2010gm}, who introduced the following assumption inspired by QEIs (for simplicity we have set the smooth function $r_0(\tau)=0$ here):
\be 
\label{eqn:fewgal}
\int_{-\infty}^\infty \mathrm{d}\tau \, R_{ab} t^a t^b f^2 \geq -||| f |||^2 \,,
\ee
where $f \in C^\infty_0 (\mathbb{R},\mathbb{R})$ is a fixed test function and 
\be
|||f|||^2=\sum_{\ell=0}^L Q_\ell ||f^{(\ell)}||^2 \,,
\ee
is the Sobolev norm, where the $Q_\ell$'s are positive constants and $|| \cdot||$ the usual $L^2$ norm. They proved singularity theorems assuming (\ref{eqn:fewgal}) and that there exist $c > 0$, $\tau_0 > 0$ and $h \in C^\infty (\mathbb{R},\mathbb{R})$ obeying $\supp \, h \in [-\tau_0,\infty)$ and $h(\tau) = e^{-c\tau/s}$ on $[0,\infty)$, for which
\be
\theta\big|_{\Sigma} \leq -\int_{-\tau_0}^0 \mathrm{d}\tau\, R_{ab} t^a t^b h^2  -|||h|||^2-\frac{c}{2} \,.
\ee
It is interesting to note that the dependence on $R_{ab} t^a t^b$  is only through its values before the initial contraction is measured. It might seem strange at first that positive values require larger initial contraction, but combining this equation with the averaged energy condition (\ref{eqn:fewgal}) we see that large positive energies in the present allow large negative values in the future, an effect known as quantum interest, cf.~Section \ref{subs:uncertainty}.

Using the same method Brown, Fewster and Kontou \cite{Brown:2018hym} managed to proved a singularity theorem for timelike geodesic incompleteness for the classical massive free scalar field, even though this field violates the SEC. Recently Fewster and Kontou \cite{Fewster:2019bjg} showed that singularity theorems with energy conditions of the form (\ref{eqn:fewgal}) can also be proven using the index form method described earlier. The method applies immediately to singularity theorems with weakened energy conditions just by making a judicious choice for the function $f$ in (\ref{eqn:index2}). For example, using an exponential function one gets the Fewster-Galloway theorem \cite{Fewster:2010gm}. Additionally, they made quantitative estimates of the initial contraction required in order to guarantee timelike or null geodesic incompleteness. 

So far there hasn't been proven a singularity theorem with an energy condition derived directly from proven QEIs. In the timelike case the relevant QEI is the QSEI. No QSEIs were derived until the recent work of Fewster and Kontou \cite{Fewster2017QEIs} on the non-minimally coupled scalar field. The corresponding singularity theorem is currently work in progress \cite{Fewster:unpublished}.

\medskip

The null case presents greater challenges since, as we discussed in Section \ref{sec:QEI}, there are no known QEIs along null geodesics in four-dimensions. Even if the AANEC on complete geodesics is not violated, one may wonder whether violations on half geodesics suffice to circumvent Penrose's singularity theorem. This issue was investigated in detail for Schwarzschild black holes by Ford and Roman \cite{Ford:1995gb}. They showed that the magnitude and duration of any violations of the ANEC or AWEC on a half geodesic must satisfy a certain quantum inequality. The results also suggest that this quantum inequality prevents the quantum theory from circumventing Penrose's singularity theorem, so despite the evaporation process taking place at the horizon, a singularity should still form.

Finally, when using any QEI as an assumption to a singularity theorem it is necessary to use the semiclassical Einstein equation
(\ref{eqn:SEE}), which connects the classical curvature with the renormalized stress tensor (see Sec.~\ref{subs:AANEC}). We are not aware of any singularity theorems for solutions of the semiclassical Einstein equation in the literature.

\subsection{Black holes}
\label{sub:blackholes}

When an asymptotically flat spacetime $M$ has a future singularity, i.e.~a future incomplete causal geodesic, then it is generally expected that the singularity is hidden behind an event horizon, so it is inside a black hole. More precisely, if $\mathscr{I}^+$ denotes future null infinity (see e.g.~\cite{wald1984general} for an explanation of this terminology), then the singularity should not be visible near $\mathscr{I}^+$, so it should not arise in $J^-(\mathscr{I}^+)$. This expectation is called the \emph{Cosmic Censorship Conjecture}. Each connected component of the region $M\setminus J^-(\mathscr{I}^+)$ is called a \emph{black hole}. The boundary $H:=\partial J^-(\mathscr{I}^+)\cap M$ of this region is the \emph{event horizon}, cf.~Figure \ref{fig:asymptoticflatness}. In this section we will see what information the energy conditions can provide about black holes and their event horizons. For an application of QWEIs to the region outside a static compact object see \cite{Marecki:2005zf}.

If $\Sigma$ is a Cauchy surface for $M$, then the spacelike surface $\Sigma\cap H$ divides $\Sigma$ into regions inside and  outside the black hole(s), so it must be orientable. For a four-dimensional stationary black hole, if $\Sigma\cap H$ is compact and if the DEC holds, Hawking has shown that each connected component of $\Sigma\cap H$ must be homeomorphic to a sphere \cite{hawking1972}. A generalization to higher dimensions was established by Galloway and Schoen \cite{Galloway:2005mf} and improved by Galloway \cite{Galloway2008}. Assuming the DEC they showed that each connected component of $\Sigma\cap H$ (and, more generally, any outermost outer trapped surface, see loc.cit.~for definitions), endowed with the induced metric, must be of positive Yamabe type, i.e.~it can be endowed with a Riemannian metric of constant positive scalar curvature. In four dimensions this result can be combined with the Gauss-Bonnet Theorem to reproduce Hawking's result of a spherical horizon.

\begin{figure}[h!]
	\centering
	\includegraphics[width=.8\textwidth,viewport=0cm 22.6cm 7.1cm 28cm,clip]{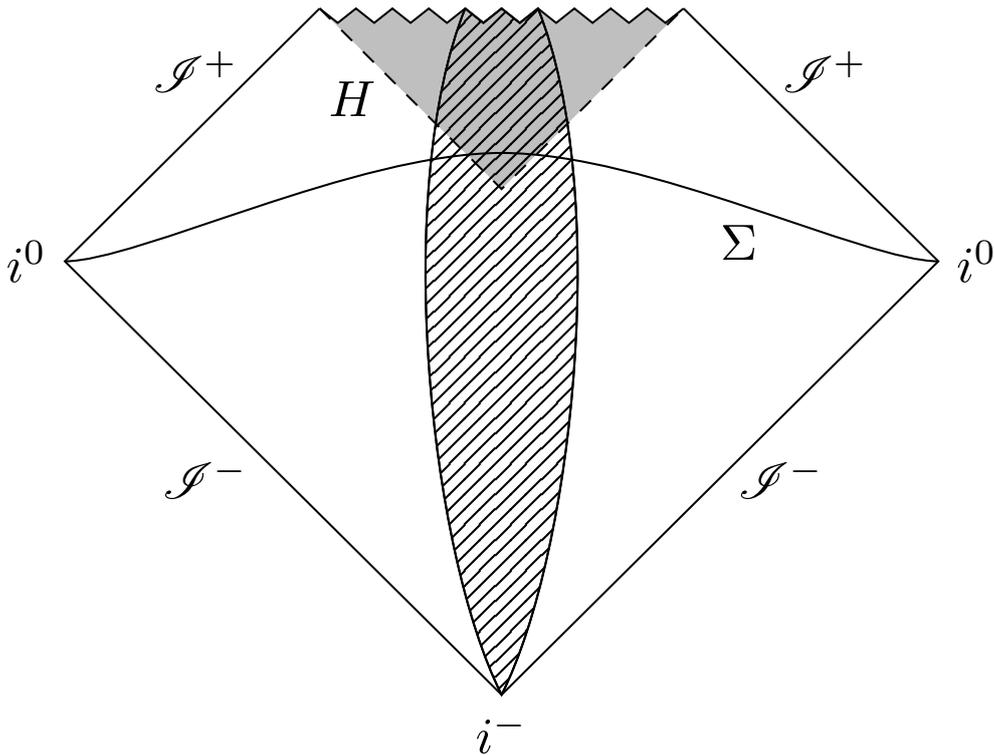}
	\caption{Penrose diagram of an asymptotically flat spacetime with a black hole. Here $\mathscr{I}^{\pm}$ denote future and past null infinity, $i^0$ is spacelike infinity and $i^-$ is past infinity. The black hole region is indicated in gray with event horizon $H$. $\Sigma$ is a Cauchy surface and the shaded region indicates matter collapsing to a black hole.}
	\label{fig:asymptoticflatness}
\end{figure}

\medskip

Much weaker energy conditions lead to a result on the area of $\Sigma\cap H$ in the induced metric. Assuming the NEC, Hawking 
\cite{hawking1972} showed that the area of a black hole cannot decrease in time, a result known as the Area Theorem, or the second law of black hole thermodynamics \cite{bardeen1973,waldQFT}. More precisely, if $\Sigma'$ is another Cauchy surface for $M$ with $\Sigma'\subset I^+(\Sigma)$, then $\Sigma'\cap H$ lies to the future of $\Sigma\cap H$. When the NEC holds, the area of
of $\Sigma'\cap H$ must be larger than that of $\Sigma\cap H$. Because black holes cannot bifurcate this theorem also holds separately for each connected component of the black hole region. The Area Theorem can be strengthened by weakening the assumption of the NEC to an averaged energy condition. Lesourd \cite{Lesourd:2017} has recently shown that it suffices to assume
a damped ANEC, which imposes a restriction on a weighted average over half null geodesics whereby the contribution near the horizon is emphasized.

\medskip

We now turn to the topology of the interior of a black hole. If the ANEC holds, the topology of the interior cannot be probed by an observer outside the black hole. This follows from the \emph{Topological Censorship Theorem} of Friedman, Schleich and Witt 
\cite{Friedman:1993ty}. In this theorem we may think of the timelike curve $\gamma_0$ as a the worldline of an observer and the causal curve $\gamma$ as the worldline of a probe or signal that is sent out and received by the observer.

\begin{theorem}\label{thm:topcensorship}
	Let $M$ be a globally hyperbolic spacetime which is asymptotically flat and satisfies the ANEC. Let $\gamma$ and $\gamma_0$ be causal curves from $\mathscr{I}^-$ to $\mathscr{I}^+$, such that $\gamma_0$ is timelike and contained in a simply connected neighbourhood of $\mathscr{I}$. Then $\gamma$ can be continuously deformed to $\gamma_0$ while keeping its endpoints on $\mathscr{I}$. (Here $\mathscr{I}^-$ is past null infinity and $\mathscr{I}=\mathscr{I}^-\cup\mathscr{I}^+$.)
\end{theorem}

The Topological Censorship Theorem does allow observers to probe the topology using signals that originate from, e.g., a past singularity rather than $\mathscr{I}^-$. However, in view of the cosmic censorship conjecture this scenario does not apply to black holes.

\medskip

In the presence of quantum fields black holes will lose energy through Hawking radiation and they are generally expected to evaporate \cite{hawking1975,waldQFT}. This process involves a flux of negative energy across the event horizon, which will then shrink in size. In this process the pointwise energy conditions are certainly violated (c.f.~\cite{PhysRevD.100.124025} for NEC violation at spherically symmetric apparent horizons). The expectation that black holes evaporate also entails the violation of the Area Theorem in semiclassical gravity. The damped ANEC on half geodesics used in \cite{Lesourd:2017} should then also be violated, but to our knowledge this has not yet been verified in detail.

\subsection{Restrictions on exotic spacetimes}
\label{sub:wormholes}

So-called \textit{exotic spacetimes}, such as those with closed timelike curves or superluminal travel, can be easily constructed in the context of general relativity. The basic idea is to solve the Einstein equation in reverse: construct any geometry of spacetime and then find the matter responsible for that geometry. Closed timelike curves and ``effective" superluminal speeds (``effective" as the speed of light cannot be surpassed locally) can be achieved by creating \textit{wormholes}, bridges that connect different parts of spacetime, or other kinds of warp drive spacetimes. (We refer to \cite{lobo2017wormholes,visser1996lorentzian} for scientific reviews of wormhole physics and to \cite{everett2012time,thorne1995black} for generally accessible accounts incorporating historical information.

Wormhole geometry ideas can be traced back to as early as 1916 with the work of Flamm \cite{flamm1916beitrage, flamm2015republication}, but the first wormhole solution was constructed by Einstein and Rosen in 1935 \cite{Einstein:1935tc}. Their motivation was not to create a wormhole, but rather to construct a singularity-free particle model. They considered the isometric embedding of the equatorial section of the Schwarzschild solution to a 3-dimensional Euclidean space with two flat sheets and the ``particle'' was represented as a bridge between these sheets. The Einstein-Rosen wormhole is not what we now call a ``traversable" wormhole: an observer attempting to cross the ``bridge'' will end up on the Schwarzschild curvature singularity. 

Twenty years passed until Wheeler \cite{Wheeler:1955zz} presented the next contribution to the field. His original motivation was not spacetime bridges; his ``geons'' are unstable solutions to the coupled Einstein-Maxwell field equations. However, in his work appears a notion of what we today call a wormhole: ``\dots a metric which on the whole is nearly flat except in two widely separated regions, where a double-connectedness comes into evidence \dots'' \cite{Wheeler:1955zz}, followed by a figure (Figure \ref{fig:wormhole}). 

\begin{figure}[h]
	\center
	\includegraphics[scale=0.35]{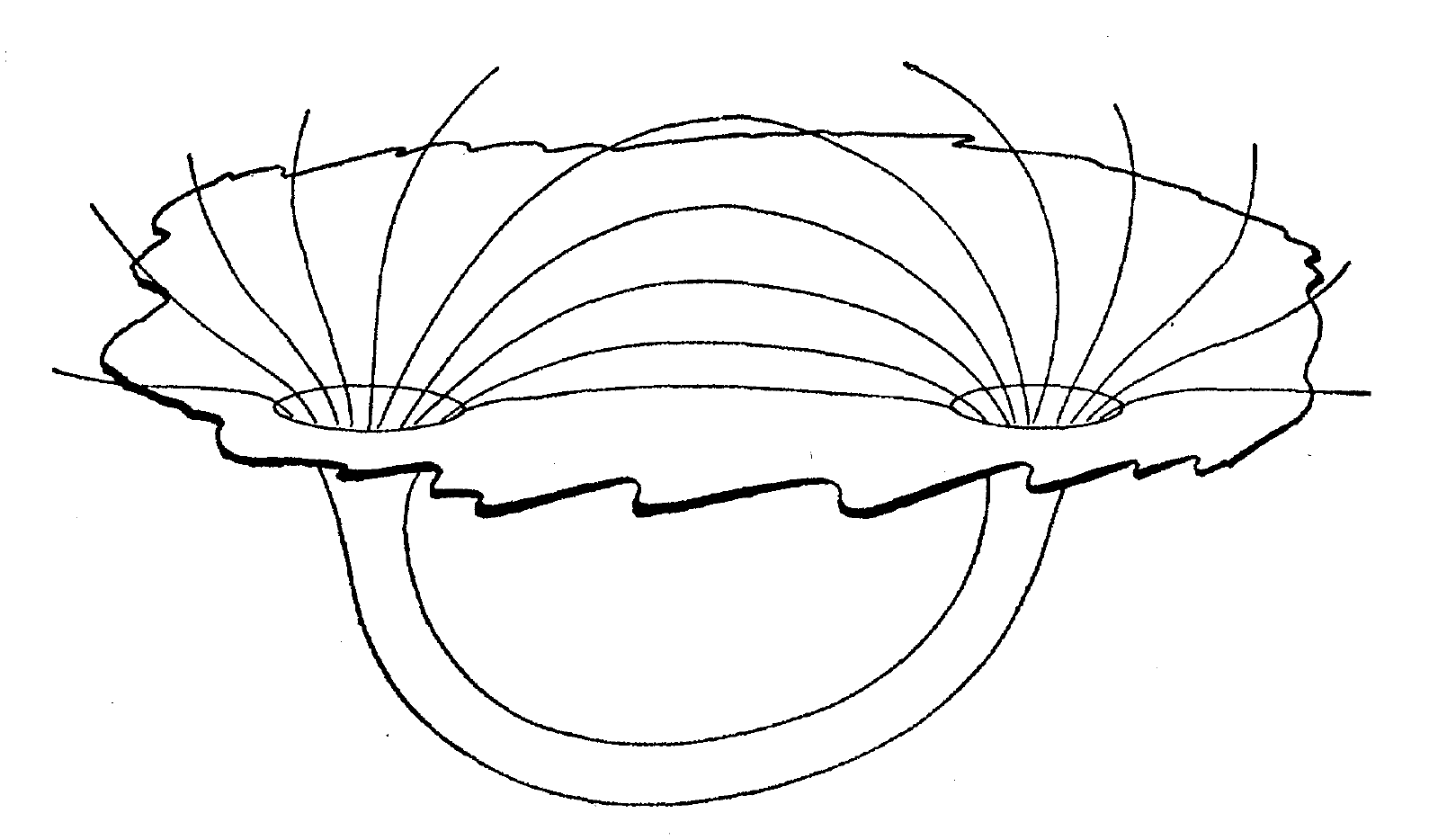}
	\caption{The first schematic representation of a wormhole by Wheeler \cite{Wheeler:1955zz}. (Reproduced with kind permission of the American Physical Society.)}
	\label{fig:wormhole}
\end{figure}

To create a traversable wormhole, a topological ``shortcut'' in spacetime, with finite curvature everywhere, closed timelike curves or a discontinuous choice of the future light cone is inevitable as proven by Geroch \cite{geroch1967topology} and slightly modified by Hawking \cite{hawking1992chronology}. Formally

\begin{theorem}
	If there is a timelike tube connecting surfaces $S$ and $S'$ of different topology, then the region  contains closed timelike curves.
\end{theorem}

Morris, Thorne, and Yurtsever \cite{Morris:1988tu} were the first to provide a concrete mechanism for the creation of a time machine using a wormhole. Their idea is to accelerate one of the mouths of the wormhole so that it is ``younger'' than the other, a variation of the twin effect. Then an observer can travel from that mouth to the other one by conventional means, jump into the wormhole, and arrive before they left! 

It was the late 70's when researchers realized that the kind of matter required to create and support such geometries is one that violates energy conditions. Tipler \cite{Tipler:1977eb} proved that the creation of a wormhole required the violation of the WEC. He provided a variation of Geroch's theorem where he replaced causality conditions by the WEC and Einstein's equation. In particular Tipler showed that topology changes cannot happen in a compact region if the WEC holds; and if topology changes happen in a finite (but not compact) region, the change must be accompanied with singularities. Keeping a wormhole ``open'' means looking through the wormhole at a two-sphere surrounding one mouth and far from it there is an outer trapped surface. Since there is no event horizon, the wormhole's stress tensor must also violate the WEC (see \cite{HawkingEllis:1973} Prop.~2.9.8). 

\medskip

Warp drive spacetimes are similar to wormhole geometries in that they allow effective superluminal travel and causality violations, but they do not require the change of topology. The first warp drive was introduced by Alcubierre \cite{Alcubierre:1994tu}. It describes a spaceship in a warp bubble characterized by a local expansion of space behind the bubble, and an opposite contraction ahead of it. The idea is similar to that of an inflationary universe: the relative speed of separation of two co-moving observers defined as the rate of change of proper spatial distance over proper time, can be much greater than the speed of light. However, the principle of relativity is not violated as the observers are always moving inside their local light cones. Similarly the spaceship in the warp bubble is pushed away from one destination and pulled towards another by the distortion of spacetime. The absence of wormhole topologies in the warp drive geometry does not mean it is free of energy condition violations. In fact, in the same work Alcubierre \cite{Alcubierre:1994tu} showed that his warp drive violates the WEC. The Alcubierre bubble also necessarily involves causality violations as was shown by Everett \cite{Everett:1995nn}.

The pilot of the spaceship in an Alcubierre bubble cannot control it in any way as the outside front of the bubble is always spacelike separated from its centre. Inspired by this problem Krasnikov suggested a different way to shorten travel between distant stars \cite{Krasnikov:1995ad}. He suggested that a spaceship could first travel to a destination at subluminal speed while modifying the spacetime in a tubular neighbourhood around it in the process. The modification consists of ``stretching'' the past light cones, a causal process. Then the trip back through the tube could take an arbitrarily short time, making the round trip more efficient.  Soon after Krasnikov's original work, Everett and Roman \cite{Everett:1997hb} showed that even though a ``Krasnikov tube'' does not include closed timelike curves, one can construct a time machine using two of these tubes. Additionally they showed that, as with the Alcubierre bubble, its construction requires negative energy. 

A question is then whether it is possible to have a warp drive without violation of the WEC. Olum \cite{Olum:1998mu} proved that it is not (see also \cite{penrose1993positive} for related ideas). First he defined superluminal travel as a situation where null geodesics are emitted perpendicularly to a flat 2-surface but one of them arrives at a destination flat 2-surface earlier (in an externally defined time) than all other geodesics. That geodesic would be for example the path in a Krasnikov tube. Formally
\begin{definition}
	\label{def:super}
	A causal path $P$ is superluminal from $A$ to $B$ only if there exist 2-surfaces $\Sigma_A$ around $A$ and $\Sigma_B$ around $B$ such that
	\begin{enumerate}
		\item 
		if $p \in \Sigma_A$ then a spacelike geodesic lying in $\Sigma_A$ connects $A$ to $p$, and similarly for
		$\Sigma_B$, and
		\item
		if $p\in \Sigma_A $ and $q \in \Sigma_B$ then $q$ is in the causal future of $p$ only if $p=A$ and $q=B$.
	\end{enumerate}
\end{definition}
The proof that warp drive spacetimes violate the WEC uses the Raychaudhuri equation (\ref{eqn:raynull}) for null geodesics. For a spacetime that obeys the WEC and the generic condition, which excludes some special metrics where the curvature vanishes in particular directions (see comments on Prop.~4.4.5 \cite{HawkingEllis:1973}), we have $\theta<0$. But following Definition \ref{def:super} of superluminal travel the null congruence at $B$ diverges thus $\theta \geq 0$. This contradiction proves the WEC violation. 

\medskip

We have seen that wormhole spacetimes as well as those with closed timelike curves and warp drives all require violations of the classical pointwise energy conditions. But we have also seen (Sec.~\ref{sec:QEI}) that such conditions are violated by quantum fields. Could it be possible that exotic spacetimes exist in a semiclassical context? The work by Morris, Thorne, and Yurtsever \cite{Morris:1988tu} is considered seminal in wormhole physics for describing such a possibility. It was the first to provide a wormhole model with a source for negative energy, in this case a Casimir plate system (see Sec.~\ref{subs:AANEC}) that violates the WEC and thus provides the necessary negative energy to keep the wormhole open. However, even before that work there were arguments that wormhole spacetimes also require violation of averaged energy conditions and especially the AWEC \cite{Deutsch:1978sc}. Later, Hawking \cite{hawking1992chronology}, proved more generally that if the causality violation developed from a non-compact initial surface, the AWEC must be violated on the Cauchy horizon. But even those theorems allow the possibility of constructing exotic spacetimes using AWEC violating quantum fields. 

Ford and Roman~\cite{Ford:1995wg} used QEIs to provide restrictions to wormhole spacetimes. They showed that static wormholes are possible semiclassically only if their throat is close to Planck size or the negative energy must be concentrated in bands many orders of magnitude smaller than the throat size. Work by Everett and Roman~\cite{Everett:1997hb} and also Ford and Pfenning~\cite{Pfenning:1997wh} showed similar results for Krasnikov tubes and Alcubierre bubbles respectively. For a Krasnikov tube, the negative energy needs to be concentrated in a region no thicker than a few thousand Planck lengths. Additionally, the total negative energy required for a tube of length and diameter of $1$m is of order $10^{16}$ galactic masses! Similarly, for an Alcubierre bubble the thickness is constrained to be of the order of only a few hundred Planck lengths, while the negative energy required for a macroscopic bubble is greater than the total mass of the entire visible universe.

Even though these results significantly constrain exotic spacetimes they leave open possibilities, for example the creation of atomic size warp bubbles \cite{Pfenning:1997wh} or acquiring the necessary amount of negative energy by multiple quantum fields \cite{Ford:1995wg}.  Visser, Kar, and Dadhich  \cite{Visser:2003yf, Kar:2004hc} quantified the minimum required violations of averaged energy conditions and discussed the possibility of creating such violations using quantum fields. Since there are counter-examples even to ANEC in semiclassical gravity (see Sec.~\ref{subs:ANEC}) certain possibilities remained open. 

But the creation of causality violating and faster-than travel spacetimes was proven to be even more difficult. Work by Graham and Olum \cite{Graham:2007va} generalized Hawking's theorem \cite{hawking1992chronology} and showed that such spacetimes require the violation not only of ANEC but also the self-consistent AANEC.  

\begin{theorem}\label{thm:nowormholes}
	Let $M$ be an asymptotically flat, globally hyperbolic spacetime satisfying on-shell AANEC and the generic condition, with a partial Cauchy surface $S$. Then $M$ cannot have a compactly generated Cauchy horizon $H_+(S)$.
\end{theorem}

By Olum's earlier proof \cite{Olum:1998mu}, this theorem also prohibits warp drive spacetimes. There is evidence (see Sec.~\ref{subs:AANEC}) that AANEC is generally obeyed in semiclassical gravity making this theorem a powerful argument against the existence of exotic spacetimes. 

So what possibilities remain? The main proposals for time machines that do not contradict Graham and Olum's theorem involve spacetimes that are not asymptotically flat. Perhaps the simplest case is a cylindrical universe or the little more complicated G\"odel universe \cite{godel1949example}. Of course such possibilities are easily ruled out by observations and there is no hope of an advanced civilization ever constructing them. A similar idea is the Mallett time machine \cite{mallett2003gravitational} which utilizes an infinite dust cylinder. As long as the cylinder is infinite the spacetime is not asymptotically flat and thus it evades Hawking's or Graham and Olum's theorems. But such a cylinder requires infinite energy and so it is impossible to construct. Mallett's proposal for a finite cylinder contradicts existing theorems as was explicitly shown by Everett and Olum \cite{Olum:2004kz}. To evade the problem of constructing an infinite cylinder, Gott \cite{Gott:1990zr} suggested the existence of a time machine made of infinite cosmic strings, very thin, but massive objects produced by some theories of early universe phase transitions. So far there are only observational constraints on the existence of cosmic strings. We should note that even if infinite cosmic strings exist, they cannot be manipulated so Gott's time machine can only be discovered.\footnote{We thank Ken Olum for this observation.}

Ori investigated a different possibility \cite{Ori:2007kk}. He suggested a time machine in an asymptotically flat spacetime without the presence of negative energy. It evades the previous theorems, because the horizon in his model is not compactly generated. Criticism of the model \cite{Eveson:2007ns} involves the unavoidable presence of naked singularities or the enlargement of an initially finite region to form what is called an ``internal infinity''.

Hawking in his influential 1992 work suggested that all the constructions involving causality violations are impossible \cite{hawking1992chronology}. His proposal is known as the ``Chronology Protection Conjecture'':
\begin{conjecture}
	\label{conj:chronologyprotection}
	The laws of physics do not allow the appearance of closed timelike curves.
\end{conjecture} 
While evidence is so far in favour of the conjecture, due to its generality, there is little hope it will be proven in the near future or at least until we have a complete theory of quantum gravity. 

Finally we should mention that the so-called ``long'' wormholes, lie outside of the scope of the Graham-Olum theorem. They are called ``long'' because it takes longer to go through the wormhole than through the ambient space and so no causality violations occur. Such a wormhole was recently proposed by Maldacena, Milekhin and Popov \cite{Maldacena:2018gjk}. Their model is a solution of an Einstein-Maxwell theory with charged massless fermions which give rise to negative energy. Since the wormhole is ``long'' the null geodesics passing through it are chronal thus not requiring the violation of AANEC.

\subsection{Cosmological applications}
\label{sub:cosmology}

In this section we investigate some cosmological consequences of the energy conditions. We have already seen Hawking's singularity theorem in Section \ref{sub:singularity}, which is relevant for the big bang, but there are also results that concern the future of the universe.

The most basic cosmological models assume that the universe is homogeneous and isotropic, leading to a
Friedmann-Lema\^{i}tre-Robertson-Walker (FLRW) spacetime with a metric of the form
\be
g=-\mathrm{d}\tau^2+a(\tau)^2h\,,\nonumber
\ee
where $h$ is a homogeneous and isotropic spatial metric. The scale factor $a(\tau)$ can vary with time and is related to the matter content by Einstein's equation (\ref{eqn:EE}). For the purposes of this section we will study Einstein's equation with a cosmological constant $\Lambda>0$, i.e.
\be
\label{eqn:EELambda}
G_{ab}+\Lambda g_{ab}=8\pi T_{ab}\,.
\ee
The cosmological constant is a simple and convenient explanation for the accelerated expansion of the universe, which is astronomically observed. Indeed, in the absence of matter, but with $\Lambda>0$, we find the de Sitter universe, which is a FLRW spacetime with
\be
a(\tau)=e^{H\tau}\nonumber
\ee
for the constant Hubble expansion rate
\be
\label{def:Hubble}
H=\sqrt{\frac13\Lambda}\,.
\ee

Although the FLRW model gives a fairly accurate description of the large scale structure of our universe, it does not explain how the universe became so homogeneous and isotropic. This homogeneity and isotropy even extend across regions of the universe which are causally disconnected. This surprising observational fact is called the \emph{horizon problem} and it is generally accepted that it requires an explanation. It is unsatisfactory to merely assume that our universe was described at very early times by very special homogeneous and isotropic initial conditions. A much more promising explanation, which is generally accepted, is that any initial inhomogeneities and anisotropies were diluted by inflation. The expectation that a generic universe with $\Lambda>0$ and reasonable matter will locally approach the de Sitter universe at late times is called the \emph{Cosmic No-Hair Conjecture}
(cf.~\cite{Andreasson2016} Conjecture 11 for a mathematically precise formulation). The characterization of reasonable matter in this statement typically involves the energy conditions.

In a short and very elegant paper, Wald \cite{Wald:1983ky} proved the Cosmic No-Hair Conjecture under the assumptions that the SEC and the DEC hold and that the universe is homogeneous with a positive initial expansion and a negative spatial scalar curvature (which excludes the Bianchi type IX models). The argument focuses on $n=4$ and notes that the Raychaudhuri equation (\ref{eqn:raytime}) for the congruence of comoving observers can be written as
\be
\frac{\mathrm{d} \theta}{\mathrm{d}\tau} =\Lambda - 8\pi\left(T_{ab}-\frac{T}{2}g_{ab} \right)t^at^b-\sigma_{ab} \sigma^{ab} -\frac13\theta^2
\,,\nonumber
\ee
where we noted that $\omega_{ab}=0$ and we eliminated $R_{ab}$ using Einstein's equation (\ref{eqn:EELambda}). Using the SEC and $\sigma_{ab} \sigma^{ab}\ge0$ this leads to the estimate
\be
\label{ineq:ray}
\frac{\mathrm{d} \theta}{\mathrm{d}\tau} \le \Lambda - \frac13\theta^2\,.
\ee
The initial value constraint $G_{ab}t^at^b-\Lambda=8\pi T_{ab}t^at^b$ can be written as
\be
\label{eq:constraint}
\theta^2=3\Lambda+\frac32\sigma_{ab}\sigma^{ab}-\frac32\ ^{(3)}R+24\pi T_{ab}t^at^b\,.
\ee
When the WEC holds and since the spatial curvature is $^{(3)}R\le0$ this leads to $\theta^2\ge3\Lambda$. If $0<\theta<\sqrt{3\Lambda}$ at some initial time, then one can integrate the differential inequality (\ref{ineq:ray}) to find
\be
\theta(\tau)\le \frac{\sqrt{3\Lambda}}{\tanh(H(\tau-\tau_0))}\nonumber
\ee
for some integration constant $\tau_0$. This shows that $\theta$ converges exponentially to the constant value $\sqrt{3\Lambda}$.
Similarly we see from (\ref{eq:constraint}) that $0\le T_{ab}t^at^b\le \theta^2-3\Lambda$, so $T_{ab}t^at^b\to 0$ exponentially fast. The exponential convergence $T_{ab}\to 0$ of the other components of the stress tensor then follows from the DEC.

The result of \cite{Wald:1983ky} is sometimes known as the \emph{Cosmic No-Hair Theorem}. For a discussion of more recent work in this direction we refer to \cite{Andreasson2016}. Although the elegant argument of \cite{Wald:1983ky} explains how initial anisotropy can be diluted, it does not explain where homogeneity comes from. Moreover, under the assumptions of \cite{Wald:1983ky} the universe cannot exit a period of inflation, which is not in line with current cosmological models. For more realistic inflation models the assumptions of SEC and DEC (and of spatial homogeneity) would have to be relaxed. In a nice investigation, Maleknejad and Sheikh-Jabbari \cite{PhysRevD.85.123508} considered inflationary models satisfying the SEC and the WEC and with a time-dependent positive cosmological ``constant'' $\Lambda$. Under these circumstances they found that the anisotropy of the universe can increase, but it must remain bounded.

\medskip

To conclude we mention a rather different result on the future fate of a spacetime, namely the \emph{Lorentzian Splitting Theorem} \cite{galloway1989lorentzian,eschenburg1988splitting}. This mathematical result provides an interesting characterization of standard ultrastatic spacetimes (cf.~\cite{SANCHEZ2005e455}). These spacetimes can be split into space, given by a Riemannian manifold $(\Sigma,h)$, and time, given by the real line $(\mathbb{R},-\mathrm{d}\tau^2)$, which combine to a spacetime
$M=(\mathbb{R}\times\Sigma,g)$ with metric
\be
g=-\mathrm{d}\tau^2+h\,.\nonumber
\ee
Note that $h$ is independent of the time $\tau$ and that $\chi:=\partial_{\tau}$ is a complete timelike Killing vector field which is irrotational. $M$ is globally hyperbolic if and only if $(\Sigma,h)$ is complete.

In a standard ultrastatic spacetime the energy conditions take a particularly simple form, because the Ricci curvature $R_{ab}$ is purely spatial, i.e.~$R_{ab}\chi^b=0$. The WEC, SEC, DEC and NEC all reduce to $R_{ab}x^ax^b\ge0$ for all vectors $x^a$ tangent to $\Sigma$. For any fixed $x\in\Sigma$ the \emph{static curve} $\gamma_x(\tau):=(\tau,x)$ is a complete timelike geodesic. Moreover, each $\gamma_x$ maximizes the length between every pair of points on this geodesic. Such a geodesic is called a \emph{complete timelike line}. The existence of such complete timelike lines is surprisingly special, as the following Lorentzian Spliting Theorem shows \cite{galloway1989lorentzian,eschenburg1988splitting}:

\begin{theorem}
	If $M$ is a connected globally hyperbolic spacetime which satisfies the SEC and which contains a complete timelike line $\gamma$, then $M$ is isometric to a standard ultrastatic spacetime in which $\gamma$ is a static curve.
\end{theorem}

\section{Outlook}
\label{sec:outlook}

Since their invention, the energy conditions of general relativity have undergone various developments, which were often driven by the tension between their different purposes. The desire for strong consequences seems incompatible with the desire for as general validity as possible. Challenges from QFT, where all pointwise energy conditions are necessarily violated, have led to a variety of QEIs and averaged energy conditions, which complicate the picture. Because even the averaged energy conditions seem to have relatively innocent looking counter-examples, some authors have already sounded the death knell for the energy conditions and criticized them as ad hoc assumptions to prove theorems \cite{Barcelo:2002bv}. Moreover, the averaged conditions have been criticized for their lack of a clear interpretation \cite{Curiel:2017}, whereas the lower bounds in QEIs seem to be theory-dependent, so they do not characterize normal quantum matter with a single general condition.

Based on our review we would like to draw a slightly more optimistic balance, but we do see a need for further research and clarification at several points.

\medskip

Let us start with QEIs along timelike geodesics. Recall from Section \ref{subs:uncertainty} that QWEIs in static spacetimes have been connected to global stability properties and the thermodynamics of quantum fields \cite{FewsterVerch2003}, which gives us a conceptually clear reason to consider them. Even though QEIs have not yet led to a clearly stated general condition, this does not mean that such a condition cannot be given. If an appropriate (renormalized) stress tensor is given, one can certainly postulate the existence of a lower bound for suitable averages. Although the lower bound may be theory-dependent (e.g.~for $n\ge1$ free scalar fields the bound scales with $n$), some of its characteristics, e.g.~the behaviour under scaling of the averaging function, might be formulated quite generally. In this sense, QEIs could well be formulated as a general condition, even for interacting fields. Once an appropriate formulation along these lines has been achieved one is of course still left with the tension between the range of validity and the desire for strong geometric results.

The violations of the pointwise and averaged energy conditions show that their range of validity is limited. This means in particular that some of their consequences, such as the singularity theorems (cf.~Sec.~\ref{sub:singularity}), are also restricted in their validity. However, this does not immediately make the energy conditions useless. Indeed, many energy conditions do cover physically interesting examples and for any given theory it is often possible to find out whether it satisfies a given condition (certainly for pointwise conditions). In this way the limited range of validity can often be taken into account. Moreover, we have seen strong indications that the range of validity of the energy conditions is not random, but it is related to the validity of suitable QEIs, cf.~Section \ref{subs:claslimits}. At least for free fields the validity of a state-independent QEI implies the validity of the corresponding pointwise energy condition, even off-shell, cf.~Theorem \ref{thm:classlimit}. We also expect the converse to be true. Of course there are fields, like non-minimally coupled scalar fields, which do not satisfy pointwise energy conditions (cf.~Sec.~\ref{sub:examples}). However, for these fields one can still often establish state-dependent QEIs (cf.~the comments in Sec.~\ref{sub:examplesQEI}). In suitable situations these can then be used to derive averaged energy conditions for the quantized or classical field (cf.~Sec.~\ref{subs:AANEC} and Sec.~\ref{subs:ANEC} for the classical non-minimally coupled scalar field).

Determining the range of validity of the various energy conditions and QEIs is still an active research topic. One of the main open challenges is to consider interacting quantum fields, which should be formulated in general spacetimes in order to derive the most interesting results. The rigorous formulation of such theories and the axiomatic characterization of their stress tensor still poses many challenging questions, whereas a perturbative approach introduces renormalization ambiguities. However, also free fields give rise to interesting open questions, e.g.~to what extent the range of validity of the AWEC characterizes isolated systems (in a suitable sense of the word), cf.~the discussion in Section \ref{subs:uncertainty}. Because the exact computation of the renormalized stress tensor is often cumbersome, numerical methods may be of great help to investigate their properties and we refer the reader to \cite{Levi:2016quh,PhysRevLett.118.141102,PhysRevD.99.061502} for some recent results in black hole spacetimes.

The above considerations do not immediately apply to the null energy conditions, because the appropriate starting point would be a QEI along null geodesics, which often fails, cf.~\cite{Fewster:2002ne} and the discussion in Section \ref{sub:examplesQEI}. This is curious, because the AANEC is the only condition which may well be valid in full generality on-shell. Starting with classical scalar fields we have seen in Section \ref{subs:ANEC} that the ANEC holds on-shell for any coupling $\xi$, if we exclude trans-Planckian values. Trans-Planckian values should be viewed with scepticism for several reasons, relating to the possible ill-posedness of the initial value formulation or to the change of sign in the effective Newton's constant. Even for quantum fields, in the semiclassical approach and using a transversal smearing over Planck scale distances, there are no known counter-examples to the on-shell AANEC (cf.~Sec.~\ref{subs:AANEC}). The condition has also been shown to hold for test fields in Minkowski space in many cases, including in particular various interacting quantum fields in two dimensions (including massive theories and CFTs) and there has been recent work relating it to developments in black hole thermodynamics and quantum gravity. However, its general validity, and the underlying ideas and techniques, still require further investigation.

If we believe that the AANEC is generally valid in curved spacetimes, then we may hope that there are deeper reasons for it, which lie in the realm of quantum gravity. Note, however, that this is not the only way in which the AANEC can teach us something about quantum gravity. Several of the applications discussed in Section \ref{sec:applications} only assume the AANEC, so they can be expected to remain valid in quantum gravity (although their geometric formulation possibly needs to be adapted). This is true in particular for the Topological Censorship Theorem \ref{thm:topcensorship} of Section \ref{sub:blackholes} and for the strong restrictions on wormholes and warp drives (Thm.~\ref{thm:nowormholes} in Sec.~\ref{sub:wormholes}). It has also been conjectured that causality violations cannot appear in physics, cf.~Conjecture \ref{conj:chronologyprotection}, and in semiclassical gravity this likely follows from the AANEC. The singularity theorems of Section \ref{sub:singularity}, however, typically require stronger conditions than the AANEC, so they may be circumvented in quantum gravity and singularities might not occur. 
This is particularly interesting in the context of the Cosmic Censorship Conjecture and black hole evaporation.

\medskip

The connections between the original energy conditions for classical matter and the recent developments in QEIs and black hole thermodynamics clearly illustrate the depth of this topic and we expect that the energy conditions will remain a very active area of research. The deep question whether the AANEC is a truly fundamental energy condition that remains valid in quantum gravity and the search for a deeper explanation will pose an important challenge for the future. The investigation of the consequences of the AANEC may very well turn out to be an important path towards a better understanding of quantum gravity.

\bigskip

	\noindent{\bf Acknowledgements}\\
	It is our pleasure to thank Chris Fewster, Erik Curiel and Ken Olum for stimulating discussions and helpful suggestions. We thank 
	an anonymous referee for carefully checking the manuscript and for suggesting references on the geometric interpretation of the WEC. KS also thanks the GR group at Dublin City University for a useful discussion. EAK's contribution to this work is part of a project that has received funding from the European Union's Horizon 2020 research and innovation programme under the Marie Sk\l odowska-Curie grant agreement No. 744037 ``QuEST''.

\bibliographystyle{unsrt}
\bibliography{review}

\end{document}